\RequirePackage{silence}
\documentclass[aps,prx,reprint,superscriptaddress,nofootinbib,a4paper,10pt,showpacs,floatfix,longbibliography]{revtex4-2}

\pdfoutput=1

\usepackage[T1]{fontenc} 
\usepackage[english]{babel} 
\usepackage{microtype} 

\providecommand{\pra}{Phys.\ Rev.\ A}

\providecommand{\prl}{Phys.\ Rev.\ Lett.}

\usepackage[table,dvipsnames]{xcolor}
\usepackage[colorlinks=true,
hypertexnames=false
]{hyperref}
\PassOptionsToPackage{linktocpage}{hyperref} 

\usepackage{amsmath,amssymb,mathtools,amsthm,dsfont}
\usepackage{easybmat} 
\usepackage[capitalise]{cleveref} 

\usepackage{booktabs}
\usepackage{multirow}

\usepackage{paralist}
\usepackage{graphicx}
\usepackage{subcaption}
\usepackage{enumitem}

\usepackage{xparse}
\usepackage{ytableau}
\usepackage[printonlyused,withpage,nolist]{acronym}
\usepackage{algpseudocode}
\usepackage{algorithm}

\usepackage{tikz}
\usetikzlibrary{decorations.pathreplacing,
	calligraphy,
	matrix}

\tikzset{
	BC/.style = {decorate,
		decoration={calligraphic brace, amplitude=5pt, raise=1mm},
		very thick, pen colour={black}
	},
}

\newtheorem{theorem}{Theorem}
\newtheorem{corollary}[theorem]{Corollary}
\newtheorem{lemma}[theorem]{Lemma}
\newtheorem{proposition}[theorem]{Proposition}

\theoremstyle{definition}
\newtheorem{example}[theorem]{Example}

\DeclareMathOperator{\LandauO}{\mathrm{O}}

\DeclareMathOperator{\Tr}{Tr}

\DeclareMathOperator{\Ad}{Ad}

\DeclareMathOperator{\diag}{diag}
\newcommand{\fro}{\mathrm{F}}

\newcommand{\bounded}[1]{\mathcal{B}(#1)} 

\DeclareMathOperator{\Var}{Var}

\DeclareMathOperator{\U}{U}
\DeclareMathOperator{\SU}{SU}
\DeclareMathOperator{\GL}{GL}

\newcommand{\CC}{\mathbb{C}}
\newcommand{\RR}{\mathbb{R}}

\newcommand{\ZZ}{\mathbb{Z}}
\newcommand{\NN}{\mathbb{N}}
\newcommand{\1}{\mathds{1}}
\newcommand{\EE}{\mathbb{E}}

\newcommand{\R}{\mathbb{R}}
\newcommand{\C}{\mathbb{C}}

\newcommand{\mc}[1]{\mathcal{#1}}

\newcommand{\mcH}{\mc{H}}
\renewcommand{\H}{\mcH}

\newcommand{\argdot}{{\,\cdot\,}}
\renewcommand{\vec}[1]{\boldsymbol{#1}}

\newcommand{\ad}{\dagger}

\DeclarePairedDelimiterX{\abs}[1]{\lvert}{\rvert}{%
  \ifblank{#1}{\,\cdot\,}{#1}
}   

\DeclarePairedDelimiterX\norm[1]\lVert\rVert{%
  \ifblank{#1}{\,\cdot\,}{#1}
}   

\DeclarePairedDelimiterX{\iiiNorm}[1]{\lvert}{\rvert}{%
  \delimsize\lvert\delimsize\lvert#1\delimsize\rvert\delimsize\rvert%
}

\DeclarePairedDelimiterXPP\snorm[1]{}\lVert\rVert{_\infty}{\ifblank{#1}{\,\cdot\,}{#1}}   

\DeclarePairedDelimiterXPP\twonorm[1]{}\lVert\rVert{_2}{\ifblank{#1}{\,\cdot\,}{#1}}   

\DeclarePairedDelimiterXPP\trnorm[1]{}\lVert\rVert{_1}{\ifblank{#1}{\,\cdot\,}{#1}}   

\DeclarePairedDelimiterXPP\fnorm[1]{}\lVert\rVert{_{\fro}}{\ifblank{#1}{\,\cdot\,}{#1}}   

\DeclarePairedDelimiterXPP\dnorm[1]{}\lVert\rVert{_\diamond}{\ifblank{#1}{\,\cdot\,}{#1}}   

\DeclarePairedDelimiterXPP\cbnorm[1]{}\lVert\rVert{_\mathrm{cb}}{\ifblank{#1}{\,\cdot\,}{#1}}   
\DeclarePairedDelimiterXPP\onenorm[1]{}\lVert\rVert{_{1\rightarrow 1}}{\ifblank{#1}{\,\cdot\,}{#1}}   
\DeclarePairedDelimiterXPP\ddnorm[1]{}\lVert\rVert{_{\diamond\rightarrow \diamond}}{\ifblank{#1}{\,\cdot\,}{#1}}   
\DeclarePairedDelimiterXPP\ssnorm[1]{}\lVert\rVert{_{\infty\rightarrow\infty}}{\ifblank{#1}{\,\cdot\,}{#1}}   

\DeclarePairedDelimiterX\Set[1]\{\}{%
  
  #1
}

\DeclarePairedDelimiterX\innerp[2]{\langle}{\rangle}{%
  \ifblank{#1}{\,\cdot\,}{#1} , \ifblank{#2}{\,\cdot\,}{#2}%
}

\DeclarePairedDelimiter{\ket}{\vert}{\rangle}
\newcommand{\ketb}[1]{\ket[\big]{#1}}

\DeclarePairedDelimiterX\braket[2]{\langle}{\rangle}%
  {#1\kern0.15ex\delimsize\vert\kern0.15ex\mathopen{}#2}

\DeclarePairedDelimiterX\ketbra[2]{\vert}{\vert}%
  {#1\kern0.15ex\delimsize\rangle\delimsize\langle\kern0.15ex\mathopen{}#2}

\DeclarePairedDelimiterX\sandwich[3]{\langle}{\rangle}%
  {#1\,\delimsize\vert\kern0.15ex\mathopen{}#2\kern0.15ex\delimsize\vert\kern0.15ex\mathopen{}#3}

\DeclarePairedDelimiterX\obraket[2]{(}{)}%
  {#1\kern0.15ex\delimsize\vert\kern0.15ex\mathopen{}#2}

\DeclarePairedDelimiterX\oketbra[2]{\vert}{\vert}%
  {#1\kern0.15ex\delimsize)\delimsize(\kern0.15ex\mathopen{}#2}

\DeclarePairedDelimiterX\osandwich[3]{(}{)}%
  {#1\,\delimsize\vert\kern0.15ex\mathopen{}#2\kern0.15ex\delimsize\vert\kern0.15ex\mathopen{}#3}

\newcommand{\myleft}{\mathopen{}\mathclose\bgroup\left}
\newcommand{\myright}{\aftergroup\egroup\right}

\newcommand{\nmodes}{m} 
\newcommand{\ninput}{\vec{n}_0} 
\newcommand{\ninputdual}{\bar{\vec{n}}_0} 
\newcommand{\ninputGT}{N_0} 
\newcommand{\ninputGTdual}{\bar N_0} 
\newcommand{\seqlength}{l}
\newcommand{\sequences}{\mathbb{L}}
\newcommand{\numsamples}{\mathbb{T}}

\NewDocumentCommand\Sp{mg}{ 
	\ensuremath{\mathrm{Sp}\IfNoValueTF{#1}{}{(#1)}}%
}

\newcommand{\fock}{\mathcal F} 
\newcommand{\GT}[1]{\mathrm{GT}({#1})} 
\newcommand{\ssyt}[1]{\mathrm{SSYT}({#1})} 
\newcommand{\ytzero}[1]{\mathrm{SSYT}^{(\vec 0)}({#1})} 
\newcommand{\cg}{C} 
\DeclareMathOperator{\per}{\mathrm{Per}} 

\hypersetup{
	pdfauthor = {Mirko Arienzo, Dmitry Grinko, Martin Kliesch, Markus Heinrich}, 
	pdfsubject = {Quantum science},
	pdfkeywords = {%
		quantum, information, computing, randomized benchmarking, characterization, continuous variables, boson, bosonic systems, Gaussian, passive, symplectic
	}
}

\newcommand{\hhu}{%
  Heinrich Heine University D{\"u}sseldorf, 
  Faculty of Mathematics and Natural Sciences,
  D{\"u}sseldorf, 
  Germany
}

\newcommand{\tuhh}{%
	Hamburg University of Technology,
    Institute for Quantum Inspired and Quantum Optimization, 
    Blohmstraße 15,
    Hamburg,
    Germany
}
\newcommand{\qusoft}{%
    QuSoft,
    Amsterdam, The Netherlands
}
\newcommand{\uva}{%
	Institute for Logic, Language and Computation,
    University of Amsterdam,
    Amsterdam, The Netherlands
}

\newcommand{\col}{%
Institute for Theoretical Physics, University of Cologne, Cologne, Germany.
}

\begin{document}

\title{Bosonic randomized benchmarking with passive transformations}

\author{Mirko Arienzo}
\email{mirko.arienzo@tuhh.de}
\affiliation{\tuhh}
\author{Dmitry Grinko}
\affiliation{\uva}
\affiliation{\qusoft}
\author{Martin Kliesch}
\affiliation{\tuhh}
\author{Markus Heinrich}
\affiliation{\hhu}
\affiliation{\col}

\begin{abstract}
	Randomized benchmarking (RB) is the most commonly employed protocol for the characterization of unitary operations in quantum circuits due to its reasonable experimental requirements and robustness against state preparation and measurement (SPAM) errors.
	So far, the protocol has been limited to discrete or fermionic systems, whereas extensions to bosonic systems have been unclear for a long time due to challenges arising from the underlying bosonic Hilbert space. 
	In this work, we close the gap for bosonic systems and develop an RB protocol to benchmark passive Gaussian transformations on any particle number subspace,
	which we call \emph{passive bosonic RB}.
	The protocol is built on top of
	the recently developed filtered RB framework \cite{helsenGeneralFrameworkRandomized2020,heinrichRandomizedBenchmarkingRandom2023} and is designed to isolate the multitude of exponential decays arising for passive bosonic transformations.
	We give explicit formulas and a Julia implementation for the necessary post-processing of the experimental data.
	We also analyze the sampling complexity of passive bosonic RB by deriving analytical expressions for the variance.
	They show a mild scaling with the number of modes, suggesting that passive bosonic RB is experimentally feasible for a moderate number of modes.
	We focus on experimental settings involving Fock states and particle number resolving measurements, but also discuss Gaussian settings, deriving first results for heterodyne measurements.
\end{abstract}

\maketitle

\section{Introduction}

The characterization of quantum devices, and, in particular, of the involved unitary operations, is a fundamental task in quantum information processing \cite{Kliesch2020TheoryOfQuantum, Eisert2020QuantumCertificationAnd}.
While being a well-studied field for discrete variable systems, a similar standing for \ac{CV} systems has not yet been achieved.
In fact, there are only very few works apart from \ac{CV} tomography and the first rigorous guarantees for the latter were proven only recently \cite{wuQuantumEnhancedLearningContinuousVariable2023,meleLearningQuantumStates2024}.
The required number of measurements is experimentally extremely challenging
\cite{lvovskyContinuousvariableOpticalQuantum2009, lopezProgressScalableTomography2010, schmiegelowSelectiveEfficientQuantum2011, benderskySelectiveEfficientQuantum2013, namikiSchmidtnumberBenchmarksContinuousvariable2016,altepeterAncillaAssistedQuantumProcess2003,baiTestOneTest2018}
and as a result, full quantum tomography is often not feasible for \ac{CV} systems.
Moreover, tomographic protocols typically suffer from \ac{SPAM} errors, as in the discrete variable case.

Bosonic quantum systems play a major role in the design of quantum computing platforms.
Most prominently, this includes photonic quantum computing as a popular proposal for scalable quantum computers \cite{knillSchemeEfficientQuantum2001, koashiProbabilisticManipulationEntangled2001, kokLinearOpticalQuantum2007, rudolph2016i, tanResurgenceLinearOptics2018, alexanderManufacturablePlatformPhotonic2024}.
The biggest advantages of this model rely on the implementation of particle sources, detectors, and linear optical circuits on the same integrated chips \cite{politiSilicaonSiliconWaveguideQuantum2008, sprengersWaveguideSuperconductingSinglephoton2011, silverstoneOnchipQuantumInterference2014, meanyLaserWrittenCircuits2015}, and access to mixed schemes for quantum error correction \cite{bourassaBlueprintScalablePhotonic2021}.
Bosonic systems also offer interesting non-universal models of computation to test quantum supremacy, such as (Gaussian) boson sampling \cite{aaronsonComputationalComplexityLinear2010,hamiltonGaussianBosonSampling2017, chakhmakhchyanBosonSamplingGaussian2017}, 
which have been implemented successfully in real experiments recently \cite{zhongQuantumComputationalAdvantage2020,zhongPhaseProgrammableGaussianBoson2021,madsenQuantumComputationalAdvantage2022}.

Most challenges in the characterization of bosonic systems can be attributed to the particularities of the bosonic Hilbert space \cite{lvovskyContinuousvariableOpticalQuantum2009, sharmaCharacterizingPerformanceContinuousvariable2020}.
For instance, protocols that rely on scrambling techniques via unitary designs are challenging to adapt, since Gaussian unitaries only form a unitary $1$-design \cite{blume-kohoutCuriousNonexistenceGaussian2011, zhuangScramblingComplexityPhase2019}, and unitary $2$-designs cannot exist for \ac{CV} systems, unless rigged Hilbert spaces are considered \cite{iosueContinuousvariableQuantumState2022}.
Such issues constrain characterization protocols to very specific settings, where the implemented unitary operators (which may include non-Gaussian single-mode unitaries) are only benchmarked w.r.t.\ a restricted set of input states \cite{fariasCertificationContinuousvariableGates2021}. 

For discrete variable systems, \ac{RB} \cite{EmeAliZyc05, Levi07EfficientErrorCharacterization, DanCleEme09, EmeSilMou07, KniLeiRei08,MagGamEme11, Magesan2012,proctor2017WhatRandomizedBenchmarking, wallman2018randomized, merkelRandomizedBenchmarkingConvolution2019} is the most widespread family of protocols for the estimation of average gate fidelities.
Its popularity is due to its robustness against \ac{SPAM} errors and its rather low demands on the measurement effort \cite{HarHinFer19}.
The standard \ac{RB} protocol is as follows:
To a fixed initial state, apply a sequence of Haar-random (Clifford) unitaries, followed by a final \emph{inversion gate} cancelling the action of the entire sequence.
Then, the success probability of restoring the initial state decays exponentially with the length of the sequence, and the decay rate is a proxy for the average gate fidelity of the gate set.

Recent efforts led to general guarantees for \ac{RB} protocols with general finite or compact groups \cite{francaApproximateRandomizedBenchmarking2018,helsenNewClassEfficient2019, helsenGeneralFrameworkRandomized2020, helsenEstimatingGatesetProperties2021, kongFrameworkRandomizedBenchmarking2021,heinrichRandomizedBenchmarkingRandom2023} instead of only unitary $2$-designs.
Due to the non-existence of the latter, this generality is clearly important when considering \ac{RB} of \ac{CV} systems.
However, there is still a fundamental problem:
The \ac{RB} signal generally consists of a linear combination of exponential decays in correspondence with the relevant \acp{irrep} of the used group \cite{francaApproximateRandomizedBenchmarking2018,helsenNewClassEfficient2019,helsenGeneralFrameworkRandomized2020}. 
Estimating these decay rates is already difficult in practice if more than a few irreps are involved \cite{helsenGeneralFrameworkRandomized2020, heinrichRandomizedBenchmarkingRandom2023} and may become impossible in the \ac{CV} setting.

Very recently, a bosonic randomized benchmarking protocol using random displacements has been put forward \cite{valahu_benchmarking_2024}.
Due to the Abelian character of the (projective) Heisenberg-Weyl group and the considered noise models, the complicated behaviour involving many decays is avoided in this setting, allowing for a simple analysis.

However, for more complicated -- in particular non-Abelian -- groups, one cannot hope for such a simple analysis.
To resolve the problems coming with many irreps, filtered \ac{RB} has been recently proposed \cite{helsenGeneralFrameworkRandomized2020, heinrichRandomizedBenchmarkingRandom2023}, building on insights from character \ac{RB} \cite{helsenNewClassEfficient2019}.
This protocol omits the inversion gate and instead performs a suitable post-processing of the data.
During the post-processing, contributions associated to  individual irreps can be isolated in a \ac{SPAM}-robust way, allowing to handle compact groups with many irreps.

In this work, we introduce a \ac{RB} protocol 
for the benchmarking of passive bosonic (Gaussian) transformations, called \emph{passive bosonic \ac{RB}}, or short \emph{passive \ac{RB}}.
The passive \ac{RB} protocol can be understood as a bosonic incarnation of the filtered \ac{RB} framework \cite{heinrichRandomizedBenchmarkingRandom2023}, and we build on its results to give a comprehensive analysis of passive \ac{RB}, including its sampling complexity, under the assumption of gate-dependent, stationary, and Markovian noise.

In contrast to the discrete variable systems explicitly considered in Ref.~\cite{heinrichRandomizedBenchmarkingRandom2023}, reasonable guarantees for a filtered \ac{RB} protocol in bosonic systems are harder to obtain.
This is because the action of passive transformation on the full bosonic Hilbert space is considerably more complicated than, e.g., the action of the qubit Clifford group, and thus a well-behaved protocol requires a careful identification of relevant and feasible experimental settings.

In the passive \ac{RB} protocol, we navigate these complications by considering experiments where the input state is a number state and \ac{PNR} measurements are performed, resembling boson sampling experiments \cite{aaronsonComputationalComplexityLinear2010}.
For practical reasons, we propose to use collision-free states with at most one particle per mode.
We then show that the passive \ac{RB} successfully benchmarks the average quality of passive transformations on a fixed number of particles, thereby separating the effect of particle number-preserving noise from particle loss using different post-processing of the same experimental data.
In general, we prove that there are as many decay rates as we have particles in the experiments -- however it is sufficient to only estimate a constant number of them in practice.
By representation-theoretic means, we derive explicit and concise formulas for the functions used in the post-processing and provide a Julia implementation.
As the protocol involves non-Gaussian elements, the classical post-processing requires the computation of matrix permanents, and it is therefore generally inefficient \cite{aaronsonComputationalComplexityLinear2010}.
Specifically, we expect the post-processing to be feasible for $\approx 30$ particles and modes, thus allowing the protocol to target an experimentally interesting regime.
The inefficiency of the post-processing should not be unexpected, as a similar behaviour can be found in linear cross-entropy benchmarking \cite{liu_benchmarking_2022,heinrichRandomizedBenchmarkingRandom2023}.
There, the post-processing involves the simulation of random circuit sampling, which is also known to be computationally hard \cite{boixoCharacterizingQuantumSupremacy2018,aruteQuantumSupremacyUsing2019,liuRedefiningQuantumSupremacy2021,hangleiter_computational_2023}.
We simulate the passive \ac{RB} protocol for a small number of modes, demonstrating that it works as intended.
Finally, we analyze the sampling complexity of our protocol by computing the relevant variances. 
We evaluate the obtained expressions numerically and find a very mild scaling with the number of modes $m$ which seems be logarithmic.

The effectively finite-dimensional setting is introduced for a good reason:
Probing the full bosonic Hilbert space involves infinitely many relevant irreps and thus also decay rates which, a priori, render a meaningful analysis impossible.
Moreover, such irreps would also appear with non-trivial multiplicities, making the protocol numerically challenging from a practical point of view, as it would require fitting matrix exponential decays \cite{helsenGeneralFrameworkRandomized2020,heinrichRandomizedBenchmarkingRandom2023}.
As a first variation of the passive \ac{RB} setup, we thus consider scenarios in which the input state is still a number state, but the measurement involves (balanced) heterodyne detectors.

The remainder of this work is structured as follows.
In \cref{sec:results} we highlight our main results:
\cref{sec:protocol} spells out our passive \ac{RB} protocol in full detail.
In \cref{sec:guarantees} we analyze our protocol, discuss the role of input states, and prove rigorous guarantees including bounds on its sampling complexity.
\cref{sec:discussion} is devoted to a general discussion of the passive \ac{RB} protocol, open problems, and possible extensions.
The technical details and proofs are deferred to the technical part, \cref{sec:passive_RB_Fock}.
Finally, \cref{sec:passive_rb_Gaussian} contains an analysis of passive \ac{RB} with balanced heterodyne measurements.

\section{Main results}
\label{sec:results}

\subsection{Notation}

We consider a bosonic system of $\nmodes \in \NN_{\geq 2}$ modes described by annihilation and creation operators $a_k$ and $a_k^\ad$ ($k=1,\dots,\nmodes$), respectively, satisfying the \acp{CCR} 
\begin{equation*}
	[a_k, a_l^\ad] = \delta_{kl} \, , \quad [a_k, a_l] = [a_k^\ad, a_l^\ad] = 0 \, , \quad k, l = 1, \dots, \nmodes \, , 
\end{equation*}
where $\delta_{k,l}$ is the Kronecker delta.
These operators act on
the Fock-Hilbert space $\fock_\nmodes \coloneqq \bigoplus_{n=0}^{\infty} \H_n^\nmodes$, where 
$\H_n^\nmodes$ is the subspace of $n$ bosons distributed over $\nmodes$ modes of dimension  $\dim \H_n^\nmodes = \binom{n+m-1}{n}$.
It is spanned by the \emph{Fock} or \emph{number states}
\begin{equation}
	\ket{\vec n} \equiv \ket{n_1, \dots, n_\nmodes} 
	\coloneqq
	\prod_{k=1}^\nmodes \frac{1}{\sqrt{n_k!}} a_k^{\ad n_k} \ket{\mathbf{0}}
	\, ,
\end{equation}
where $|\vec n|\coloneqq\sum_{i=1}^m n_i = n$ and $\ket{\mathbf{0}}\equiv \ket{0, \dots, 0}$ denotes the vacuum state of $\nmodes$ decoupled one-dimensional harmonic oscillators \cite{adessoContinuousVariableQuantum2014}.
We shall also consider coherent states, defined as $\ket{\alpha} = e^{-\abs{\alpha}^2/2} \sum_{n=0}^\infty \frac{\alpha^n}{\sqrt{n!}} \ket{n}$ for a single mode, with a straightforward extension to the multimode setting.
The set of passive transformations is the group of unitary operators on $\fock_\nmodes$ that leave the total number of particles invariant.
These are exactly the unitaries that induce a transformation of the bosonic operators as $a_k \mapsto \sum_{l=1}^\nmodes U_{lk} a_l$
for a unitary matrix $U=(U_{lk})_{l,k=1}^\nmodes$. 
Hence, the group of passive transformations can be identified with the unitary group $\U(\nmodes)$ \cite{arvindRealSymplecticGroups1995}.
Practically, these can also be thought as multimode interferometers, which can be decomposed in quadratically many two-modes interferometers and phase shift transformations only \cite{reckExperimentalRealizationAny1994, carolanUniversalLinearOptics2015, clementsOptimalDesignUniversal2017, deguiseSimpleFactorizationUnitary2018, tanResurgenceLinearOptics2018}.

The mapping $\rho:\,\U(\nmodes)\rightarrow \U(\fock_\nmodes)$ of $\U(\nmodes)$ to the group of passive transformations on $\fock_\nmodes$ is an example of a (unitary) \emph{representation}.
Generally, we require that a representation is compatible with group operations, i.e.~$\rho(gh)=\rho(g)\rho(h)$ and $\rho(g^{-1})=\rho(g)^{-1}$.
The above representation on $\fock_\nmodes$ is \emph{reducible}, i.e.~there are subspaces $V\subset\fock_\nmodes$ which are invariant under $\rho$: $\rho(g)(V)=V$ for all $g$.
The restriction $\rho|_V$ to such a subspace $V$ forms a representation in its own right, a so-called \emph{subrepresentation} of $\rho$.
$V$ is called irreducible if it does not contain another non-trivial invariant subspace.
We then call $\rho|_V$ an \emph{irreducible subrepresentation} (irrep) of $\rho$.
Since any $\rho(g)$ is particle number-preserving, one can show that the irreducible subspaces of $\fock_\nmodes$ are exactly given by the number spaces $\H_n^\nmodes$ \cite{anielloExploringRepresentationTheory2006}.

\subsection{The passive randomized benchmarking protocol}
\label{sec:protocol}

\subsubsection{Description of the protocol}
The passive \ac{RB} protocol follows the general procedure of the filtered \ac{RB} framework \cite{heinrichRandomizedBenchmarkingRandom2023}, adapted to the group of passive transformations acting on bosonic systems.
Additionally, we use post-selection to isolate particle loss rates, see Sec.~\ref{sec:details-protocol} and \ref{sec:particle_loss} for details.

The protocol consists of two phases:
During the \emph{data collection phase}, a fixed number state $\rho = \ketbra{\ninput}{\ninput}$ is used as an input to passive transformations and \acf{PNR} measurements
$\{\ketbra{\vec n}{\vec n}\}_{\vec n \in \NN^\nmodes}$ are performed on the output.
The collected data then undergoes a suitable \emph{post-processing phase}.
Motivated from boson sampling experiments, we propose to use a \emph{collision-free} input state $\ninput = (1,1,1,\dots,1,0,\dots,0)$ with $n\leq \nmodes$ particles \cite{aaronsonComputationalComplexityLinear2010}, see also the discussion in \cref{sec:input-state} below.

The passive \ac{RB} protocol can be summarized as follows.

\begin{enumerate}[label=(\Roman*)]
\item \textbf{Data collection.}
For different sequence lengths $\seqlength \in \sequences$, repeat the following steps $\numsamples$ times, 
\begin{enumerate}[label=(\roman*)]
	\item\label{enum:bosonic_rb_2} Prepare the state $\rho=\ketbra{\ninput}{\ninput}$.
	\item\label{enum:bosonic_rb_3} Apply passive transformations $g_1,\dots,g_\seqlength$ drawn i.i.d.~from the Haar probability measure on $\U(\nmodes)$.
	\item\label{enum:bosonic_rb_4} Measure the output state using \ac{PNR} detectors and store the outcome together with the sampled unitaries.
\end{enumerate}
\item \textbf{Post-processing.}
Fix a suitable \ac{irrep} $\lambda$ of $\U(\nmodes)$ and use its later-to-be-defined \emph{filter function} $f_\lambda(\vec n,g|\rho)$.
Assuming the data $\{ (\vec n^{(i)}, g_1^{(i)}, \dots, g_\seqlength^{(i)}) \}_{i = 1}^{\numsamples}$ has been gathered, compute the following mean estimator
\begin{equation}
	\label{eq:filtered_data_estimator}
	\hat F_\lambda(\seqlength) = \frac{1}{\numsamples'} \sum_{i=1}^{\numsamples} f_\lambda(\vec n^{(i)}, g_1^{(i)} \cdots g_\seqlength^{(i)} | \rho) \, ,
\end{equation}
where $\numsamples'\leq \numsamples$ is given in \cref{sec:details-protocol} below.
We refer to the data series $(\seqlength, \hat F_\lambda(\seqlength))_{\seqlength\in\sequences}$ as the (filtered) \emph{RB signal}.
Finally, perform an exponential fit according to the model $\hat F_\lambda(\seqlength) = A_\lambda r_\lambda^\seqlength$ to extract the \emph{decay rate} $r_\lambda$. 
Repeat for different \acp{irrep} $\lambda$ as needed.
\end{enumerate}

The definition of the filter function is given in \cref{sec:details-protocol}.
We also provide a numerical implementation of the protocol, including the filter function, on GitHub \cite{gitHubPassiveRB}.

For this work, we restrict our attention to number state inputs and \ac{PNR} measurements and justify this choice in more detail in \cref{sec:details-protocol}.
At the end of the main part, \cref{sec:extensions-protocol}, we discuss extensions of the passive \ac{RB} protocol to Gaussian states or measurements.
We present first results, using balanced heterodyne measurements, in \cref{sec:passive_rb_Gaussian}.

We would like to emphasize that Haar-random sampling from $\U(\nmodes)$ can be substituted by any distribution which converges sufficiently fast to the Haar measure \cite{heinrichRandomizedBenchmarkingRandom2023}.
For the main part of this paper, we resort to the Haar measure to simplify the presentation.
We refer the reader to \cref{sec:filtered-RB-background} in the technical part for more details and comments on the use of non-uniform distributions.

Finally, the procedure described above is sometimes referred to as \emph{single-shot estimation}, in the sense that only a single shot (or measurement) is taken per sequence.
In practice, it is often advantageous to use \emph{multi-shot estimation} instead, where multiple shots per sequence are recorded.
Then, the estimator \eqref{eq:filtered_data_estimator} changes in an obvious way, but converges to the same expectation value as \eqref{eq:filtered_data_estimator}.
Except for the sampling complexity discussion in \cref{sec:sampling-complexity} all results of this work also apply to the multi-shot estimator.

\subsubsection{Decay rates and fidelities}
The decay rates can be combined into an average performance measure of passive transformations on the $n$-particle subspace using the formula, cf.~Refs.~\cite[Cor.~10]{francaApproximateRandomizedBenchmarking2018} and \cite[Eq.~(242)]{helsenGeneralFrameworkRandomized2020}
\begin{equation}
\label{eq:rb-fidelity}
 F = (\dim\H_n^\nmodes)^{-2} \sum_{\lambda} d_{\lambda} r_\lambda \,,
\end{equation}
where the sum runs over all relevant irreps $\lambda$ with dimensions $d_\lambda$, see \cref{sec:guarantees}.
As $F$ corresponds to a weighted average, it is typically sufficient to only consider the largest weights which coincide with largest irreps.
This means that it may not necessary to run the post-processing phase (II) of the protocol for \emph{every} irrep $\lambda$.
We justify this statement and make it more precise later in \cref{sec:input-state}.

Although commonly done, we remind the reader that caution is advised if $F$ is interpreted as the average \emph{entanglement fidelity} of passive transformations.
The intuition behind this reasoning comes from the analysis of \emph{gate-independent errors}, for which it is straightforward to show that $F$ indeed coincides with the entanglement fidelity of the ``noise in-between gates'' \cite{francaApproximateRandomizedBenchmarking2018,helsenGeneralFrameworkRandomized2020}.
It should be clear, however, that this noise model is highly unrealistic and the relation to gate fidelities becomes more involved if gate-dependent errors are considered.
This is due to the inherent \emph{gauge freedom} of \ac{RB}, a property which \ac{RB} shares with other protocols that are agnostic against \ac{SPAM} errors, such as gate-set tomography \cite{stark_self-consistent_2014,merkel_self-consistent_2013,nielsen_gate_2021,brieger_compressive_2023}.
This gauge freedom has sparked a vivid discussion of the interpretation of \ac{RB} experiments, see~Refs.~\cite{proctor2017WhatRandomizedBenchmarking,helsenGeneralFrameworkRandomized2020} for more details.
In particular, whether such a physical interpretation is justified or not cannot be deduced from RB results alone.
If additional information about, e.g., the physical origin of noise processes is available, the interpretation of $F$ as fidelity may be justified -- we however cannot draw such a conclusion under the assumptions made in this paper.
Instead, we take the point of view of Ref.~\cite{helsenGeneralFrameworkRandomized2020} in that RB decays rates, and thus $F$ in Eq.~\eqref{eq:rb-fidelity}, should be regarded as quantities in their own right which provide a benchmark for the average quality of the used unitaries.
Besides the mentioned interpretational issues, let us note that, for discrete variable systems, a suitable affine transformation is typically applied to \cref{eq:rb-fidelity} to produce a proxy for the better-known \emph{average gate fidelity}.
This relies on the well-known relation $F_\mathrm{avg} = \frac{d F + 1}{d+1}$ for discrete variable systems in $d$ dimensions.
In the here considered setting, such a transformation would result in an expression akin to the average gate fidelity, where the average is performed over input states in the $n$-particle subspace.

\subsubsection{Details of the protocol}
\label{sec:details-protocol}
Choosing a number state $\ket{\ninput}$ as input to the protocol has the advantage that --\emph{ideally}-- the dynamics should happen on the finite-dimensional $n$-particle subspace $\H_n^\nmodes \subset \fock_\nmodes$ only, where $n=|\ninput|$.
On $\H_n^\nmodes$, the passive transformations $\U(\nmodes)$ act via the totally symmetric irrep $\tau_n^\nmodes$.
To describe noise, we transition to the density operator formalism, in which this action is described by conjugation, $\omega_n^\nmodes(g)(\rho) \coloneqq \tau_n^\nmodes(g)\rho\,\tau_n^\nmodes(g)^\dagger$.
If we assume -- for the sake of the argument -- that the noise is particle number-preserving, we can model the noisy transformations by quantum channels $\phi_n^\nmodes(g)$ on the space of bounded operators $\bounded{\H_n^\nmodes}$, provided that the noise is time-stationary and Markovian.
Regarding the $\phi_n^\nmodes(g)$ as (small) perturbations of $\omega_n^\nmodes(g)$, the filtered RB framework \cite{heinrichRandomizedBenchmarkingRandom2023} then predicts \ac{RB} decays in one-to-one correspondence with the \acp{irrep} of the \emph{reference representation} $\omega_n^\nmodes$.
As we show in \cref{sec:ref_repr}, the latter decomposes into exactly $n+1$ distinct irreps $\lambda_k$,\footnote{Here, $\lambda_0=\vec{1}$ is the trivial irrep and $\lambda_1=\Ad$ is the adjoint irrep of $\U(\nmodes)$.} and this is the number of decays to expect.

Arguably, \emph{particle loss} is one of the major noise sources in, e.g., photonic systems, hence assuming particle-number preserving noise is unrealistic.
This can nevertheless be accounted for in our framework, as we are performing \ac{PNR} measurements:
By post-selecting on particle number-preserving events, we obtain effective dynamics described by completely positive, \emph{trace non-increasing} maps $\phi_n^\nmodes(g)$ on $\bounded{\H_n^\nmodes}$ and we can again rely on the results in Ref.~\cite{heinrichRandomizedBenchmarkingRandom2023}.
Note that this reasoning also applies to noisy \ac{PNR} measurements by modelling them as a sub-normalized \ac{POVM}.
This approach is justified as long as the product of the photon loss probability
and
the one of photon \emph{gain} 
(e.g.~through thermal noise in detectors) is very small such that the error introduced by ignoring photon gain is comparable to the shot noise.
We expect this to be the case in modern photonic chips.

Finally, we introduce the \emph{filter function} which isolates the exponential decays for each irrep $\lambda = \lambda_k$ in the post-processing of the collected data.
Let $P_\lambda$ be the projector on the carrier space of $\lambda$ and let $d_\lambda$ be its dimension.
Then, given the measurement channel $\mathcal M(A) \coloneqq \sum_{n\in\NN^\nmodes} \sandwich{\vec n}{A}{\vec n} \ketbra{\vec n}{\vec n}$, we define the filter function as \cite{heinrichRandomizedBenchmarkingRandom2023}
\begin{equation}
\label{eq:filter_function}
\begin{aligned}
	f_\lambda (\vec n, g | \rho) &\coloneqq s_\lambda^{-1} \sandwich{\vec n}{\omega_n^\nmodes(g) \circ P_\lambda(\rho)}{\vec n} \, , \\
	s_\lambda &\coloneqq d_\lambda^{-1} \Tr[ P_\lambda \mathcal M ] \in \RR_{\geq 0} \, .
\end{aligned}	
\end{equation}
We will often simply write $f_\lambda (\vec n, g | \rho) \equiv f_\lambda (\vec n, g) $ once the input state $\rho$ is fixed.
We give a more concise formula for $f_\lambda$ in \cref{sec:guarantees} as \cref{thm:general_filter_informal}.
In particular, we show that $s_\lambda > 0$ as long as the number of modes $\nmodes$ is strictly larger than 1.
As the filter function is automatically zero if $|\vec n| \neq n$ due to $P_\lambda$, we can formally describe the post-selection on particle number-preserving events by defining 
\begin{equation}
 \numsamples_\textrm{pr} \coloneqq | \{ i \in [\numsamples] \,| \, |\vec n^{(i)}| = n \}| \,,
\end{equation}
and using the mean estimator \cref{eq:filtered_data_estimator} with the choice $\numsamples'=\numsamples_\textrm{pr}$.

\subsubsection{Estimation of particle loss rates}
\label{sec:particle_loss}
The same post-processing procedure can also be used to capture the estimation of particle loss rates as follows.
Define the filter function to be the indicator function\footnote{Note that this conicides with a suitably rescaled filter function \eqref{eq:filter_function} for the trivial irrep $\lambda_0$.}
\begin{equation}
	\label{eq:filter_indicator}
	f(\vec n) \coloneqq
	\begin{cases}
 		1 & \text{if } |\vec n| = n, \\
  		0 & \text{else},
 	\end{cases}
\end{equation}
and set $\numsamples'=\numsamples$.
Thus, the associated mean estimator $\hat F(\seqlength)$, cf.~\cref{eq:filtered_data_estimator}, simply yields the ratio $\numsamples_\textrm{pr}/\numsamples$.
To see that this leads to the wanted result, let us model particle loss by a beam splitter with transmittivity $\sqrt{p}$ in every mode (i.e.~$p$ is the probability of \emph{not} loosing a particle).
In addition, we assume transmittivities $\sqrt{p_\textrm{SP}}$ and $\sqrt{p_\textrm{M}}$ associated with state preparation and measurement, respectively.
Hence, the probability that a sequence of $\seqlength$ passive transformations preserves the particle number is $p_\textrm{SP}^n p_\textrm{M}^n p^{nl} = \EE[\numsamples_\textrm{pr}/\numsamples] = \EE[\hat F(\seqlength)]$.
Therefore, we can extract the transmittivity $\sqrt{p}$ from the data $(\seqlength, \hat F(\seqlength))_{\seqlength\in\sequences}$ using an exponential fit as before.

Note that the original filtered \ac{RB} protocol does not perform the post-selection in \cref{sec:details-protocol} above (i.e.~$\numsamples'=\numsamples$ in \cref{eq:filtered_data_estimator}).
This leads to a mixing of the decay rate of particle-preserving noise with the particle loss rate, resulting in a combined decay $(p^n r_\lambda)^\seqlength$.
By performing post-selection, we can consider these decays individually.

\subsection{Analysis and guarantees}
\label{sec:guarantees}

The filtered \ac{RB} framework \cite{heinrichRandomizedBenchmarkingRandom2023} implies a number of guarantees for the passive \ac{RB} protocol, which we specialize and extend in the following.
We rely on the \emph{implementation map} noise model which we justified using the post-selection argument in \cref{sec:protocol}.
In this model, the noise is modeled by replacing the representation $\omega_n^\nmodes$ with an implementation map $\phi_n^\nmodes$ on $\U(\nmodes)$, which takes values in the set of completely positive, trace non-increasing superoperators on $\H_n^\nmodes$.
This model allows for highly gate-dependent noise, which, however, needs to be stationary and Markovian.

In the following, we use the short-hand notations $\omega\equiv\omega_n^\nmodes$, $\tau\equiv\tau_n^\nmodes$, and $\phi\equiv\phi_n^\nmodes$ whenever the parameters $n$ and $\nmodes$ are clear from the context.

\subsubsection{The RB signal}

We now briefly sketch the arguments why the \ac{RB} signal $\hat F_\lambda(\seqlength)$, cf.~\cref{eq:filtered_data_estimator}, is well-approximated by an exponential decay \cite{heinrichRandomizedBenchmarkingRandom2023}.
In the implementation map model, a straightforward calculation shows that the expected \ac{RB} signal $F_\lambda(\seqlength) \coloneqq \EE[\hat F_\lambda(\seqlength)]$ is the contraction of a linear operator $\tilde{T}_\lambda^\seqlength$, and can thus be expanded into a linear combination of its (complex) eigenvalues, $F_\lambda(\seqlength) = \sum_j A_{\lambda,j} z_{\lambda,j}^\seqlength$.
To understand the eigenvalues $z_{\lambda,i}$ of $\tilde{T}_\lambda$, we first consider the noiseless case:
Then, $\tilde{T}_\lambda$ reduces to $T_\lambda(\mathcal{X}) = P_\lambda \int_{\U(\nmodes)} \omega(g)^\dagger\mathcal{X}\omega(g)\,dg$, also known as a \emph{channel twirl} in the quantum information literature, composed with the irrep projector $P_\lambda$.
As we show in the technical part, \cref{sec:ref_repr}, $T_\lambda$ has a single non-zero eigenvalue given by 1.
Inserting this into the expansion of $F_\lambda(\seqlength)$, we immediately see that it is constant in $\seqlength$, and in fact given by \cite{heinrichRandomizedBenchmarkingRandom2023}
\begin{equation}
	\label{eq:filtered_data_ideal}
		F_\lambda(\seqlength)
		=
		\Tr\left[ \rho P_\lambda(\rho) \right] \,.
\end{equation}
In the presence of noise, we instead find the expression $\tilde{T}_\lambda(\mathcal{X}) = P_\lambda \int_{\U(\nmodes)} \omega(g)^\dagger\mathcal{X}\phi(g)\,dg$.
Given that the noise is sufficiently weak, in the sense that $\tilde{T}_\lambda-T_\lambda$ is small, matrix perturbation theory implies that $\tilde{T}_\lambda$ has a single dominant real eigenvalue $r_\lambda \coloneqq z_{\lambda,0}$, which is separated from the remaining spectrum by a large gap and thus quickly dominates the signal.

\begin{proposition}[{\cite[Thm.~8]{heinrichRandomizedBenchmarkingRandom2023}}, informal]
\label{prop:signal-form}
 Suppose that the noise is sufficiently weak.
 Then, we have
 \begin{equation}
 \label{eq:signal-form-informal}
  F_\lambda(\seqlength) \approx A_\lambda r_\lambda^\seqlength \,,
 \end{equation}
 up to an additive error $\alpha \geq 0$ which is suppressed exponentially in $\seqlength$, and $A_\lambda\in\R$, $0<r_\lambda \leq 1$.
 Moreover, $r_\lambda$ does \emph{not} depend on the initial state and measurement, thus \ac{SPAM} errors only affect $A_\lambda$.
\end{proposition}

In the filtered \ac{RB} framework \cite{heinrichRandomizedBenchmarkingRandom2023}, noise is measured by the operator norm $\delta_\lambda\coloneq\snorm{\tilde{T}_\lambda-T_\lambda}$ and ``weak noise'' is such that $\delta_\lambda \leq 1/5$.
Then, Ref.~\cite{heinrichRandomizedBenchmarkingRandom2023} gives more precise conditions on the error suppression in the expected signal.
In particular, it is sufficient to choose the sequence length as
\begin{equation}
\label{eq:seq-length-bound}
 \seqlength \geq 
 \frac{\log \frac{d_\lambda}{s_\lambda} + 2 \log \frac1\alpha + 4}{2\log \frac{1}{2\delta_\lambda}} \,.
\end{equation}
where $d_\lambda$ is the dimension of the irrep and $s_\lambda$ is as in \cref{eq:filter_function}.
In the technical part of this work, \cref{sec:passive_RB_Fock}, \cref{prop:dimension_lambda_k} and \cref{thm:frame_operator_pnr}, we prove the following explicit formulae for $\lambda=\lambda_k$:
\begin{align}
	d_{\lambda_k} &=  \frac{2k+\nmodes-1}{\nmodes-1} \binom{k+m-2}{k}^2 \, , \\
	s_{\lambda_k} &= \frac{\nmodes-1}{2k+\nmodes-1}  \binom{k + \nmodes - 2}{k}^{-1} \, .
\end{align}
In particular, we have (for $\nmodes > 1$):
\begin{equation}
 \frac{d_{\lambda_k}}{s_{\lambda_k}} 
 = \left(\frac{2k+\nmodes-1}{\nmodes-1}\right)^2 \binom{k+m-2}{k}^3 \,.
\end{equation}
Hence, it is sufficient to take sequences of length $O( k \log\frac{k+\nmodes-2}{k} )$ which is $O(\nmodes \log\nmodes)$ in the \emph{non-collision regime} $n\leq \nmodes$ (as $k\leq n$).
However, we expect that these bounds are not tight.
In fact, \cref{eq:seq-length-bound} is not tight in the first place \cite{heinrichRandomizedBenchmarkingRandom2023}, nevertheless improving the bound is hard in the gate-dependent noise setting.
Based on the experience with discrete variable \ac{RB}, we expect that already very short sequences are sufficient and conjecture that the dependence on $d_{\lambda_k}/s_{\lambda_k}$ in \cref{eq:seq-length-bound} can in fact be dropped, meaning that constant-length sequences are sufficient.

The signal form is concretely shown in \cref{fig:decays_plot}, where we show the results of a simulation of the protocol with $\ket{\ninput} = \ket{\vec 1}_4 \equiv\ket{1111}$, assuming lossy gates for different values of the transmittivity.
By performing exponential fits, we retrieve values of the transmittivity $\sqrt{p}$ using $\numsamples=10000$ samples, matching the used transmittivities up to an error of $10^{-3}$.
We also simulated an experiment where the transmittivity is chosen at random for each transformation, resulting in decay rates that match the average over the allowed random decay rates.
We provide additional details on numerical experiments in \cref{sec:numerics}.

\begin{figure}[tb]
	 	\centering
	 	\includegraphics[width=1\linewidth]{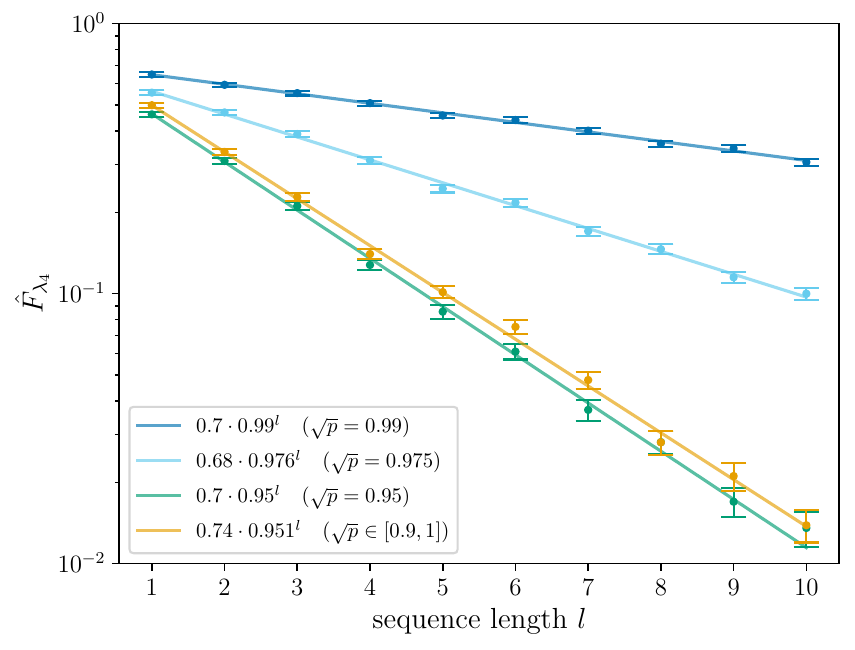}
	 	\caption{
			Passive \ac{RB} signal \eqref{eq:filtered_data_estimator} using $n=\nmodes=4$ and input state $\ket{\vec 1_4} \equiv \ket{1111}$ in the presence of particle loss noise,
			using $\numsamples=10000$ samples.
			The last (yellow) data set corresponds to simulated gate-dependent noise where the transmittivities $\sqrt{p}$ are drawn uniformly at random from $[0.9, 1]$ for every transformation in the sequence.
			More details can be found in \cref{sec:numerics}.
	 	}
	 	\label{fig:decays_plot}
\end{figure}

\subsubsection{Choice of input state}
\label{sec:input-state}

For standard \ac{RB} on discrete variable systems, the input state is typically chosen to be the all-zeros state $\ket{0^n}$, and, in fact, the choice of input state plays a minor role in this case.
The underlying reason is that we typically have a single non-trivial irrep (i.e., the adjoint \ac{irrep}) with respect to which all states are equivalent.

In the finite-dimensional bosonic setting, we show in \cref{lem:decomposition_conjugate_action} in \cref{sec:ref_repr} that $\omega_n^\nmodes$ decomposes into $n+1$ \acp{irrep} and thus there is a certain freedom in choosing the input state $\rho=\ketbra{\ninput}{\ninput}$.
In principle, the passive \ac{RB} protocol is agnostic of the input state.
However, its choice influences the overall magnitude of the \ac{RB} signal since the latter scales with the overlap of the state with the irrep of interest \cite{heinrichRandomizedBenchmarkingRandom2023}.
As the variance of $F_\lambda$ exhibits the same scaling behaviour, the mean estimator $\hat F_\lambda(l$) also converges more rapidly for smaller overlaps and hence the overall scale set by the overlap does not influence the sampling complexity of the protocol, see also \cref{sec:sampling-complexity}.
Nevertheless, it may happen that the input state has vanishing overlap with an irrep, in which case the \ac{RB} signal is identically zero and no information about decay rates can be extracted.
Moreover, very small overlaps, and thus signals, may lead the numerical issues in the post-processing.

Our proposal of using a collision-free state $\ninput = \vec{1}_n = (1,1,1,\dots,1,0,\dots,0)$ is motivated by practical considerations as generating higher Fock states can be a challenging task \cite{cooperExperimentalGenerationMultiphoton2013, lvovskyProductionApplicationsNonGaussian2020, walschaersNonGaussianQuantumStates2021}.
In \cref{fig:overlaps_n=m} and \cref{fig:overlaps_m_fixed}, we show its overlaps with the relevant irreps for $n=\nmodes$ and for $n\leq\nmodes$ with $\nmodes$ fixed, respectively.
We see that the largest overlap is always attained with the largest irrep $k=n$ and typically decreases with $k$, leading to small overlaps for $k \ll n$ if $\nmodes$ is not too small.
Generally, one can show that the collision-free state for $n=\nmodes$ particles always has zero overlap with the $\lambda_1$ (adjoint) irrep.
This behaviour can be avoided by optimizing $\ninput$, however likely leading to an input state that is hard to prepare in experiments.

\begin{figure}
 	\centering
 	\begin{subfigure}[ht]{.95\linewidth}
 		\includegraphics[width=1\linewidth]{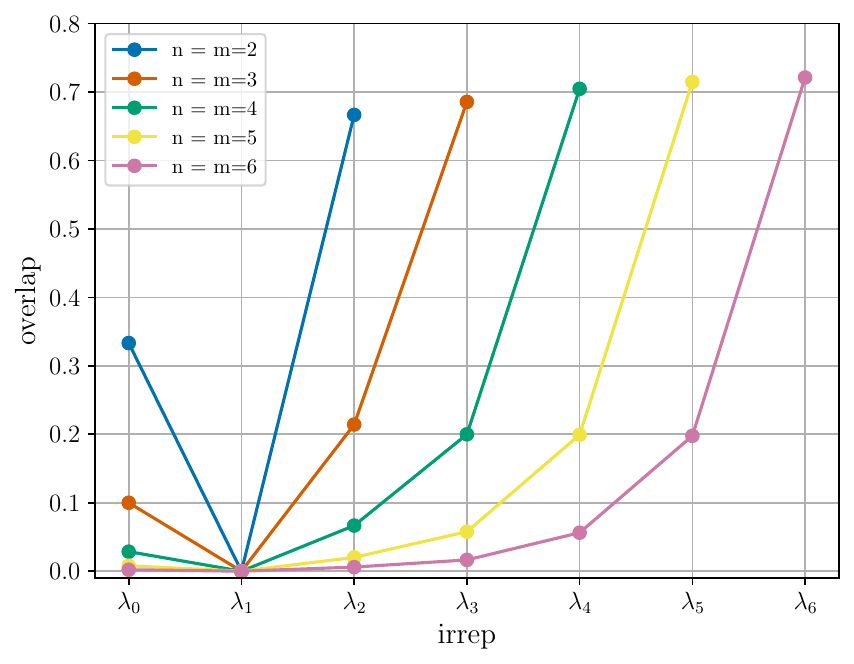}
 		\caption{For $n=\nmodes$, we observe no overlap with the adjoint \ac{irrep} $\lambda_1$.
 		For each $\nmodes$, the overlap is maximized by the largest \ac{irrep}.}
 		\label{fig:overlaps_n=m}	
 	\end{subfigure}
 	\begin{subfigure}[ht]{.95\linewidth}
 		\includegraphics[width=1\linewidth]{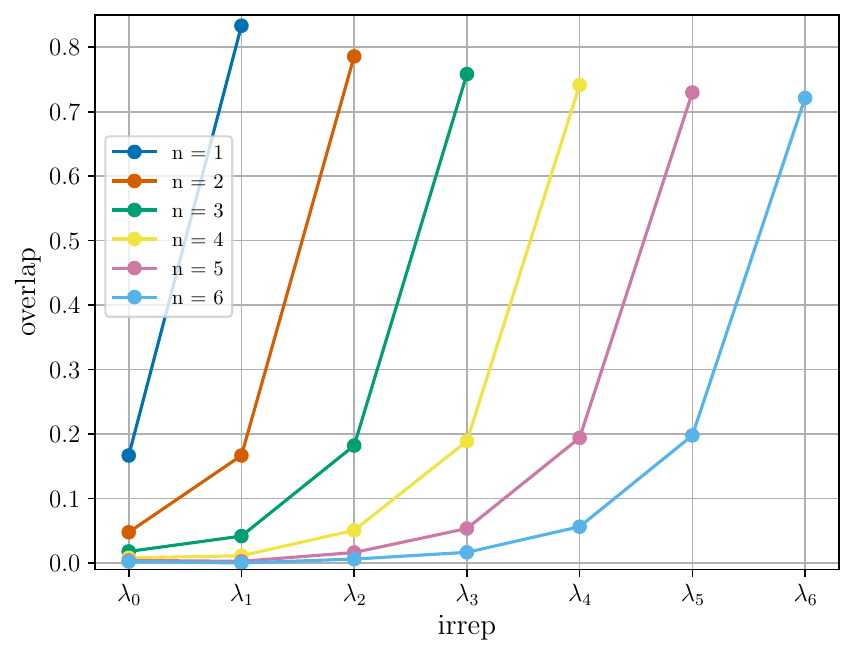}
 		\caption{Overlaps for $n\leq \nmodes = 6$.
 		As before, the maximum overlap is obtained for the largest irrep $\lambda_n$.
 		Note that the overlap with $\lambda_1$ is trivial only for $n = \nmodes$. }
 		\label{fig:overlaps_m_fixed}	
 	\end{subfigure}
 	\caption{
 		Overlaps of the collision-free $n$-particle state $\vec 1_n=(1, \dots, 1, 0, \dots, 0)$ with the \acp{irrep} $\lambda_k, k = 0, \dots, \nmodes$.
 		Lines are shown to enhance readability of data points and their relationships. 
	}
	\label{fig:overlaps}
\end{figure}

As a matter of fact, the problem of vanishing and very small overlaps of the collision-free state is not severe:
Recall that the overall benchmarking quantity $F$, cf.~\cref{eq:rb-fidelity}, is given as a weighted average of the decay rates $r_{\lambda_k}$, where the weights are given by the irrep dimension $d_{\lambda_k}$.
Since irreps with small $k$ also have small dimension, their contribution to the weighted average may be neglected.
More precisely, if we only take into account the $u$ largest irreps, their combined weight can be readily computed as
\begin{equation}
 (\dim\H_n^m)^{-2}\sum_{i=0}^{u-1} d_{\lambda_{n-i}}
 =
 1 - \left[ \prod_{i=1}^{u} \frac{n-i+1}{n+m-i}  \right]^2\,.
\end{equation}
In particular, if $n$ is fixed and $m\rightarrow\infty$, this converges to 1 for any $u$, and if $n=m\rightarrow\infty$, the weight converges to $1 - 4^{-u}$.
As the convergence is fast in both cases, taking only a constant number of irreps (the largest ones) into account gives an exponentially good approximation to $F$, already for a moderate number of modes.
For instance, for $n=5$ particles in $\nmodes=10$ modes, taking the 3 largest irreps covers 99.9\% of the total weight.
This also implies that it is sufficient to perform the post-processing phase (II) in the passive RB protocol, \cref{sec:protocol}, only a constant number of times.

\subsubsection{Evaluation of the filter function}


At the heart of the the passive RB protocol, \cref{sec:protocol}, lies the evaluation of the filter function \eqref{eq:filter_function} in the post-processing phase.
We briefly sketch the central steps of its computation leading to \cref{thm:general_filter_informal} in the following and refer to the technical part, \cref{sec:preliminaries} and \cref{sec:passive_RB_Fock}, for more details and proofs.
First, note that $\omega \simeq \tau \otimes \bar \tau$ can be explicitly decomposed into \acp{irrep} using a suitable generalization of the Clebsch-Gordan decomposition for the sum of two angular momenta in quantum mechanics \cite{sternbergGroupTheoryPhysics1994,alexNumericalAlgorithmExplicit2011}.
This yields the following expansion of Fock states
\begin{equation}
	\label{eq:cg_basis_vector}
	\ketbra{\vec n}{\vec n}
	\simeq
	\ket{\vec n, \vec n}\,
	= 
	\sum_{\lambda}
	\sum_{M_\lambda} \tilde C_{\vec n, \bar{\vec n}}^{M_\lambda} \ket{M_\lambda} \, ,
\end{equation}
where $\lambda$ runs over the \acp{irrep} of $\omega$ and the $\ket{M_\lambda}$ form a suitable basis of $\lambda$.
The $\tilde C_{\vec n, \bar{\vec n}}^{M_\lambda}$ correspond to (generalized) $\SU(\nmodes)$ Clebsch-Gordan coefficients, up to a phase.
Then, the projection onto a specific \ac{irrep} $\lambda$ acts by selecting only the corresponding terms in \cref{eq:cg_basis_vector}.
Re-expressing the $\ket{M_\lambda}$ in terms of Fock states and inserting the expansion into \cref{eq:filter_function} produces so-called \emph{permanents}, which are defined for an arbitrary $r \times r$ matrix $A$ as
\begin{equation}
\label{eq:def_permanent}
 \per(A) \coloneqq \sum_{\pi\in S_r} \prod_{i=1}^r A_{i\,\pi(i)} \, ,
\end{equation}
where $S_r$ denotes the symmetric group over $r$ symbols.
Let us define $A_{\vec n, \vec m}$ to be the matrix obtained from $A$ by taking $m_j$ copies of the $j$-th column of $U$ and then taking $n_i$ copies of the $i$-th row of the resulting matrix.
Then, the appearance of permanents is due to the well-known fact that $\sqrt{\vec n!}\sqrt{\vec m!}\,\sandwich{\vec n}{\tau(g)}{\vec m} = \per(g_{\vec n, \vec m})$ \cite{scheelPermanentsLinearOptical2004}, where we used the multi-index notation $\vec n! \coloneqq n_1! \dots n_\nmodes!$.

Finally, we have the following concise formula that can be used for numerical computation:

\begin{theorem}[PNR filter function -- informal]
	\label{thm:general_filter_informal}
	The filter function \eqref{eq:filter_function} is given as
	\begin{equation}
		\label{eq:filter_function_informal}
		f_{\lambda_k}(\vec n, g) 
		=
		\frac{1}{s_{\lambda_k}}
		\sum_{M} \tilde{C}_{\ninput, \ninputdual}^{M} \sum_{\vec n_1 \in \H_n^m} \tilde{C}_{\vec n_1, \bar{\vec n}_1}^{M} 
		\frac{\abs{ \per(g_{\vec n, \vec n_1}) }^2}{\vec n! \vec n_1!} \,.
	\end{equation}
	Here, $M$ labels a so-called \emph{Gelfand-Tsetlin} basis of the irrep $\lambda_k$ and the coefficients $\tilde C$ coincide with Clebsch-Gordan coefficients for $\SU(\nmodes)$ up to a phase.
\end{theorem}

We give more details on the Gelfand-Tsetlin basis, the computation of Clebsch-Gordan coefficients, and the occurring phases in the technical part, \cref{sec:preliminaries}, and formally prove \cref{thm:general_filter_informal} in \cref{sec:filter_function_PNR}.
We also give an alternative formula for $f_{\lambda_k}$ in \cref{cor:filter-function-matrix-coefficients} using matrix coefficients of the irrep $\lambda_k$.
These also correspond to permanents, however, of a different dimension \cite{Dhand_2015}.

Generally speaking, the computation of filter functions of the form \eqref{eq:filter_function} requires to simulate the entire experiment. 
As the proposed protocol involves non-Gaussian elements, the computation of filter functions is expected to be generally inefficient.
This is manifest in the occurrence of permanents in our filter functions.
The computational complexity of evaluating permanents is central to the complexity-theoretic arguments for boson sampling.
It is known that even approximating permanents is computationally hard \cite{aaronsonComputationalComplexityLinear2010,hangleiter_computational_2023}.
Nevertheless, these quantities can be computed efficiently in some scenarios \cite{chabaudContinuousVariableSamplingPhotonAdded2017, marshallSimulationQuantumOptics2023, bourassaFastSimulationBosonic2021, rahimi-keshariSufficientConditionsEfficient2016, heurtelStrongSimulationLinear2023,rogaClassicalSimulationBoson2020}.

This finding should not come as a surprise, as we find a similar behaviour in the discrete setting -- namely linear cross-entropy benchmarking \cite{liu_benchmarking_2022,heinrichRandomizedBenchmarkingRandom2023}.
There, the post-processing requires the simulation of random circuit sampling which is known to be computationally hard, analogous to boson sampling \cite{boixoCharacterizingQuantumSupremacy2018, aruteQuantumSupremacyUsing2019, liuRedefiningQuantumSupremacy2021,hangleiter_computational_2023}.

In practice, Clebsch-Gordan coefficients can be computed in polynomial time in the dimension of $\lambda$ \cite{alexNumericalAlgorithmExplicit2011}.
This can be significantly sped up by taking advantage of the symmetries of the weight diagrams under the action of the Weyl group of $\SU(\nmodes)$ \cite{alexSUClebschGordanCoefficients2012}.
Specifically, we expect this to be feasible for $\nmodes \approx 20-30$ modes.
In such regime, the computation of the permanents appearing in \cref{eq:filter_function_informal} is also still possible, as we consider $n \leq \nmodes$ particles and, by Ryser's formula, it can be computed in time $\LandauO(n \cdot 2^n)$ \cite{ryserCombinatorialMathematics,nijenhuis2014combinatorial}.

\subsubsection{Sampling complexity}
\label{sec:sampling-complexity}

In the following, we discuss the sample complexity of passive \ac{RB}, i.e.~the number of samples needed to guarantee that the estimator $\hat F_\lambda(\seqlength)$ is $\epsilon$-close to its expected value $F_\lambda(\seqlength)$ with high probability.
Recall from \cref{eq:filtered_data_estimator} that $\hat F_\lambda(\seqlength)$ is a mean estimator for the filter function $f_\lambda$.
Since the latter is only poorly bounded, we compute the variance $\Var[\hat F_\lambda(\seqlength)] = \Var[f_\lambda]/L$ (here the variance is still taken over length-$\seqlength$ sequences).
Then, we can use Chebyshev's inequality to ensure $\abs{\hat F_\lambda(\seqlength) - F_\lambda(\seqlength)} < \epsilon$ with probability $1-\delta$ given $L \geq \epsilon^{-2}\delta^{-1}\Var[f_\lambda]$ samples.

In general, analyzing the variance $\Var[f_\lambda]$ -- more precisely the second moment $\EE[f_\lambda^2]$ --  can be quite cumbersome, as the underlying probability distribution is given by Born probabilities involving the noisy input state, the noisy transformations, and the noisy measurements.
In the filtered \ac{RB} framework \cite{heinrichRandomizedBenchmarkingRandom2023} it is however shown that -- under reasonable assumptions on the noise\footnote{This is necessary as specially engineered noise can drastically change the behavior of the \ac{RB} signal, for instance by relabeling the measurement outcomes. Similar assumptions can be found throughout the \ac{RB} literature \cite{flammiaEfficientEstimationPauli2020,HarHinFer19}.} -- this problem can be reduced to analyzing the second moment $\EE[f_\lambda^2]_\mathrm{ideal}$ in the ideal, noiseless case.
In other words, the presence of noise cannot decrease the efficiency of filtered \ac{RB}.
Using this assumption, we have the following result bounding the variance of passive \ac{RB}:

\begin{theorem}[Variance bound -- informal]
	\label{thm:variance_bound_informal}
	The variance of the passive \ac{RB} protocol is bounded as
	\begin{equation}
	\label{eq:2nd_moment_informal}
		\Var[f_{\lambda_k}]
		\leq
		\EE[f_{\lambda_k}^2]_\mathrm{ideal}
		=
		s_{\lambda_k}^{-2}
		\sum_{\vec n\in\H_n^m} \tilde{g}_k(\vec n, \ninput) \, ,
	\end{equation}
	where $\tilde g_{k}(\vec n, \ninput)$ is a function of Clebsch-Gordan coefficients of the representations $\omega$ and $\lambda_k^{\otimes 2}$.
\end{theorem}

A formal version of \cref{thm:variance_bound_informal} is given as \cref{thm:second_moment_filter_function_pnr} in the technical part and also proven there.
It relies on expressing the second moment $\EE[f_{\lambda_k}^2]_\mathrm{ideal}$ as a suitable integral of the representations $\lambda_k^{\otimes 2}$ and $\omega_n^\nmodes$ over $\SU(\nmodes)$
\cite{heinrichRandomizedBenchmarkingRandom2023, merkelRandomizedBenchmarkingConvolution2019}.
Schur's lemma then implies that the non-trivial contributions to this integral are given by the \acp{irrep} of $\lambda_k^{\otimes 2}$ which are also contained in $\omega$.
We determine these irreps which allows us to finally write $\EE[f_{\lambda_k}^2]_\mathrm{ideal}$ in terms of suitable Clebsch-Gordan coefficients.

Arguably, \cref{eq:2nd_moment_informal} is not very explicit and it is thus not clear how $\EE[f_{\lambda_k}^2]$ scales with $k$, $n$, and $\nmodes$.
Finding a meaningful bound for the expression in \cref{eq:2nd_moment_informal} turns out to be difficult and we are left with the trivial bound
\begin{equation}
	\label{eq:2nd_moment_bound}
	 \EE[f_{\lambda_k}^2] \leq s_{\lambda_k}^{-2} = \left(1+\frac{2k}{\nmodes-1}\right) \binom{k+\nmodes-2}{k}^2\,,
\end{equation}
which suggests an exponential scaling like $2^{4\nmodes}$ in the worst case $k=n=\nmodes$.
This bound is however very loose:
A numerical evaluation of the second moment (see below) reveals that the upper bound \eqref{eq:2nd_moment_bound} already overestimates the second moment for $m=2$ by a factor of 10 and by a factor of $10^4$ for $m=5$.

\begin{figure}
	\centering
	\begin{subfigure}[ht]{.95\linewidth}
		\includegraphics[width=1\linewidth]{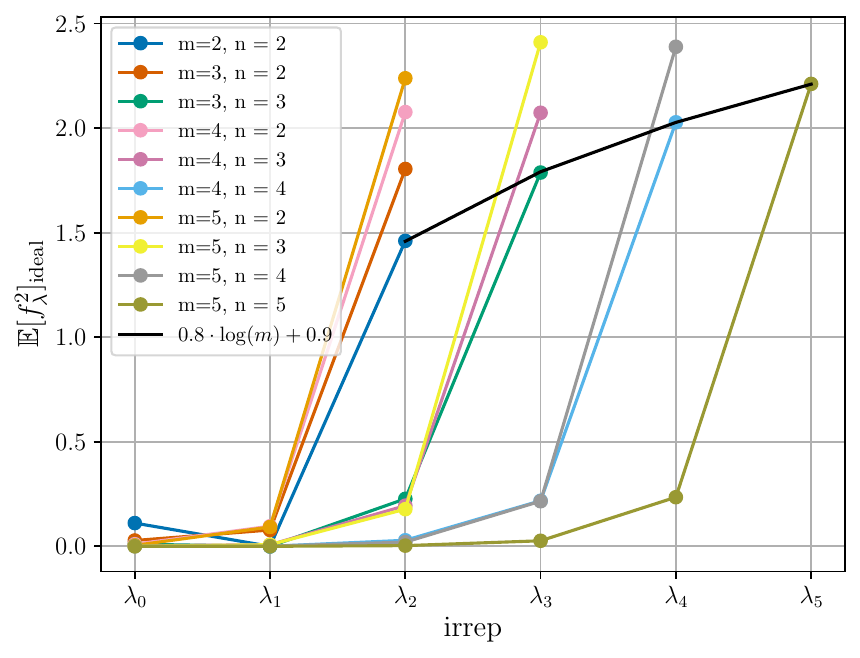}
		\caption{Empirical second moment with $\numsamples = 10000$ samples for different values of $n\leq \nmodes$ and $\nmodes\leq 6$.
 		The largest second moment (for fixed $\nmodes$) is always attained at the largest irrep $k=n$.
 		The scaling of the second moment for $\nmodes=n=k$ with $\nmodes$ seems to be logarithmic, however there is insufficient data the make a definite conclusion.}
		\label{fig:second_moment_empirical_fit}	
	\end{subfigure}
	\begin{subfigure}[ht]{.95\linewidth}
		\includegraphics[width=1\linewidth]{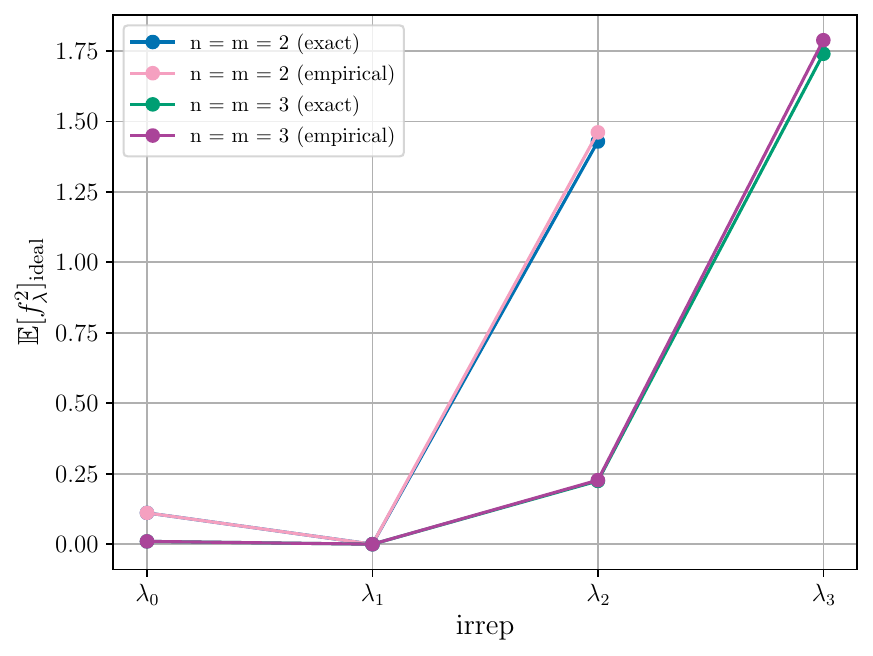}
		\caption{Empirical second moment with $\numsamples = 10000$ samples compared to the exact second moment \eqref{eq:2nd_moment_informal} for $n=\nmodes\leq 3$. 
 		We observe that this number of samples is already sufficient to give reasonable approximations of the second moment. }
		\label{fig:second_moment_empirical_vs_exact}	
	\end{subfigure}
	\caption{
		Second moments of the filter function $f_{\lambda_k}$ for different number of modes $\nmodes$, particles $n$, and irreps $\lambda_k$.
		Lines are shown to enhance readability of data points and their relationships.
	}
\end{figure}

The shear number of Clebsch-Gordan coefficients that need to be evaluated for the formula \eqref{eq:2nd_moment_informal} however limits a numerical study to very low values of $k$, $n$, and $\nmodes$ using standard algorithms \cite{alexNumericalAlgorithmExplicit2011}.\footnote{We expect \cref{eq:2nd_moment_informal} can be evaluated for a moderate number of modes exploiting the symmetries of the weight diagrams under the action of the Weyl group \cite{alexSUClebschGordanCoefficients2012}.
}
To increase the available data a bit, we furthermore estimate the variance empirically using $\numsamples = 10000$ samples from the noiseless outcome distributions.
From \cref{fig:second_moment_empirical_vs_exact}, we can observe that this  number of samples is already enough to give reasonable estimates of $\EE[f_{\lambda_k}^2]_{\mathrm{ideal}}$.
In \cref{fig:second_moment_empirical_fit}, we show the empirical second moments for $\nmodes \leq 5$ using a collision-free input state $\vec{1}_n$ with $n\leq\nmodes$ particles (cf.~\cref{sec:numerics} for further details on the sampling procedure).
Similar to the behaviour for the ideal first moments (overlaps) in \cref{fig:overlaps}, the second moment for a fixed $\nmodes$ is maximized at the largest irrep $k=n$.
In addition, we analyzed the scaling scaling of the empirical second moments for $m=n=k$ with $m$.
Although we have too few data points to make a conclusive statement, we find that the scaling is best compatible with a logarithmic scaling in $\nmodes$ (or a polynomial of very small degree $\approx \nmodes^{1/20}$).
This would be a double-exponential improvement over the trivial bound \eqref{eq:2nd_moment_bound}.

\section{Discussion}
\label{sec:discussion}

\subsection{Summary and open questions}

\emph{Passive \ac{RB}} is the first \ac{RB} protocol for the certification of passive bosonic transformations. 
Naively adopting discrete variable \ac{RB} to this setting would result in a complicated and challenging-to-analyze signal and thus additional care in the protocol detail is necessary.
Using techniques from filtered \ac{RB} \cite{heinrichRandomizedBenchmarkingRandom2023}, in combination with careful choices of initial state and measurement, we find an experimental setting in which a meaningful benchmark is indeed possible.

Using a collison-free $n$-particle state $\ket{1,\dots,1,0,\dots,0}$ as input and \ac{PNR} measurements, we effectively produce a finite-dimensional setting in which the \ac{RB} signal consists of only $(n+1)$ decays, which can be isolated in post-processing using \emph{filtering}.
In fact, we show that it is sufficient to only estimate a constant number of these decays to produce a meaningful benchmark for passive transformations on $n$ particles.
Moreover, the same procedure can be used to estimate particle loss rates.

In addition, we study the sampling complexity of passive \ac{RB} by computing the necessary variances.
The resulting expressions, involving many Clebsch-Gordan coefficients, are difficult to analyze directly.
We evaluate this expression numerically and additionally estimate the variance empirically for a small number of modes.
The obtained data suggests that the variance scales only very mildly with the number of modes $\nmodes$ (in the non-collision regime $n\leq\nmodes$), being best matched by a logarithmic function.
A definite conclusion can only be drawn from a larger data set.
However, this observation allows us to be cautiously optimistic that meaningful passive \ac{RB} experiments are possible with a low to moderate number of samples.

So far, a bottleneck of the protocol is the evaluation of the filter function \eqref{eq:filter_function_informal} during post-processing, which essentially involves the simulation of the experiment.
This can be reduced to the computation of permanents which are known to be computationally hard.
A similar problem is faced in linear cross-entropy benchmarking which requires the -- equally hard -- simulation of random circuit sampling.
Nevertheless, the latter can still be done for a moderate and experimentally interesting number of qubits.
We expect this is also possible for passive \ac{RB} by speeding up the computation of Clebsch-Gordan coefficients using symmetries \cite{alexSUClebschGordanCoefficients2012}, and reducing the permanents to the relevant irreps (cf.~\cref{cor:filter-function-matrix-coefficients}).
In principle, there are also explicit formulas available that could lead to a speed-up \cite{vilenkinRepresentationLieGroups1992Vol3,vilenkinRepresentationLieGroupsSpecialVol1,vilenkinRepresentationLieGroupsSpecialVol2}.
We provide a Julia implementation of the post-processing \cite{gitHubPassiveRB} and will study runtime improvements in future work.

The analysis of the behavior of passive \ac{RB} in the presence of \emph{particle distinguishability} is still an open problem, which we believe can be tackled by suitable extensions of the filtered \ac{RB} framework.
We leave this problem for future work.

We note that -- in principle -- character \ac{RB} is also able to isolate the individual decay rates associated to irreps \cite{helsenNewClassEfficient2019}.
The main difference to filtered \ac{RB} is that (i) an additional random $g\in\SU(\nmodes)$ is inserted before measuring, and (ii) the experimental data is post-processed using the \emph{character} $\chi_\lambda(g)$ of the irrep of interest evaluated at $g$ and the dimension of $\lambda$.
The randomization over $g$ then effectively implements the projector $P_\lambda$ (as in the filter function \eqref{eq:filter_function}).
Although the computation of the character is efficient (in $n$ and $\nmodes$), the sampling complexity is much worse than for passive \ac{RB}.
This is because sampling out the projector $P_\lambda$ converges very slowly:
Using the results of this work, it is not hard to show that the variance of character \ac{RB} is $\Omega(2^{2n})$ for irreps with $k\geq n/2$. 
Thus, both the total computational and experimental effort is exponential, while for passive \ac{RB}, it is likely only the former.
Another advantage of passive \ac{RB} is that our scheme can also be used with non-uniform sampling, cf.~\cref{sec:non-unifrom-sampling}, while this is not clear for character \ac{RB}.

\subsection{Extensions of the protocol}
\label{sec:extensions-protocol}

Arguably, the most straightforward --and experimentally simplest-- protocol would also involve a Gaussian input state and measurement.
As already noted in the introduction, this would however correspond to a classical linear optics experiment and one should not expect to capture errors which would occur in a quantum setting.
Hence, it is more instructive to consider a variation of the passive \ac{RB} protocol, in which either the input state or the measurement is made Gaussian.

However, one needs to be careful when leaving the effectively finite-dimensional setting which we explored in this work.
Indeed, an `inherently infinite-dimensional' \ac{RB} protocol is ill-behaved for the following reason:
Any passive transformation $U\in\U(m)$ is represented as a unitary operator $\tau^\nmodes(U)$ on the full Fock space $\fock_\nmodes \coloneqq \bigoplus_{n = 0}^\infty \H_n^\nmodes$.
Since $\U(\nmodes)$ is compact, $\tau^\nmodes$ is completely reducible, and decomposes into infinitely many finite-dimensional \acp{irrep} acting on the boson number subspaces \cite{anielloExploringRepresentationTheory2006}. 
Then, the conjugation representation $\omega \equiv \tau^\nmodes (\argdot) \tau^{\nmodes \ad}$ decomposes into infinitely many irreps, too.
This means we may find infinitely many decay rates (each associated to one \ac{irrep}) and it is unclear how to truncate those with a regularization argument, as the irreps lack a clear physical interpretation.
On top of that, the decomposition of $\omega$ is not multiplicity-free,\footnote{Note that $\omega$ restricted to $\mathcal{B}(\H_n^\nmodes)$ is equal to $\omega_n^\nmodes$ and --as mentioned in \cref{sec:details-protocol}-- each $\lambda_k$ appears in all $\omega_n^\nmodes$ for $n\geq k$.} which may complicate the post-processing and affect its numerical stability \cite{helsenGeneralFrameworkRandomized2020}.

As a first step, we consider a passive \ac{RB} protocol with a number input state and a balanced heterodyne measurement at the end of each mode.
We give some partial results for this setting in \cref{sec:passive_rb_Gaussian}.
In contrast to the \ac{PNR} setting, it is however not clear how to handle particle loss in such a protocol.
We thus leave a more thorough study of passive \ac{RB} protocols with Gaussian input states or Gaussian measurements for future work.

Lastly, we comment on the extension to other groups.
Broadly speaking, \emph{active} Gaussian transformations play a fundamental role in \ac{CV} systems,
for instance in \ac{GKP} encoding of qubits in bosonic states \cite{albertBosonicCodingIntroduction2022}, or in the preparation of input states for Gaussian boson sampling experiments \cite{bentivegnaExperimentalScattershotBoson2015,hamiltonGaussianBosonSampling2017,chakhmakhchyanBosonSamplingGaussian2017}.
However, randomized benchmarking of general Gaussian transformations faces many challenging issues, most strikingly the fact that this group -- the symplectic group $\Sp{2\nmodes,\RR}$ on $\nmodes$ modes-- 
is \emph{non-compact}.
As a consequence, a probability Haar measure cannot exist, and thus it is not even clear how to `randomize' in this context.
Hence, the filtered \ac{RB} framework focuses on compact groups \cite{heinrichRandomizedBenchmarkingRandom2023}.
Some of the ideas in Ref.~\cite{heinrichRandomizedBenchmarkingRandom2023} nevertheless carry over to the non-compact case if one considers suitable random walks on $\Sp{2\nmodes, \RR}$.
By the central limit theorem for reductive Lie groups \cite{benoistRandomWalksReductive2016, breuillardRandomWalksLie}, such random walks converge to a Gaussian distribution on $\Sp{2\nmodes, \RR}$.
However, this does not clarify --to the best of our knowledge-- the behaviour and convergence of the moment/twirling operator.
Since $\Sp{2\nmodes, \RR}$ has Kazhdan's property (T) \cite{bekkaKazhdanProperty2008}, 
the characterization of the relevant representation theory -- capturing the action of active transformations onto density operators-- plays a central role.
This involves the well-known metaplectic (or oscillator) representation of $\Sp{2\nmodes, \RR}$ and can -- in principle -- be approached using Howe duality \cite{howePerspectivesInvariantTheory1995,howe:oscillatorsemigroup}.
Due to infinite-dimensional \acp{irrep} \cite{folland:harmonics:phasespace} and open convergence questions, it is however far from obvious whether a meaningful \ac{RB} protocol can be developed for active transformations.
We leave the resolution of these questions for future work.

\subsection{Parallel work} 

In parallel to this work, \textcite{wilkins_2024} independently developed a randomized benchmarking protocol for non-universal, bosonic (or fermionic) analog simulators, called \emph{randomized analog benchmarking}, which is also based on filtered \ac{RB} \cite{heinrichRandomizedBenchmarkingRandom2023}.
Instead of sampling Haar-randomly, they use randomized sequences generated by a family of tunable Hamiltonians.
In their ``non-interacting'' case, their random sequences converge to the Haar measure on $\U(\nmodes)$ and thus the relevant representation theory (and thus post-processing) is the same as in this work.
Their work also provides a detailed numerical study of the protocol and the behaviour under different noise models.
Finally, the authors also discuss an ``interacting'' case (with on-site/in-mode interactions) which leads to a different group than $\U(\nmodes)$ for which all non-trivial irreps are joined into a single one, similar as for discrete \ac{RB} with unitary 2-designs.

\begin{acknowledgements}
	We would like to thank Ingo Roth, Jonas Haferkamp, Jonas Helsen, René Sondenheimer, Jan von Delft, Francesco Di Colandrea, and Mattia Walschaers for helpful discussions.
	We would like to express our gratitude to Zoltán Zimborás for bringing the new version of Piquasso \cite{kolarovszki2024piquassophotonicquantumcomputer} to our attention.
	Furthermore, we are grateful to the anonymous referees of the TQC and AQIS conferences for helpful comments on the manuscript.
	
	This work has been funded by the Deutsche Forschungsgemeinschaft (DFG, German Research Foundation) -- grant numbers 441423094 and 547595784, by the German Federal Ministry of Education and Research (BMBF) within the funding program ``quantum technologies -- from basic research to market'' via the joint project MIQRO (grant number 13N15522), and by the Fujitsu Germany GmbH and the Dataport as part of the endowed professorship ``Quantum Inspired and Quantum Optimization.''
	DG was supported by an NWO Vidi grant (Project No. VI.Vidi.192.109).
	MH has conducted part of this research while visiting QuSoft and Centrum Wiskunde \& Informatica, Amsterdam, and would like to express his gratitude for their support and hospitality.
	Publishing fees supported by Funding Programme Open Access Publishing of Hamburg University of Technology (TUHH).
\end{acknowledgements}

\clearpage

\onecolumngrid

\begin{center}
\textbf{\large Bosonic randomized benchmarking with passive transformations}\\[1em]
\textit{\large -- Technical part --}
\end{center}

\setcounter{equation}{0}
\setcounter{figure}{0}
\setcounter{table}{0}
\setcounter{section}{0}
\setcounter{theorem}{0}
\makeatletter
\renewcommand{\theequation}{T\arabic{equation}}
\renewcommand{\thefigure}{T\arabic{figure}}
\renewcommand{\thesection}{T\Roman{section}}
\renewcommand{\thetheorem}{T\arabic{theorem}}

\section{Technical Background}
\label{sec:preliminaries}

In this section, we review the main technical tools used in the proofs of our main results, shown in \cref{sec:passive_RB_Fock}.
First, we review \acp{irrep} of $\SU(\nmodes)$ and the Clebsch-Gordan decomposition in terms of \acl{GT} patterns, then we briefly review additional technical details concerning filtered \acl{RB}.

\subsection{Representations of \texorpdfstring{$\SU(\nmodes)$}{SU(m)}}
\label{sec:rep-theory}

As $\SU(\nmodes)$ is a compact Lie group, we can use \emph{weight theory} to understand its representation theory, see e.g.~Ref.~\cite{hallLieGroupsLie2015,fultonRepresentationTheory2004,goodmanSymmetryRepresentationsInvariants2009}.
To this end, we consider the diagonal subgroup of $\SU(\nmodes)$, the elements of which can be written as
\begin{equation}
  \diag( e^{i\theta_1}, \dots, e^{i\theta_{\nmodes-1}}, e^{- i \sum_{j=1}^{\nmodes-1} \theta_j })\,.
\end{equation}
Hence, the diagonal subgroup is isomorphic to the $(\nmodes-1)$-torus $T^{\nmodes-1}\simeq (S^1)^{\times(\nmodes-1)}$.
Since $T^{\nmodes-1}$ is Abelian, the restriction of any irreducible representation $\rho$ of $\SU(\nmodes)$ to $T^{\nmodes-1}$ decomposes into one-dimensional, irreducible subrepresentations, given by characters of $T^{\nmodes-1}$:
\begin{equation}
	\label{eq:weight_from_repr}
	\rho|_{T^{\nmodes-1}} \simeq \bigoplus_{w\in\ZZ^{\nmodes-1}} \chi_w \otimes \1_{m_w}\,.
\end{equation}
Here, $\chi_w(\theta) = e^{i w^\top \theta}$, $w\in\ZZ^{\nmodes-1}$, are characters of $T^{\nmodes-1}$ and $m_w$ are the \emph{(inner) multiplicities}.
Then, the \emph{weights} of $\rho$ are given by the vectors $w$ for which $m_w\neq 0$.
The \emph{weight space} of $w$ is the isotypic component of $w$, i.e.~the carrier space of $\chi_w \otimes \1_{m_w}$, and has dimension $m_v$.
We can order weights lexicographically, i.e.~$v>w$ if and only if $v_i > w_i$ for all $i$.
We then say that the weight $v$ is higher than $w$.
Importantly, the \emph{theorem of the highest weight} \cite[Thms.~9.4 and 9.5]{hallLieGroupsLie2015} states that every irrep has a unique highest weight and is uniquely determined by it.
Hence, we can characterize the irreps of $\SU(\nmodes)$ by studying the structure of weights.

A convenient way of doing that is given by \emph{Young diagrams}:
Let $n \geq 0$ be a non-negative integer and let $\lambda = (\lambda_1, \dots, \lambda_\nmodes)$ be a partition of $n$, i.e. $\lambda_1 \geq \lambda_2 \geq \dots \geq \lambda_\nmodes \geq 0$ with $\sum_{i=0}^\nmodes \lambda_i = n$.
Any such partition is identified with a Young diagram, namely a collection of boxes arranged in left-justified rows with a weakly decreasing number of boxes in each row, where the $i$-th row contains $\lambda_i$ boxes, e.g.
\begin{equation}
	\ydiagram{5,3,2}
\end{equation}
corresponds to $\lambda=(5,3,2)$.
Young diagrams uniquely determine \acp{irrep} of $\SU(\nmodes)$ up to constant shifts, namely, $(\lambda_1, \dots, \lambda_\nmodes)$ and $(\lambda_1+c, \dots, \lambda_\nmodes+c)$ identify the same \ac{irrep} for any integer $c$.
Therefore, as the rank of $\SU(\nmodes)$ is $\nmodes-1$, by convention we assume $\lambda_\nmodes=0$ without loss of generality.
In the following we will not distinguish between Young diagrams and the corresponding \acp{irrep}, unless otherwise specified.

For a given \ac{irrep} $\lambda$, the \emph{dual (or contragredient)} representation $\lambda^*$ is defined as $\lambda^*(g) \coloneqq \lambda(g^{-1})^T$ for each $g \in \SU(\nmodes)$, acting on the vector space dual of $\lambda$'s carrier space.
Notably, $\lambda^*$ is also irreducible, see \cite[Prop~4.22]{hallLieGroupsLie2015}.
In a fixed orthonormal basis, we also have $\lambda^* \cong \bar \lambda$, where $\bar \lambda(g) \coloneqq \overline{\lambda(g)}$ for each $g \in \SU(\nmodes)$ denotes the complex conjugate representation of $\lambda$ (w.r.t.~the fixed basis).
For any irrep $\lambda = (\lambda_1, \dots, \lambda_\nmodes)$, this implies $\lambda^*$ is identified by the \emph{dual} Young diagram $\bar \lambda \coloneqq (\lambda_1 - \lambda_\nmodes, \lambda_1 - \lambda_{\nmodes-1}, \dots, \lambda_1 - \lambda_2, 0)$.
More practically, $\bar \lambda$ is constructed by completing $\lambda$ to a $\nmodes \times \lambda_1$ rectangular Young diagram, and rotating the diagram formed by the newly added boxes by $180^\circ$.
For instance, 
\begin{equation}
	\ytableausetup{centertableaux, boxsize=1.2em}
	\ydiagram[*(green)]
	{5+0,3+2,2+3,1+4,0+5}
	*[*(yellow)]{5,5,5,5,5}
	\longrightarrow\quad
	\ydiagram{5,3,2,1} \, , \quad \ydiagram{5,4,3,2}
\end{equation}
are dual Young diagrams in $\SU(5)$.

A \emph{semi-standard Young tableaux} is a filling of a Young diagram with entries taken from any totally ordered set (here, $\NN$) such that the entries are weakly increasing across each row and strictly increasing down each column.
For instance,
\begin{equation}
	\label{eq:ytableaux_examples}
	\ytableaushort{1 1 2 2 3, 2 2 3, 3 4}*{5, 3, 2} \, , \quad
	\ytableaushort{1 1 1 1 1, 2 2 3, 3 4}*{5, 3, 2}
\end{equation}
are semi-standard Young tableaux of shape $\lambda = (5,3,2,0)$.
For a fixed \ac{irrep} $\lambda$ of $\SU(\nmodes)$, the semi-standard Young tableaux of shape $\lambda$ label an orthonormal basis of $\lambda$, sometimes referred as the \emph{Weyl basis}.
For instance, the Young tableaux
\begin{equation}
	\label{eq:Weyl_basis_example}
	\begin{aligned}
		\ytableausetup{boxsize=1.2em}
		\ytableaushort{1 1, 2}*{2, 1} \, , \quad \ytableaushort{1 1, 3}*{2, 1} \, , \quad \ytableaushort{1 2, 2}*{2, 1} \, , \quad \ytableaushort{1 2, 3}*{2, 1} \, , \\
		\ytableaushort{1 3, 2}*{2, 1} \, , \quad \ytableaushort{1 3, 3}*{2, 1} \, , \quad \ytableaushort{2 2, 3}*{2, 1} \, , \quad \ytableaushort{2 3, 3}*{2, 1} \, .
	\end{aligned}
\end{equation}
identify an orthonormal basis for the $\SU(3)$ adjoint \ac{irrep} $\lambda=(2,1,0)$.

\subsubsection{Gelfand-Tsetlin patterns.}
\label{sec:GTpattern}

A more convenient way of labelling basis vectors for any irrep $\lambda = (\lambda_1, \dots, \lambda_\nmodes)$ of $\SU(\nmodes)$ is by \emph{\ac{GT} patterns}.
A \ac{GT} pattern $M$ of shape $\lambda$ and length $\nmodes$ is represented by a triangular table with $\nmodes$ rows, the $j$-th row containing $j$ integers (counting from the bottom to the top)
\begin{equation}
    \label{eq:GTpattern}
    M =
	\begin{psmallmatrix}
        M_{1,\nmodes} && M_{2,\nmodes} && \dots && M_{\nmodes-1, \nmodes} && M_{\nmodes,\nmodes} \\
        &M_{1,\nmodes-1} && M_{2,\nmodes-1} & \dots & M_{\nmodes-2,\nmodes-1} && M_{\nmodes-1,\nmodes-1}&\\
        &&\ddots&&\vdots&&\reflectbox{$\ddots$}&&\\
        &&& M_{1,2}&&M_{2,2}&&&\\
        &&&& M_{1,1}&&&&
    \end{psmallmatrix} \, ,
\end{equation}
and $M_{i,\nmodes} = \lambda_i$ for every $i \in [\nmodes]$ (and, in particular, $M_{\nmodes, \nmodes}=0$ by convention).
The entries satisfy the \emph{interlacing} or \emph{inbetweenness} condition:
\begin{equation} 
	M_{i,j+1} \geq M_{i,j} \geq M_{i+1,j+1}
\end{equation}
for every $i \in [\nmodes-1]$ and $j \in [\nmodes-1]$.
We remark that by convention the indices in a \ac{GT} pattern are swapped, namely the first one is the column index (which is increasing from left to right, as usual), while the second one is the row index. 
We denote the set of \ac{GT} patterns of shape $\lambda$ by $\GT{\lambda}$.

An orthonormal basis for $\lambda$ --referred as the \acl{GT} basis-- is given by state vectors $\{\ket{M}\}$, where $M$ is a valid \ac{GT} pattern 
with top row $\vec M_1 \equiv (M_{1,1}, \dots, M_{1,\nmodes}) = \lambda$.
Hence, the dimension of $\lambda$ is equal to the number of such states, for which the following formula holds:
\begin{equation}
	\label{eq:dim_irrep_SU(N)}
	\dim \lambda = \prod_{1 \leq i \leq j \leq \nmodes} \left( 1 + \frac{M_{i,\nmodes} - M_{j,\nmodes}}{j-i} \right) \, .
\end{equation}
In terms of the \ac{GT} basis, the highest weight vector of $\lambda$ is identified by the pattern maximizing the inbetweenness conditions, namely
\begin{equation}
	\label{eq:GTpattern_HWV}
	M_\mathrm{max} =
	\begin{psmallmatrix}
		M_{1,\nmodes} && M_{2,\nmodes} && \dots && M_{\nmodes-1, \nmodes} && M_{\nmodes,\nmodes} \\
		&M_{1,\nmodes} && M_{2,\nmodes} & \dots & M_{\nmodes-2,\nmodes} && M_{\nmodes-1,\nmodes-1}&\\
		&&\ddots&&\vdots&&\reflectbox{$\ddots$}&&\\
		&&& M_{1,\nmodes}&&M_{2,\nmodes}&&&\\
		&&&& M_{1,\nmodes}&&&&
	\end{psmallmatrix}
\end{equation}
(likewise, the lowest weight vector of $\lambda$ is obtained by minimizing the inbetweenness conditions).

\ac{GT} patterns are in one-to-one correspondence with semi-standard Young tableaux.
In fact, for a given Young tableau $T$ of shape $\lambda$, the shape of the corresponding \ac{GT} pattern $M$ is the same shape as $T$ and the entry $M_{j,k}$ of $M$ (remember the different indexing in \eqref{eq:GTpattern} compared to matrix indices) is given by the number of entries in the $j$-th row of $T$ which are less or equal than $k$.
Conversely, given a \ac{GT} pattern $M$ of shape $\lambda$, the shape of the corresponding Young tableau $T$ is determined by the first row of $M$ and $M_{j,k} - M_{j,k-1}$ is the number of $k$'s in the $j$'th row of $T$.
Throughout this to work, we assume that all illegal coefficients are set to $0$.
For instance, the Young tableaux in \cref{eq:ytableaux_examples} corresponds to the \ac{GT} patterns
\begin{equation}
	\begin{psmallmatrix}
		5 && 3 && 2 && 0 \\
		& 5 && 3 && 2 & \\
		&& 4 && 2 & \\
		&&& 2 &&&
	\end{psmallmatrix}
	\, , \quad
	\begin{psmallmatrix}
		5 && 3 && 2 && 0 \\
		& 5 && 3 && 1 & \\
		&& 5 && 2 & \\
		&&& 5 &&&
	\end{psmallmatrix} \, ,
\end{equation}
and the Weyl basis \cref{eq:Weyl_basis_example} corresponds to the \ac{GT} basis
\begin{equation}
	\begin{aligned}
		\begin{psmallmatrix}
			2 && 1 && 0 \\
			& 2 && 1 & \\
			&& 2 && 
		\end{psmallmatrix} \, , \quad
		\begin{psmallmatrix}
			2 && 1 && 0 \\
			& 2 && 0 & \\
			&& 2 && 
		\end{psmallmatrix} \, , \quad
		\begin{psmallmatrix}
			2 && 1 && 0 \\
			& 2 && 1 & \\
			&& 1 && 
		\end{psmallmatrix} \, , \quad
		\begin{psmallmatrix}
			2 && 1 && 0 \\
			& 2 && 0 & \\
			&& 1 && 
		\end{psmallmatrix} \, , \\
		\begin{psmallmatrix}
			2 && 1 && 0 \\
			& 1 && 1 & \\
			&& 1 && 
		\end{psmallmatrix} \, , \quad
		\begin{psmallmatrix}
			2 && 1 && 0 \\
			& 1 && 0 & \\
			&& 1 && 
		\end{psmallmatrix} \, , \quad
		\begin{psmallmatrix}
			2 && 1 && 0 \\
			& 2 && 0 & \\
			&& 0 && 
		\end{psmallmatrix} \, , \quad
		\begin{psmallmatrix}
			2 && 1 && 0 \\
			& 1 && 0 & \\
			&& 0 && 
		\end{psmallmatrix} \, .
	\end{aligned}
\end{equation}

For a given \ac{GT} pattern $M$, the \emph{weight} of $\ket{M}$ is the $(\nmodes-1)$-tuple $w_M \coloneqq (w_{1}^{(M)}, \dots, w_{\nmodes-1}^{(M)})$ introduced in \cref{eq:weight_from_repr}.
Each entry $w_i^{(M)}$ can be determined by $M$ as follows:
\begin{equation}
	\label{eq:weight_element}
	w_j^{(M)} = 2 \sum_{i=1}^j M_{i,j} - \sum_{i=1}^{j-1} M_{i,j-1} - \sum_{i=1}^{j+1} M_{i,j+1} \, .
\end{equation}
We remark weights generalizes the notion of magnetic quantum number $m$ for $\SU(2)$ in the quantum theory of angular momentum to arbitrary many modes.
In the latter expression, we assume the convention that $\sum_{i=1}^{j-1} M_{i,j-1} = 0$ if $j = 1$.

All states $\ket{M}$ of the same weight $w$ form a basis of the weight space of $w$.
Note that these are typically not one-dimensional, except for $\SU(2)$ \cite{biedenharnRepresentationsSemisimpleLie2004, bairdRepresentationsSemisimpleLie2004, bairdRepresentationsSemisimpleLie1964}.
For instance, the \ac{GT} patterns
\begin{equation}
	M =
	\begin{psmallmatrix}
		3 && 2 && 0 \\
		& 3 && 0 & \\
		&& 1 &&
	\end{psmallmatrix} \, , \quad
	\tilde M =
	\begin{psmallmatrix}
		3 && 2 && 0 \\
		& 2 && 1 & \\
		&& 1 &&
	\end{psmallmatrix}
\end{equation}
are such that $w^{M} = w^{\tilde M} = \left(-1, 0\right)$.
Thus the dimension of the weight space of $w$ (the inner multiplicity) corresponds to the number of \ac{GT} states (or, equivalently, to the number of semi-standard Young tableaux) with weight $w$.
These amount to Kostka's numbers, and can be computed e.g. with recursive algorithms \cite{delaneyInnerRestrictionMultiplicity1969}.

It is worth reminding here a related definition of a \emph{weight} for $\U(\nmodes)$ (or $\GL(\nmodes)$) representations. 
\acl{GT} patterns also label basis vectors of the corresponding irreps of such groups.
However, the condition $M_{m,m} = 0$ is dropped.
In this case, the weight of the GT pattern $M$ is defined as $w_{\U(\nmodes)}^M \coloneqq (w_{\U(\nmodes), 1}^M, \dots, w_{\U(\nmodes), \nmodes}^M)$ with
\begin{equation}
	w_{\U(\nmodes), j}^M \coloneqq \sum_{i=1}^j M_{i,j} - \sum_{i=1}^{j-1} M_{i,j-1} \, ,
\end{equation}
therefore
\begin{equation}
	w_j^{(M)} =  w_{\U(\nmodes), j}^M - w_{\U(\nmodes),j+1}^M \, .
\end{equation}
An equivalent definition that will be used in this work is the following:
The \emph{tableau weight} is the $\nmodes$-tuple $w^T = (w_1^T, \dots, w_\nmodes^T)$ where $w_i^T$ is the total number of $i$ entries in $T$.
If $T$ is the semi-standard Young tableau associated with $M$, we also have
\begin{equation}
	w_j^T \coloneqq \sum_{i=1}^j M_{i,j} - \sum_{i=1}^{j-1} M_{i,j-1} \, .
\end{equation}
The weights $w_M$ and $w^T$ are clearly related in a similar way:
\begin{equation}
	w_j^{(M)} =  w_j^T - w_{j+1}^T  \, .
\end{equation}
We notice en passant that the tableau weight is always a sequence of non-negative integers, which sometimes is more convenient for combinatorics.
For instance, the tableaux $T$ and $\tilde T$ associated with the \ac{GT} patterns $M$ and $\tilde M$ from the former examples are
\begin{equation*}
	T = \ytableaushort{1 2 2, 3 3}*{3, 2} \, , \quad 
	\tilde T = \ytableaushort{1 2 3, 2 3}*{3, 2}
\end{equation*}
and we have $w^{T} = w^{\tilde T} = (1, 2, 2)$.

\subsubsection{Dual GT patterns.}
For a given \ac{GT} pattern $M$ of shape $\lambda$, we define the \emph{dual} \ac{GT} pattern $\bar M$ of shape $\bar \lambda$ as the pattern with entries satisfying the relation
\begin{equation}
	\label{eq:GTpattern_dual}
	\bar M_{i,l} \coloneqq M_{1,\nmodes} - M_{l-i+1,l} \, ,
\end{equation}
i.e.~$\bar M = M_{1,\nmodes} + \tilde M$, where the sum shall be interpreted as the element-wise constant shift of $\tilde M$ by $M_{1, \nmodes}$, and $\tilde M$ is given by
\begin{equation}
	\label{eq:GTpattern_dual_no_shift}
	\tilde M =
	\begin{psmallmatrix}
		-M_{\nmodes,\nmodes} && -M_{\nmodes-1,\nmodes} && \dots && -M_{2, \nmodes} && -M_{1,\nmodes} \\
		&-M_{\nmodes-1,\nmodes-1} && -M_{\nmodes-2,\nmodes-1} & \dots & -M_{2,\nmodes-1} && M_{1,\nmodes-1}&\\
		&&\ddots&&\vdots&&\reflectbox{$\ddots$}&&\\
		&&& -M_{2,2}&&-M_{1,2}&&&\\
		&&&& -M_{1,1}&&&&
	\end{psmallmatrix} \,.
\end{equation}
By construction, $\bar M$ defines a basis state for the dual irrep $\lambda^* \cong \bar \lambda$.
The conjugate operation is also such that each state $\ket{M}$ of $\lambda$ is associated with a unique conjugate state $\ket{\bar M}$ of $\bar \lambda$.
Specifically, the conjugation operation is such that \cite{bairdRepresentationsSemisimpleLie1964}
\begin{equation}
	\label{eq:conjugation_relative_phase}
	\ket{M} = (-1)^{\varphi(M)} \ket{\bar M} \, ,
\end{equation}
for a suitable phase function that can be determined as follows:
For a \ac{GT} pattern $M$, define the function
\begin{equation}
	s_M(k) = \sum_{j=1}^{k} \sum_{i=1}^{j} M_{i,j} \, ,
\end{equation}
which corresponds to the sum of the labels of $M$ in the first $k$ rows (counting from bottom to top).
Then \cite{bairdRepresentationsSemisimpleLie1964},
\begin{equation}
	\label{eq:phase_dual_GTpattern}
	\varphi(M) = s_M(\nmodes-1) - s_{M_\mathrm{max}}(\nmodes-1) \, , 
\end{equation}
where $M_\mathrm{max}$ is defined in \cref{eq:GTpattern_HWV}.
We remark that the latter holds up to a redefinition of every GT basis state by a global phase.
In this work, we consider the Condon-Shortley convention that fixes the global phase in such a way that usual relations for $\SU(2)$ are retrieved.

\subsubsection{Symmetric irreps in \texorpdfstring{$\SU(\nmodes)$}{SU(m)}.}
\label{sec:symmetric_irrep}

In this section, we summarize a few basic facts concerning symmetric irreps of $\SU(\nmodes)$, as they are of central importance throughout this work.
By construction, the space of $n$ particles over $\nmodes$ modes is maximally symmetric under permutations over the modes.
This implies that the action of $g \in \SU(\nmodes)$ on such space is described by the \ac{irrep}
\begin{equation}
	\label{eq:young_diagram_completely_symmetric}
	\tau_n^\nmodes
	\equiv
	(n, \underbrace{0, \dots, 0}_{\nmodes-1})
	=
	\begin{tikzpicture}
		[scale=.5,baseline={([yshift=-.5ex]current bounding box.center)}]
		\matrix (m) [matrix of math nodes,
		nodes={draw, minimum size=4mm, anchor=center},
		column sep=-\pgflinewidth,
		row sep=-\pgflinewidth
		]
		{
			\ & |[draw=none]|\dots & \ \\
		};
		\draw[BC] (m-1-3.south) -- node[below=2mm] {$n$} (m-1-1.south);
	\end{tikzpicture}
	,
\end{equation}
where the number of boxes has the interpretation of the number of particles in the system.
Formally, the Young diagram on the r.h.s. labels the maximally symmetric \ac{irrep} in $\SU(\nmodes)$.
Notably, the weights of the maximally symmetric irreps uniquely identify \ac{GT} basis elements, as it can be easily checked via the associated tableaux weights.

In $\H_n^\nmodes$, a common orthonormal basis is the \emph{Fock basis}, given by Fock states $\{\ket{\vec n} \mid \vec n \in \NN^\nmodes, \, \sum_{i=1}^\nmodes n_i = n\}$.
We remark that the \ac{GT} basis, as well as the Weyl basis, labels the same set of orthonormal vectors as the Fock basis.
In fact, for symmetric \acp{irrep}, $\vec n$ is exactly the tableau weight of the corresponding Young tableau, i.e. $n_i$ is the number of boxes filled with $i$ for each $i \in [\nmodes]$, e.g.
\begin{equation}
	\ket{3, 2, 1} =	\ketb{ \, \ytableaushort{1 1 1 2 2 3}*{6} \, } \, ,
\end{equation}
from which follows the correspondence with the \ac{GT} basis.
In particular, for any Fock state $\ket{\vec n} = \ket{n_1, \dots, n_\nmodes}$, the corresponding \ac{GT} pattern --that will be denoted by $N$ throughout this work-- is given as follows:
\begin{equation}
	\label{eq:GTpattern_Fock_state}
	N =
	\begin{pmatrix}
		n && 0 && \dots && 0 && 0 \\
		&\sum_{i=1}^{\nmodes-1} n_i && 0 & \dots & 0 && 0&\\
		&&\ddots&&\vdots&&\reflectbox{$\ddots$}&&\\
		&&& n_1 + n_2&&0&&&\\
		&&&& n_1&&&&
	\end{pmatrix} \, .
\end{equation}

The complex conjugate representation (dual representation) on $\H_n^\nmodes \cong (\H_n^\nmodes)^*$ is identified by the Young diagram
\begin{equation}
	\label{eq:young_diagram_dual_completely_symmetric}
	\bar \tau_n^\nmodes
	=
	\begin{tikzpicture}
		[scale=.5,baseline={([yshift=-.5ex]current bounding box.center)}]
		\matrix (m) [matrix of math nodes,
		nodes={draw, minimum size=4mm, anchor=center},
		column sep=-\pgflinewidth,
		row sep=-\pgflinewidth
		]
		{
			\ & |[draw=none]|\dots & \ \\
			|[draw=none]|\vdots & |[draw=none]|\dots & |[draw=none]|\vdots \\
			\ & |[draw=none]|\dots & \ \\
		};
		\draw[BC] (m-3-3.south) -- node[below=2mm] {$n$} (m-3-1.south);
		\draw[BC] (m-3-1.west) -- node[left=2mm] {$\nmodes-1$} (m-1-1.west);
	\end{tikzpicture}
	\, .
\end{equation}
Note that although $\bar \tau_n^\nmodes$ acts on $\H_n^\nmodes$, it is not symmetric (unless $m=2$).
We can construct the dual \ac{GT} basis labeled by \ac{GT} patterns $\bar N$ of the form (cf.~\cref{eq:GTpattern_dual,eq:GTpattern_Fock_state})
\begin{equation}
	\label{eq:GTpattern_Fock_state_dual}
	\bar N =
	\begin{pmatrix}
		n && n && \dots && n && 0 \\
		&n && n & \dots & n && n_\nmodes&\\
		&&\ddots&&\vdots&&\reflectbox{$\ddots$}&&\\
		&&& n && \sum_{i=3}^\nmodes n_i &&&\\
		&&&& \sum_{i=2}^\nmodes n_i &&&&
	\end{pmatrix} \, .
\end{equation}
By \cref{eq:conjugation_relative_phase}, the dual GT basis again coincides with the Fock basis of $\H_n^\nmodes$, but modified with the phase function \cref{eq:phase_dual_GTpattern}.

Finally, unlike the general case of arbitrary irreps of $\SU(\nmodes)$, the weight spaces of symmetric \acp{irrep} are clearly one-dimensional:
Each Fock state corresponds to one and only one Young tableau, as there are no degrees of freedom for box-labeling.
Weight spaces of dual symmetric \acp{irrep} have the same property by duality.

\subsubsection{Clebsch-Gordan coefficients.}
\label{sec:clebsch-gordan}

A crucial step for filtered \ac{RB} is the decomposition of the reference representation into \acp{irrep}.
In this section, we recap the role of the Clebsch-Gordan series for $\SU(\nmodes)$ in the decomposition of any tensor product representation, which will be employed in the remaining of this work.
This topic has been investigated extensively over the years due to its relevance in particle physics, so we refer to standard references such as \cite{sternbergGroupTheoryPhysics1994, cornwellGroupTheoryPhysics1984vol1, cornwellGroupTheoryPhysics1984vol2} for further details.

For two given irreps $\pi_1, \, \pi_2$ of $\SU(\nmodes)$, we consider the (completely reducible) tensor product representation $\pi_1 \otimes \pi_2\colon \SU(\nmodes) \rightarrow \U(\H_{\pi_1} \otimes \H_{\pi_2})$.
In the following, we will identify the carrier space $\H_{\pi}$ with $\pi$ for any \ac{irrep} of $\SU(\nmodes)$.
By the compact version of Maschke's theorem \cite[Thm. 5.2]{follandCourseAbstractHarmonic2015}, we have
\begin{equation}
\label{eq:maschke}
	\pi_1 \otimes \pi_2 \simeq \bigoplus_{\lambda} \lambda^{\oplus m_\lambda} \, ,
\end{equation}
where $m_\lambda$ is the multiplicity of $\lambda$ in $\pi_1 \otimes \pi_2$.
For $\SU(\nmodes)$, such decomposition can be computed in terms of Young diagrams with \emph{Littlewood-Richardson's rules} that we summarize in \cref{app:littlewood_richardson}.

In the context of second quantization, this decomposition can be interpreted as the generalization of the Clebsch-Gordan decomposition of sums of angular momenta in Quantum Mechanics \cite{alexNumericalAlgorithmExplicit2011, alexNonAbelianSymmetriesNumerical, mathurSUIrreducibleSchwinger2010, mathurInvariantsProjectionOperators2011}.
In particular, there exists a unitary matrix $\cg$ --here referred as the \emph{Clebsch-Gordan matrix}-- that realizes the isomorphism implicit in \cref{eq:maschke}: 
\begin{equation}
	\label{eq:cb_decomposition_matrix}
	\cg \left( \pi_1 \otimes \pi_2 \right) \cg^\ad = 
	\bigoplus_\lambda \lambda \otimes \1_{m_\lambda} \, ,
\end{equation}
where $\1_{m_{\lambda}}$ is the $m_\lambda \times m_\lambda$ identity matrix.
We explicitly define $C$ as the basis change matrix that takes the product \ac{GT} basis on the LHS of \cref{eq:cb_decomposition_matrix} to the union of \ac{GT} bases of every $\lambda$ on the RHS.
More precisely, for \ac{GT} patterns $M_1 \in \GT{\pi_1}, M_2 \in \GT{\pi_2}$, we have
\begin{equation}
	\ket{M_1, M_2} = \sum_{\lambda, r} \sum_{M \in \GT{\lambda}} C_{M_1, M_2}^{M, r} \ket{M, r} \, ,
\end{equation}
where $r \in [m_\lambda]$ denotes the $r$-th copy of $\lambda$ in $\pi_1 \otimes \pi_2$. \
The matrix coefficients $C_{M_1, M_2}^{M, r}$ of $\cg$ are called the \emph{Clebsch-Gordan coefficients}.
$\cg$ is uniquely defined up to global phases, and by convention it is chosen to be real.
Hence, we have the inverse transformation
\begin{equation}
	\ket{M, r} = \sum_{M_1 \in \GT{\pi_1}} \sum_{M_2 \in \GT{\pi_2}} C_{M_1, M_2}^{M, r} \ket{M_1, M_2} \, .
\end{equation}
By unitarity of $\cg$, the following orthogonality relations hold true:
\begin{align}
	\label{eq:CG_orthogonality}
	\sum_{\lambda, r} \sum_{M \in \GT{\lambda}} C_{M_1, M_2}^{M, r} C_{M_3, M_4}^{M, r}
	&=
	\delta_{M_1, M_3} \delta_{M_2, M_4} \, , \\
	\sum_{M_1 \in \GT{\pi_1}} \sum_{M_2 \in \GT{\pi_2}} C_{M_1, M_2}^{M, r} C_{M_1, M_2}^{M', r'} 
	&=
	\delta_{M, M'} \delta_{r, r'} \, .
\end{align}

As in the case of $\SU(2)$, selection rules for Clebsch-Gordan coefficients of $\SU(\nmodes)$ are available:
For \ac{GT} patterns $M_1 \in \GT{\pi_1}, M_2 \in \GT{\pi_2}, M \in \GT{\lambda}$, $C_{M_1, M_2}^{M, r} = 0$ if
\begin{equation}
	\label{eq:selection_rules_CGC}
	w_M \neq w_{M_1} + w_{M_2} \, ,
\end{equation}
where $w_{(\cdot)}$ is the weight defined in \cref{eq:weight_element}.

\subsection{Filtered randomized benchmarking}
\label{sec:filtered-RB-background}

\subsubsection{Filter functions.}
In \cref{sec:protocol}, we introduced the filter function, \cref{eq:filter_function}, to isolate and analyze the exponential decays associated with each \ac{irrep} of the reference representation $\omega_n^\nmodes$.
In this section, for the sake of completeness, we briefly motivate it in the bosonic case, before delving into the main technical results of this work.
For a comprehensive discussion, we refer to Ref.~\cite{heinrichRandomizedBenchmarkingRandom2023}.

In general, for an arbitrary compact group $G$ represented by $\omega$, an input state $\rho$ and a \ac{POVM} $\{ \ketbra{\vec x}{\vec x} \}_{\vec x \in \Omega}$, one defines the filter function for an irrep $\lambda\subset\omega$ as
\begin{equation}
	\label{eq:filter_function_general_POVM}
	f_\lambda(\vec x,  g) = \sandwich{\vec x}{P_\lambda \circ S^+ \circ \omega(g)^\ad (\rho)}{\vec x} \, ,
\end{equation}
where $S^+$ is the Moore-Penrose pseudo-inverse of the \emph{frame operator} $S$ defined as
\begin{equation}
	\label{eq:frame_operator}
	S
	\coloneqq
	\int_\Omega d\vec x \int_G dg \, \Tr[\omega(g)^\ad(\ketbra{\vec x}{\vec x}) (\cdot)] \,
	\omega(g) (\ketbra{\vec x}{\vec x})
	=
	\int_G dg \, \omega(g)^\ad \mathcal M \omega(g) \, ,
\end{equation}
where $\mathcal M = \int_{\Omega} d\vec x \, \Tr[\ketbra{\vec x}{\vec x} (\cdot)] \, \ketbra{\vec x}{\vec x}$ is the (possibly infinite dimensional) measurement channel associated with the \ac{POVM} $\{ \ketbra{\vec x}{\vec x} \}_{\vec x \in \Omega}$.

This choice of filter function is such that, in the ideal case of a noise-free, perfect implementation of the gates, the filtered \ac{RB} signal is of the form $F_{\lambda}(l) = \Tr[\rho P_\lambda \circ S^+ \circ S(\rho)]$, where $S^+ S$ coincides with the projector onto the span of the \ac{POVM}.
In particular, in the case of an informationally complete \ac{POVM}, $F_\lambda(l) = \Tr[\rho P_\lambda(\rho)]$, i.e.~the filtered signal is the overlap of $\rho$ with the filtering \ac{irrep}.

Then, \cref{eq:filter_function} follows from the following observation:
As the reference representation $\omega=\omega_n^\nmodes \coloneqq \tau_n^\nmodes (\cdot) \tau_n^{\nmodes \ad}$ of $G=\U(\nmodes)$ preserves the number of particles, $S$ acts non-trivially on the $n$-th Fock sector only, i.e.
\begin{equation}
	S = 
	\begin{pmatrix}
		\vec 0 &&&&& \\
		& \ddots &&&& \\
		&& \vec 0 &&& \\
		&&& S^{(n)} && \\
		&&&& \vec 0 & \\
		&&&&& \ddots 
	\end{pmatrix} \, ,
\end{equation}
where $S^{(n)}$ is obtained via the restriction of $\mathcal M$ to the subspace of $n$ particles.
Moreover, since $\omega_n^\nmodes$ decomposes as $\bigoplus_{k=0}^n \lambda_k$ (see \cref{lem:decomposition_conjugate_action}), Schur's lemma implies \cite{heinrichRandomizedBenchmarkingRandom2023}:
\begin{equation}
	\label{eq:frameop_irrep_subspace}
	S^{(n)} = \bigoplus_{\lambda} S^{(n)}_\lambda \, , \quad S_\lambda^{(n)} = s_\lambda \1_\lambda \, ,
\end{equation}
where the direct sum is over all \acp{irrep} of $\omega_n^\nmodes$ and, in general \cite{heinrichRandomizedBenchmarkingRandom2023},
\begin{equation}
	\label{eq:frameop_coefficients}
	s_\lambda = d_\lambda^{-1} \Tr[ P_\lambda \mathcal M]
	= d_\lambda^{-1} \int_\Omega d\vec x\,  \Tr\left[\ketbra{\vec x}{\vec x} P_\lambda(\ketbra{\vec x}{\vec x}) \right] \, .
\end{equation}
Here, $d_\lambda \equiv \dim \H_{\lambda}$, with $\lambda \in\{\lambda_k\}_{k=0}^n$, and $P_\lambda$ is the corresponding projector onto its carrier space.
In the second step, we used the fact that the Bochner integral commutes with the trace since the latter is a continuous linear operator in the trace norm and the trace of $[\ketbra{\vec x}{\vec x} P_\lambda(\ketbra{\vec x}{\vec x})]$ is finite.

\subsubsection{Non-uniform sampling}
\label{sec:non-unifrom-sampling}

The filtered \ac{RB} framework is not restricted to Haar-random sampling and was, in fact, intentionally designed for non-uniform sampling, e.g.~from generators of the considered group \cite{heinrichRandomizedBenchmarkingRandom2023}.
As this work builds on the filtered \ac{RB} framework, it is straightforward to replace the Haar-random sampling in the passive \ac{RB} protocol, \cref{sec:protocol}, with the sampling from some other probability measure $\nu$ on $\U(\nmodes)$.
For our results to uphold, it is sufficient that $\nu$ defines a random walk on $\U(\nmodes)$ that converges to the Haar measure $\mu_H$, that is $\nu^{\ast \seqlength}\rightarrow \mu_H$ with $\seqlength\rightarrow\infty$.
In fact, we only need that $\nu$ is suitably gapped w.r.t.~the reference representation $\omega_n^m$, this is that the \emph{moment operator}
\begin{equation}
 \mathsf{M}(\nu) \coloneqq  \int_{\U(\nmodes)} d\nu(g)  \, \omega_n^m(g)^\ad (\cdot) \omega_n^m(g) \,,
\end{equation}
has a \emph{spectral gap} $\Delta>0$ such that 
\begin{equation}
 \snorm{\mathsf{M}(\nu) - \mathsf{M}(\mu_H)} \leq 1 - \Delta\,.
\end{equation}
This is for instance the case if $\nu$ has support on generators of $\U(\nmodes)$.

Using such a gapped probability measure $\nu$ instead of the Haar measure, typically has some advantages in practice.
For instance, it might be easier to sample and to take some practical limitations into account (perhaps not all passive transformations are equally simple to implement).
In particular, one can simply generate passive transformations using random beam splitters and single-mode rotations.
Note that in such a case, the sequence length $\seqlength$ changes its interpretation from the number of passive transformation to be applied to the number of \emph{elementary} transformations.
The filtered \ac{RB} framework \cite{heinrichRandomizedBenchmarkingRandom2023} then predicts that the sequence of elementary transformations has to be sufficiently long in order to ``appear random enough'' and this threshold $\seqlength_0$ scales with the inverse spectral gap $\Delta^{-1}$.
In practice, $\seqlength_0$ seems to be reasonable small such that we expect that already very short random sequences can be used for passive \ac{RB}.
In particular, this should outperform the quadratic depth needed to decompose Haar-random passive transformations in terms of elementary transformations.

\subsection{Further notations}

As the Clebsch-Gordan decomposition is naturally related with the direct sum decomposition of an Hilbert space of the form $\H_1 \otimes \H_2$, it will be convenient to introduce a vectorized notation for operators and super-operators on $\H_n^\nmodes$.

We consider the basis of linear operators for $\bounded{\H_n^\nmodes}$ given by $\Phi = \{\ketbra{\vec n}{\vec m}\}_{\vec n, \vec m \in \NN^\nmodes}$ with $\sum_{i=1}^\nmodes n_i = \sum_{i=1}^\nmodes m_i = n$.
Any linear operator $A \in \bounded{\H_n^\nmodes}$ can be vectorized to an element $\ket{A} \in \H_n^\nmodes \otimes \H_n^\nmodes$ w.r.t.~$\Phi$ as
\begin{equation}
	\ket{A} = \sum_{\vec n, \vec m} \sandwich{\vec n}{A}{\vec m} \ket{\vec n, \vec m} \, , \quad \ket{\vec n, \vec m} \equiv \ket{\vec n} \otimes \ket{\vec m} \, .
\end{equation}
Under vectorization, we have the induced mapping $\tau_n^\nmodes (\cdot) \tau_n^{\nmodes \ad} \mapsto \tau_n^\nmodes \otimes \bar \tau_n^\nmodes$, where $\bar \tau_n^\nmodes$ denotes the complex conjugate representation of $\tau_n^\nmodes$ (in the Fock basis $\phi$).
Moreover, as long as it is clear from the context, we will not distinguish between super-operators and their corresponding quantities acting on $\H_1 \otimes \H_2$.

Then, the filter function defined in \cref{eq:filter_function} becomes
\begin{equation}
	f_\lambda(\vec n, g) = \frac{1}{s_\lambda} \sandwich{\ninput, \ninput}{P_\lambda (\tau_n^\nmodes \otimes \bar \tau_n^\nmodes)(g)^\ad}{\vec n, \vec n} \, ,
\end{equation}
where $\rho = \ketbra{\ninput}{\ninput} \cong \ket{\ninput, \ninput}$ is the input state and $s_{\lambda} = \left(d_\lambda\right)^{-1} \sum_{\vec n \in \NN^\nmodes} \sandwich{\vec n, \vec n}{P_\lambda}{\vec n, \vec n}$.

\section{Technical details and proofs}
\label{sec:passive_RB_Fock}

Here, we provide proofs for the theorems introduced in \cref{sec:results}:
In \cref{sec:filter_function_PNR} we prove \cref{thm:general_filter_informal} and in \cref{sec:moments_PNR} we prove \cref{thm:variance_bound_informal} based on notation and technical tools introduced in \cref{sec:preliminaries}.
Throughout this section we assume that the number of modes is $\nmodes\geq 2$.

\subsection{Clebsch-Gordan decomposition of the reference representation}
\label{sec:ref_repr}

In this section, we study the \ac{irrep} decomposition of $\omega_n^\nmodes$.
From \cref{sec:preliminaries}, we have
\begin{equation}
	\omega_n^\nmodes = \tau_n^{\nmodes} (\argdot) \tau_n^{\nmodes \ad} \cong \tau_n^\nmodes \otimes \bar \tau_n^\nmodes \, ,
\end{equation}
where again complex conjugate representation is taken in the Fock basis.
We can restrict our focus to the \acp{irrep} of $\SU(\nmodes)$ (or, equivalently, its corresponding Lie algebra $\mathfrak{su}(\nmodes)$) as $\tau_n^\nmodes$ can be extended to \acp{irrep} of $\U(\nmodes)$ using nontrivial characters of the unit circle group (roughly speaking, resulting in a multiplication by a global phase which vanishes in the conjugate action of $\mathcal B(\H_n^\nmodes)$).
The decomposition of $\omega_n^\nmodes$ into \acp{irrep} can be computed using \emph{Littlewood-Richardson's rules}, a general tool to classify the decomposition of tensor product representations.
We refer to \cref{app:littlewood_richardson} for a brief overview on how they can be employed in the context of $\SU(\nmodes)$.
\begin{lemma}
	\label{lem:decomposition_conjugate_action}
	Let $\tau_n^\nmodes : \SU(\nmodes) \rightarrow \U(\H_n^\nmodes)$ be the irreducible representation of $\SU(\nmodes)$ on the space of $n$ bosons distributed over $\nmodes$ modes as in \cref{eq:young_diagram_completely_symmetric}.
	Define the Young diagram
	\begin{equation}
		\label{eq:lambda_diagram}
		\lambda_k \equiv
		\begin{tikzpicture}
			[scale=.5,baseline={([yshift=-.5ex]current bounding box.center)}]
			\matrix (m) [matrix of math nodes,
			nodes={draw, minimum size=4mm, anchor=center},
			column sep=-\pgflinewidth,
			row sep=-\pgflinewidth
			]
			{
				\ & |[draw=none]|\dots & \ & \ & |[draw=none]|\dots & \ \\
				|[draw=none]|\vdots & |[draw=none]|\ddots & |[draw=none]|\vdots \\
				\ & |[draw=none]|\dots & \ \\
			};
			\draw[BC] (m-1-1.north) -- node[above=2mm] {$k$} (m-1-3.north);
			\draw[BC] (m-1-4.north) -- node[above=2mm] {$k$} (m-1-6.north);
			\draw[BC] (m-3-1.west) -- node[left=2mm] {$\nmodes-1$} (m-1-1.west);
		\end{tikzpicture} \, ,
	\end{equation}
	where $\lambda_0$  and $\lambda_1$ denote the trivial \ac{irrep} and the adjoint \ac{irrep} of 
	$\SU(\nmodes)$, 
	respectively.
	Then, for any $n, \nmodes \in \NN \setminus \{0\}$,
	\begin{equation}
		\label{eq:decomposition_conjugate_action}
		\omega_n^\nmodes = \bigoplus_{k=0}^n \lambda_k \, ,
	\end{equation}
	where each $\lambda_k, \, k = 0, \dots, n$, appears exactly one time.
\end{lemma}
\begin{proof}
	We will prove the following equivalent fact by induction:
	\begin{equation}
		\label{eq:conjugate_action_decomposition_recursive}
		\omega_n^\nmodes = \lambda_n \oplus \omega_{n-1}^\nmodes \, , \quad \forall n \in \NN \setminus \{0\} \, .
	\end{equation}
	First, notice that 
	\begin{equation}
		\omega_1^\nmodes = 
		\begin{tikzpicture}
			[scale=.5,baseline={([yshift=-.5ex]current bounding box.center)}]
			\matrix (m) [matrix of math nodes,
			nodes={draw, minimum size=4mm, anchor=center},
			column sep=-\pgflinewidth,
			row sep=-\pgflinewidth
			]
			{
				\ & \ \\
				|[draw=none]|\vdots \\
				\ \\
			};
			\draw[BC] (m-3-1.west) -- node[left=2mm] {$\nmodes-1$} (m-1-1.west);
		\end{tikzpicture}
		\oplus 
		\mathbf{1}
		\equiv
		\lambda_1 \oplus \omega_0^\nmodes \, ,
	\end{equation}
	as $\omega_0^\nmodes = \mathbf{1}$ trivially.
	
	The conjugate representation is associated with the tensor product of Young diagrams
	\begin{equation}
		\begin{tikzpicture}
			[scale=.5,baseline={([yshift=-.5ex]current bounding box.center)}]
			\matrix (m) [matrix of math nodes,
			nodes={draw, minimum size=4mm, anchor=center},
			column sep=-\pgflinewidth,
			row sep=-\pgflinewidth
			]
			{
				\ & |[draw=none]|\dots & \ \\
				|[draw=none]|\vdots & |[draw=none]|\ddots & |[draw=none]|\vdots \\
				\ & |[draw=none]|\dots & \ \\
			};
			\draw[BC] (m-1-1.north) -- node[above=2mm] {$n$} (m-1-3.north);
			\draw[BC] (m-3-1.west) -- node[left=2mm] {$\nmodes-1$} (m-1-1.west);
		\end{tikzpicture}
		\otimes
		\begin{tikzpicture}
			[scale=.5,baseline={([yshift=1ex]current bounding box.center)}]
			\matrix (m) [matrix of math nodes,
			nodes={draw, minimum size=3mm, anchor=center},
			column sep=-\pgflinewidth,
			row sep=-\pgflinewidth
			]
			{
				\ & |[draw=none]|\dots & \ \\
			};
			\draw[BC] (m-1-3.south) -- node[below=2mm] {$n$} (m-1-1.south);
		\end{tikzpicture}
	\end{equation}
	(swapping tensor factors does not influence the result).
	By Littlewood-Richardson's rules, we first have
	\begin{equation}
		\begin{aligned}
			\begin{tikzpicture}
				[scale=.5,baseline={([yshift=-.5ex]current bounding box.center)}]
				\matrix (m) [matrix of math nodes,
				nodes={draw, minimum size=4mm, anchor=center},
				column sep=-\pgflinewidth,
				row sep=-\pgflinewidth
				]
				{
					\ & |[draw=none]|\dots & \ \\
					|[draw=none]|\vdots & |[draw=none]|\ddots & |[draw=none]|\vdots \\
					\ & |[draw=none]|\dots & \ \\
				};
				\draw[BC] (m-1-1.north) -- node[above=2mm] {$n$} (m-1-3.north);
				\draw[BC] (m-3-1.west) -- node[left=2mm] {$\nmodes-1$} (m-1-1.west);
			\end{tikzpicture}
			\otimes
			\begin{tikzpicture}
				[scale=.5,baseline={([yshift=-.5ex]current bounding box.center)}]
				\matrix (m) [matrix of math nodes,
				nodes={draw, minimum size=3mm, anchor=center},
				column sep=-\pgflinewidth,
				row sep=-\pgflinewidth
				]
				{
					a & |[draw=none]|\dots & a \\
				};
				\draw[BC] (m-1-1.north) -- node[above=2mm] {$n$} (m-1-3.north);
			\end{tikzpicture}
			& =
			\begin{tikzpicture}
				[scale=.5,baseline={([yshift=-.5ex]current bounding box.center)}]
				\matrix (m) [matrix of math nodes,
				nodes={draw, minimum size=4mm, anchor=center},
				column sep=-\pgflinewidth,
				row sep=-\pgflinewidth
				]
				{
					\ & |[draw=none]|\dots & \ & a \\
					|[draw=none]|\vdots & |[draw=none]|\ddots & |[draw=none]|\vdots \\
					\ & |[draw=none]|\dots & \ \\
				};
				\draw[BC] (m-1-1.north) -- node[above=2mm] {$n$} (m-1-3.north);
				\draw[BC] (m-3-1.west) -- node[left=2mm] {$\nmodes-1$} (m-1-1.west);
			\end{tikzpicture}
			\otimes
			\begin{tikzpicture}
				[scale=.5,baseline={([yshift=-.5ex]current bounding box.center)}]
				\matrix (m) [matrix of math nodes,
				nodes={draw, minimum size=4mm, anchor=center},
				column sep=-\pgflinewidth,
				row sep=-\pgflinewidth
				]
				{
					a & |[draw=none]|\dots & a \\
				};
				\draw[BC] (m-1-1.north) -- node[above=2mm] {$n-1$} (m-1-3.north);
			\end{tikzpicture} \\
			& \oplus
			\begin{tikzpicture}
				[scale=.5,baseline={([yshift=-.5ex]current bounding box.center)}]
				\matrix (m) [matrix of math nodes,
				nodes={draw, minimum size=4mm, anchor=center},
				column sep=-\pgflinewidth,
				row sep=-\pgflinewidth
				]
				{
					\ & |[draw=none]|\dots & \ \\
					|[draw=none]|\vdots & |[draw=none]|\ddots & |[draw=none]|\vdots \\
					\ & |[draw=none]|\dots & \ \\
				};
				\draw[BC] (m-1-1.north) -- node[above=2mm] {$n-1$} (m-1-3.north);
				\draw[BC] (m-3-1.west) -- node[left=2mm] {$\nmodes-1$} (m-1-1.west);
			\end{tikzpicture}
			\otimes
			\begin{tikzpicture}
				[scale=.5,baseline={([yshift=-.5ex]current bounding box.center)}]
				\matrix (m) [matrix of math nodes,
				nodes={draw, minimum size=4mm, anchor=center},
				column sep=-\pgflinewidth,
				row sep=-\pgflinewidth
				]
				{
					a & |[draw=none]|\dots & a \\
				};
				\draw[BC] (m-1-1.north) -- node[above=2mm] {$n-1$} (m-1-3.north);
			\end{tikzpicture}
			\, .
		\end{aligned}
	\end{equation}
	Notice that the second term in the r.h.s. is by definition $\omega_{n-1}^\nmodes$.
	Hence, we shall only prove that
	\begin{equation}
		\label{eq:decomposition_lemma_inductive_step}
		\begin{tikzpicture}
			[scale=.5,baseline={([yshift=-.5ex]current bounding box.center)}]
			\matrix (m) [matrix of math nodes,
			nodes={draw, minimum size=4mm, anchor=center},
			column sep=-\pgflinewidth,
			row sep=-\pgflinewidth
			]
			{
				\ & |[draw=none]|\dots & \ & a \\
				|[draw=none]|\vdots & |[draw=none]|\ddots & |[draw=none]|\vdots \\
				\ & |[draw=none]|\dots & \ \\
			};
			\draw[BC] (m-1-1.north) -- node[above=2mm] {$n$} (m-1-3.north);
			\draw[BC] (m-3-1.west) -- node[left=2mm] {$\nmodes-1$} (m-1-1.west);
		\end{tikzpicture}
		\otimes
		\begin{tikzpicture}
			[scale=.5,baseline={([yshift=-.5ex]current bounding box.center)}]
			\matrix (m) [matrix of math nodes,
			nodes={draw, minimum size=4mm, anchor=center},
			column sep=-\pgflinewidth,
			row sep=-\pgflinewidth
			]
			{
				a & |[draw=none]|\dots & a \\
			};
			\draw[BC] (m-1-1.north) -- node[above=2mm] {$n-1$} (m-1-3.north);
		\end{tikzpicture}
		=
		\oplus_{k=0}^{n} \lambda_k 
		\, .
	\end{equation}
	For this purpose, let us consider the factor
	\begin{equation}
		\tilde \lambda_r^{(s)}
		\coloneqq
		\begin{tikzpicture}
			[scale=.5,baseline={([yshift=-.5ex]current bounding box.center)}]
			\matrix (m) [matrix of math nodes,
			nodes={draw, minimum size=4mm, anchor=center},
			column sep=-\pgflinewidth,
			row sep=-\pgflinewidth
			]
			{
				\ & |[draw=none]|\dots & \ & a & |[draw=none]|\dots & a \\
				|[draw=none]|\vdots & |[draw=none]|\ddots & |[draw=none]|\vdots \\
				\ & |[draw=none]|\dots & \ \\
			};
			\draw[BC] (m-1-1.north) -- node[above=2mm] {$r$} (m-1-3.north);
			\draw[BC] (m-1-4.north) -- node[above=2mm] {$s$} (m-1-6.north);
			\draw[BC] (m-3-1.west) -- node[left=2mm] {$\nmodes-1$} (m-1-1.west);
		\end{tikzpicture}
		\, .
	\end{equation}
	Clearly, $\tilde \lambda_r^{(r)} = \lambda_r$ and $\tilde \lambda_0^{(0)} = \mathbf{1}$.
	Notice that
	\begin{equation}
		\begin{aligned}
			\tilde \lambda_r^{(s)} \otimes \tau_l^{\nmodes}
			&=
			\left(
			\begin{tikzpicture}
				[scale=.5,baseline={([yshift=-.5ex]current bounding box.center)}]
				\matrix (m) [matrix of math nodes,
				nodes={draw, minimum size=4mm, anchor=center},
				column sep=-\pgflinewidth,
				row sep=-\pgflinewidth
				]
				{
					\ & |[draw=none]|\dots & \ & a & |[draw=none]|\dots & a\\
					|[draw=none]|\vdots & |[draw=none]|\ddots & |[draw=none]|\vdots \\
					\ & |[draw=none]|\dots & \ \\
				};
				\draw[BC] (m-1-1.north) -- node[above=2mm] {$r$} (m-1-3.north);
				\draw[BC] (m-1-4.north) -- node[above=2mm] {$s+1$} (m-1-6.north);
				\draw[BC] (m-3-1.west) -- node[left=2mm] {$\nmodes-1$} (m-1-1.west);
			\end{tikzpicture}
			\oplus
			\begin{tikzpicture}
				[scale=.5,baseline={([yshift=-.5ex]current bounding box.center)}]
				\matrix (m) [matrix of math nodes,
				nodes={draw, minimum size=4mm, anchor=center},
				column sep=-\pgflinewidth,
				row sep=-\pgflinewidth
				]
				{
					\ & |[draw=none]|\dots & \ & a & |[draw=none]|\dots & a\\
					|[draw=none]|\vdots & |[draw=none]|\ddots & |[draw=none]|\vdots \\
					\ & |[draw=none]|\dots & \ \\
				};
				\draw[BC] (m-1-1.north) -- node[above=2mm] {$r-1$} (m-1-3.north);
				\draw[BC] (m-1-4.north) -- node[above=2mm] {$s$} (m-1-6.north);
			\end{tikzpicture}
			\right)
			\otimes
			\begin{tikzpicture}
				[scale=.5,baseline={([yshift=-.5ex]current bounding box.center)}]
				\matrix (m) [matrix of math nodes,
				nodes={draw, minimum size=4mm, anchor=center},
				column sep=-\pgflinewidth,
				row sep=-\pgflinewidth
				]
				{
					a & |[draw=none]|\dots & a \\
				};
				\draw[BC] (m-1-1.north) -- node[above=2mm] {$l-1$} (m-1-3.north);
			\end{tikzpicture}
			\\
			&=
			\left( \tilde \lambda_r^{(s+1)} \oplus \tilde \lambda_{r-1}^{(s)} \right) \otimes \tau_{l-1}^{\nmodes}
			\, .
		\end{aligned}
	\end{equation}
	With this notation, expanding the l.h.s. of \cref{eq:decomposition_lemma_inductive_step} we get
	\begin{equation}
		\begin{aligned}
			\tilde \lambda_n^{(1)} \otimes \tau_{n-1}^\nmodes
			&=
			\left(
			\begin{tikzpicture}
				[scale=.5,baseline={([yshift=-.5ex]current bounding box.center)}]
				\matrix (m) [matrix of math nodes,
				nodes={draw, minimum size=4mm, anchor=center},
				column sep=-\pgflinewidth,
				row sep=-\pgflinewidth
				]
				{
					\ & |[draw=none]|\dots & \ & a & a\\
					|[draw=none]|\vdots & |[draw=none]|\ddots & |[draw=none]|\vdots \\
					\ & |[draw=none]|\dots & \ \\
				};
				\draw[BC] (m-1-1.north) -- node[above=2mm] {$n$} (m-1-3.north);
				\draw[BC] (m-3-1.west) -- node[left=2mm] {$\nmodes-1$} (m-1-1.west);
			\end{tikzpicture}
			\oplus
			\begin{tikzpicture}
				[scale=.5,baseline={([yshift=-.5ex]current bounding box.center)}]
				\matrix (m) [matrix of math nodes,
				nodes={draw, minimum size=4mm, anchor=center},
				column sep=-\pgflinewidth,
				row sep=-\pgflinewidth
				]
				{
					\ & |[draw=none]|\dots & \ & a \\
					|[draw=none]|\vdots & |[draw=none]|\ddots & |[draw=none]|\vdots \\
					\ & |[draw=none]|\dots & \ \\
				};
				\draw[BC] (m-1-1.north) -- node[above=2mm] {$n-1$} (m-1-3.north);
			\end{tikzpicture}
			\right)
			\otimes
			\begin{tikzpicture}
				[scale=.5,baseline={([yshift=-.5ex]current bounding box.center)}]
				\matrix (m) [matrix of math nodes,
				nodes={draw, minimum size=4mm, anchor=center},
				column sep=-\pgflinewidth,
				row sep=-\pgflinewidth
				]
				{
					a & |[draw=none]|\dots & a \\
				};
				\draw[BC] (m-1-1.north) -- node[above=2mm] {$n-2$} (m-1-3.north);
			\end{tikzpicture}
			\\
			&=
			\left( \tilde \lambda_n^{(2)} \oplus \tilde \lambda_{n-1}^{(1)} \right) \otimes \tau_{n-2}^\nmodes \\
			&=
			\left( \tilde \lambda_n^{(3)} \oplus \tilde \lambda_{n-1}^{(2)} \oplus \tilde \lambda_{n-1}^{(2)} \oplus \tilde \lambda_{n-2}^{(1)} \right) \otimes \tau_{n-3}^\nmodes \\
			&=
			\left( \tilde \lambda_n^{(3)} \oplus \tilde \lambda_{n-1}^{(2)} \oplus \tilde \lambda_{n-2}^{(1)} \right) \otimes \tau_{n-3}^\nmodes \\
			& \ \, \vdots
			\\
			&=
			\left( \tilde \lambda_n^{(i)} \oplus \tilde \lambda_{n-1}^{(i-1)} \oplus \dots \oplus \tilde \lambda_{n-i+1}^{(1)} \right) \otimes \tau_{n-i}^\nmodes \\
			& \ \, \vdots
			\\
			&=
			\bigoplus_{i=0}^n \tilde \lambda_{n-i}^{(n-i)} \\
			&=
			\bigoplus_{k=0}^n \lambda_k
			\, .
		\end{aligned}
	\end{equation}
	In the latter, we used the merging rule for Young diagrams, see \cref{app:littlewood_richardson}.
	Finally, we have
	\begin{equation}
		\begin{aligned}
			\omega_n^\nmodes 
			=
			\sum_{k=0}^n \lambda_k \oplus \omega_{n-1}^\nmodes
			=
			\sum_{k=0}^n \lambda_k \oplus \sum_{l=0}^{n-1} \lambda_l
			=
			\lambda_n \oplus \omega_{n-1}^\nmodes
			\, ,
		\end{aligned}
	\end{equation}
	by the merging rule again.
\end{proof}
For instance, we have the following explicit decomposition for $n = \nmodes = 3$:
\begin{equation}
	\omega_3^3 =
	\mathbf{1} \oplus \ydiagram{2, 1} \oplus \ydiagram{4, 2} \oplus \ydiagram{6, 3} \,,
\end{equation}
since:
\begin{align}
	\ydiagram{3,3} \otimes \ytableaushort{a a a}*{3}
	&= 
	\left(
	\ytableaushort{\none \none \none a, \none \none \none}*{4,3}
	\oplus
	\ydiagram{2,2}
	\right) \otimes \ytableaushort{a a}*{2} \\
	& = 
	\left(
	\ytableaushort{\none \none \none a a, \none \none \none}*{5,3}
	\oplus
	\ytableaushort{\none \none a, \none \none}*{3,2}
	\oplus
	\ydiagram{1,1}
	\right) \otimes \ytableaushort{a}*{1} \\
	& = 
	\ytableaushort{\none \none \none a a a, \none \none \none}*{6,3}
	\oplus
	\ytableaushort{\none \none a a, \none \none}*{4,2}
	\oplus 
	\ytableaushort{\none a, \none}*{2,1}
	\oplus
	\mathbf{1} \,.
\end{align}
The dimension of $\lambda_k$ admits a nice closed-form expression in terms of the dimension of the number subspace,
\begin{equation}
	\label{eq:dim_number_space}
	\dim \H_k^\nmodes = \binom{k+\nmodes-1}{k}\,,
\end{equation}
as follows:
\begin{proposition}
	\label{prop:dimension_lambda_k}
	For any $k \in \NN$, set $d_{\lambda_k} \equiv \dim \lambda_k$.
	Then, the following holds:
	\begin{equation}
		\label{eq:dimension_lambda_k}
		d_{\lambda_k} = \left( 1 - \frac{k^2}{(k+\nmodes-1)^2} \right) \left( \dim \H_k^\nmodes \right)^2
		= \frac{2k+\nmodes-1}{\nmodes-1} \binom{k+\nmodes-2}{k}^2\,.
	\end{equation}
\end{proposition}
We prove this fact in \cref{app:proof_dim_irrep}.

We remark that for the one-to-one correspondence of \ac{RB} decays with irreps, it is important that the reference representation $\omega_n^\nmodes$ decomposes into \emph{multiplicity-free} \acp{irrep} \cite{helsenGeneralFrameworkRandomized2020,heinrichRandomizedBenchmarkingRandom2023}.
Moreover, these irreps should be of \emph{real type}.
This is the case for the \acp{irrep} $\lambda_{k}, k = 1, \dots, n$, because each term in the decomposition is self-dual and multiplicity free (this can also be checked using e.g. \cite[Proposition 26.24]{fultonRepresentationTheory2004}, where all the complex, real, and quaternionic irreps of $\SU(\nmodes)$ are classified).
In general, both conditions are not necessarily fulfilled, which can lead to more than one decay per irrep that may also be complex (i.e.~oscillating).

\subsection{Filter function for passive RB with PNR measurements}
\label{sec:filter_function_PNR}

As the \ac{irrep} decomposition of $\omega_n^\nmodes$ can be computed for any $n$ and $\nmodes$ (cf. \cref{lem:decomposition_conjugate_action}), we can evaluate explicit expressions for the filter function.
This will provide the proof of \cref{thm:general_filter_informal}.

By construction, $\omega_n^\nmodes = \tau_n^\nmodes (\cdot) \tau_n^{\nmodes \ad} \cong \tau_n^\nmodes \otimes \bar \tau_n^\nmodes$ acts on elements $\ket{\vec n, \vec n}$ (from here on referred as the uncoupled basis). 
However, as pointed out in the \cref{sec:GTpattern}, the second entry shall be suitably interpreted as a basis element of $\bar \tau_n^\nmodes$, which requires the specification of the relative phases between states of $\tau_n^\nmodes$ and its dual.
In particular, we have
\begin{equation}
	\ket{\vec n} = \ket{N} = (-1)^{\varphi(N)} \ket{\bar N} \, ,
\end{equation}
where $\ket{N}$ is the \ac{GT} pattern defined in \cref{eq:GTpattern_Fock_state}, $\ket{\bar N}$ is the dual state, cf. \cref{eq:GTpattern_dual}, and the relative phase is given in \cref{eq:phase_dual_GTpattern}.
For the rest of this work, we fix the following notation:
For each Fock state $\ket{\vec m}$, the corresponding \ac{GT} pattern will be denoted with the corresponding Latin capital letter, namely the identification $\ket{\vec m} \equiv \ket{M}$ holds.

From previous considerations, it follows
\begin{equation}
	\label{eq:cg_decomposition_vector}
	\rho \cong \ket{\ninput, \ninput} = \ket{\ninputGT, \ninputGT} = (-1)^{\varphi(\ninputGT)} \ket{\ninputGT, \ninputGTdual}
	= 
	(-1)^{\varphi(\ninputGT)} \sum_{k=0}^n \sum_{M \in \GT{\lambda_k}} C_{\ninputGT, \ninputGTdual}^{M}  \, \ket{M} \, ,
\end{equation}
with $\{\ket{M}\}_{M \in \GT{\lambda_k}}$ being an orthonormal basis (referred as the coupled basis from here on) for the carrier space of $\lambda_k$.
Notice that in \cref{eq:cg_decomposition_vector} we do not need to specify the multiplicity index of states and Clebsch-Gordan coefficients since $\lambda_k$ is multiplicity free for each $k = 0, \dots, n$.

Then, for a fixed \ac{irrep} $\lambda \in \{\lambda_k\}_{k=0}^n$, we have $P_\lambda = \sum_{M \in \GT{\lambda}} \ketbra{M}{M}$ and the following relation holds true:
\begin{equation}
	\label{eq:projection_fock_state}
	P_\lambda \ket{\ninputGT, \bar \ninputGT}
	=
	\sum_{M \in \GT{\lambda}} C_{\ninputGT, \ninputGTdual}^{M} \ket{M} \, ,
\end{equation}
We remark that -- by the selection rules of Clebsch-Gordan coefficients \cref{eq:selection_rules_CGC} -- the sum is restricted to all the basis vectors such that the associated weight corresponds to the sum of the weights of the states $\ket{\ninputGT}$ and $\ket{\ninputGTdual}$, Cf. \cref{eq:selection_rules_CGC}.
Specifically, by \cref{eq:GTpattern_dual}, $C_{\ninputGT, \ninputGTdual}^{M}$ is possibly nonzero provided that
\begin{align}
	\label{eq:weight_sum_GT+GTdual}
	w_j^{(M)} 
	&=
	w_j^{(\ninputGT)} + w_j^{(\ninputGTdual)} \nonumber \\
	&=
	2 \sum_{i=1}^{j} N_{i,j}^{(0)} - \left( \sum_{i=1}^{j-1} N_{i,j-1}^{(0)} + \sum_{i=1}^{j+1} N_{i,j+1}^{(0)} \right)
	+
	2 \sum_{i=1}^{j} \bar N_{i,j}^{(0)} - \left( \sum_{i=1}^{j-1} \bar N_{i,j-1}^{(0)} + \sum_{i=1}^{j+1} \bar N_{i,j+1}^{(0)} \right) \nonumber \\
	&=
	2 \sum_{i=1}^{j} N_{1,\nmodes}^{(0)} - \left( \sum_{i=1}^{j-1} N_{1,\nmodes}^{(0)} + \sum_{i=1}^{j+1} N_{1,\nmodes}^{(0)} \right) \nonumber \\
	&=
	0
\end{align}
for each $j = 1, \dots, \nmodes$, i.e.~$w_{\ninputGT} + w_{\ninputGTdual} = \vec 0$.

Moreover, for the $\vec 0$ weight spaces, the inner multiplicity $\gamma_{\lambda_k}(\vec 0)$ of $\vec 0$ in $\lambda_k$ can be easily computed, as we prove the following fact:
\begin{lemma}
	\label{lem:weight_space_zero}
	Let $\lambda_k$ be an \ac{irrep} of $\SU(\nmodes)$ as in \cref{eq:lambda_diagram} for any $k \in \NN$.
	Then,
	\begin{equation}
		\label{eq:weight_space_zero}
		\gamma_{\lambda_k}(\vec 0) = \binom{k + \nmodes - 2}{k} \, .
	\end{equation}
\end{lemma}
We prove this fact in \cref{app:weight_space_zero}.
This provides the number of (possibly) non-zero terms in \cref{eq:projection_fock_state} and is relevant for the evaluation of the eigenvalues of the frame operator of passive \ac{RB} with \ac{PNR} measurements, for which we obtain the following simple expression:
\begin{theorem}
	\label{thm:frame_operator_pnr}
	Let $\lambda_k$, $k=0, \dots, n$ be an \ac{irrep} in $\omega_n^{\nmodes}$.
	For a \ac{PNR} measurement setting, the eigenvalues of the frame operator of the passive \ac{RB} protocol are given by
	\begin{equation}
		\label{eq:frameop_pnr}
		s_{\lambda_k} = \frac{\nmodes-1}{2k+\nmodes-1} \gamma_{\lambda_k}(\vec 0)^{-1} \, ,
	\end{equation}
	where $\gamma_{\lambda_k}(\vec 0)$ is as in \cref{eq:weight_space_zero}.
\end{theorem}
\begin{proof}
	First, recall that a single mode (ideal) \ac{PNR} detector measures the number of particles in such mode \cite{provaznikBenchmarkingPhotonNumber2020}.
	In the case of $\nmodes$ modes, the (ideal) \ac{POVM} is therefore given by $\{\ketbra{\mathbf n}{\mathbf n}\}_{\mathbf n \in \NN^\nmodes}$
	and the (vectorized) measurement channel can be written as
	\begin{equation}
		\label{eq:measurement_channel_PNR}
		\mathcal M \coloneqq \sum_{\mathbf n \in \NN^\nmodes} \ketbra{\vec n, \vec n}{\vec n, \vec n} = \sum_{n = 0}^\infty \sum_{n \in \GT{\tau_n^\nmodes}} \ketbra{N, N}{N, N} \, .
	\end{equation}
	Denoting by $P_{\lambda_k}$ the projector onto ${\lambda_k} \in \hat \omega_n^\nmodes$, we have the following:
	\begin{align}
		s_{\lambda_k}
		&=
		\frac{1}{d_{\lambda_k}} \sum_{\vec n \in \NN^\nmodes} \sandwich{\vec n, \vec n}{P_{\lambda_k}}{\vec n, \vec n} \\
		&=
		\frac{1}{d_{\lambda_k}} \sum_{N \in \GT{\tau_n^\nmodes}} \sandwich{N, N}{P_{\lambda_k}}{N, N} \\
		&=
		\frac{1}{d_{\lambda_k}} \sum_{N \in \GT{\tau_n^\nmodes}} (-1)^{2\varphi(N)} \sandwich{N, \bar N}{P_{\lambda_k}}{N, \bar N} \\
		&=
		\frac{1}{d_{\lambda_k}} \sum_{N \in \GT{\tau_n^\nmodes}} \sum_{M \in \GT{\lambda_k}} \lvert C_{N, \bar N}^{M} \rvert^2 \, ,
	\end{align}
	where in the second step we used the fact that $P_{\lambda_k}$ acts non-trivially on the $n$ particle subspace, and in the fourth step we used the fact that the phases $\varphi(N)$ introduced in the labeling of \ac{GT} dual patterns are integers by construction, see \cref{eq:phase_dual_GTpattern}.
	Finally, recall that the non-trivial Clebsch-Gordan coefficients are determined by the selection rule $w_{M} = w_{N} + w_{\bar N} = 0$.
	Moreover, symmetric \acp{irrep} of $\SU(\nmodes)$ --and their dual-- preserve the 1-to-1 correspondence between weights and basis vectors, see \cref{sec:symmetric_irrep}.
	This implies
	\begin{equation}
		\label{eq:CGCs_symmetric_orthogonality}
		\sum_{N \in \GT{\tau_n^\nmodes}} \abs{C_{N, \bar N}^{M}}^2
		=
		\sum_{N\in \GT{\tau_n^\nmodes}} \sum_{\bar N' \in \GT{\bar \tau_n^\nmodes}}
		C_{N, \bar N'}^{M} C_{N, \bar N'}^{M}
		=
		1
	\end{equation}
	by orthogonality relations, cf. \cref{eq:CG_orthogonality}.
	Therefore, 
	\begin{equation}
		s_{\lambda_k} = \frac{\gamma_{\lambda_k}(\vec 0)}{d_{\lambda_k}}
	\end{equation}
	and the assertion follows from \cref{eq:weight_space_zero,prop:dimension_lambda_k}.
\end{proof}

Notice that $s_{\lambda_k}$ scales exponentially in both $k$ and $\nmodes$.
In particular, in the case $k = n = \nmodes$, it is proportional to the $\nmodes$-th Catalan number, and it scales as $\LandauO(4^\nmodes/\sqrt{\nmodes})$ for large values of $\nmodes$.

\begin{theorem}[Restatement of \cref{thm:general_filter_informal} - \ac{PNR} version]
	\label{thm:filter_function_PNR}
	Let $\rho = \ketbra{\ninput}{\ninput} \cong \ket{\ninput, \ninput} = \ket{\ninputGT, \ninputGT}$ be a $\nmodes$ modes state and let $\{ \ketbra{\vec n}{\vec n} \}_{\vec n \in \NN^\nmodes}$ be the Fock state \ac{POVM}.
	Let $N$ be the \ac{GT} pattern associated with $\vec n$.
	Then, for a given \ac{irrep} $\lambda_k \in$ in $\omega_n^\nmodes$, and assuming $s_{\lambda_k} \neq 0$,
	\begin{equation}
		\label{eq:filter_function_PNR}
		f_{\lambda_k}(\vec n, g) 
		=
		\frac{1}{s_{\lambda_k}}
		\frac{(-1)^{\varphi(\ninputGT)}}{\vec n!}
		\sum_{M \in \GT{\lambda_k}} C_{\ninputGT, \ninputGTdual}^{M} \sum_{N' \in \GT{\tau_n^\nmodes}} \frac{(-1)^{\varphi(N')}}{\vec n'!} C_{N', \bar N'}^{M} \abs{ \per(g_{\vec n, \vec n'}) }^2 
	\end{equation}
	where $s_{\lambda_k}$ is evaluated in \cref{thm:frame_operator_pnr}, the sums are restricted to all basis states such that \cref{eq:weight_sum_GT+GTdual} is satisfied and $\per(g_{\vec n, \vec m})$ denotes the permanent of the matrix obtained by $g$ by taking $m_j$ copies of the $j$-th column of $U$ and then by taking $n_i$ copies of the $i$-th row of the resulting matrix, and we used the multi-index notation $\vec n! \coloneqq n_1! \dots n_\nmodes!$.
\end{theorem}
\begin{proof}
	By \cref{eq:cg_decomposition_vector}, and denoting by $\ninputGT$ and $N$ the \ac{GT} patterns associated with $\ninput$ and $\vec n$, respectively, the filter function defined in \cref{eq:filter_function} becomes
	\begin{align}
		f_{\lambda_k}(N, g)
		&=
		\frac{1}{s_{\lambda_k}} \sandwich{\ninputGT, \ninputGT}{P_{\lambda_k} \tau_n^\nmodes \otimes \bar \tau_n^\nmodes(g)^\ad}{N, N} 
		\label{eq:filter_proof_2}\\
		&=
		\frac{1}{s_{\lambda_k}} (-1)^{\varphi(\ninputGT)} \sum_{M \in \GT{\lambda_k}} C_{\ninputGT, \ninputGTdual}^{M} \sandwich{M}{(\tau \otimes \bar \tau)(g)^\ad}{N, N} 
		\label{eq:filter_proof_3}\\
		&=
		\frac{1}{s_{\lambda_k}} (-1)^{\varphi(\ninputGT)} \sum_{M \in \GT{\lambda_k}} C_{\ninputGT, \ninputGTdual}^{M} 
		\sum_{N_1, N_2 \in \GT{\tau_n^\nmodes}}
		C_{N_1, \bar N_2}^{M} \sandwich{N_1, \bar N_2}{\tau_n^\nmodes \otimes \bar \tau_n^\nmodes(g)^\ad}{N, N}
		\label{eq:filter_proof_4} \\
		&=
		\frac{1}{s_{\lambda_k}} (-1)^{\varphi(\ninputGT)} \sum_{M \in \GT{\lambda_k}} C_{\ninputGT, \ninputGTdual}^{M} 
		\sum_{N_1, N_2 \in \GT{\tau_n^\nmodes}} (-1)^{\varphi(N_2)} C_{N_1, \bar N_2}^{M} \\
		&\times
		\sandwich{N_1, N_2}{\tau_n^\nmodes \otimes \bar \tau_n^\nmodes(g)^\ad}{N, N} \, , \nonumber
	\end{align}
	where in \cref{eq:filter_proof_2} we used \cref{eq:projection_fock_state} and relabeled the second entries as dual basis vectors by introducing the corresponding relative phases (Cf. \cref{eq:phase_dual_GTpattern}).
	In \cref{eq:filter_proof_3}, we used the Clebsch-Gordan decomposition $M = \sum_{N_1, N_2} C_{N_1, \bar N_2}^{M} \ket{N_1, \bar N_2}$, and in \cref{eq:filter_proof_4} we relabeled $\bar N_2$ as a basis vector for $\tau_n^\nmodes$.
	In particular, the latter implies that $f_{\lambda_k}$ is manifestly related to the computation of permanents \cite{scheelPermanentsLinearOptical2004}, as each inner product resembles the boson sampling problem when expressed in the Fock basis:
	\begin{align}
		\sandwich{N_1, N_2}{\tau_n^\nmodes \otimes \bar \tau_n^\nmodes(g)^\ad}{N, N}
		&=
		\sandwich{\vec n_1, \vec n_2}{\tau_n^\nmodes \otimes \bar \tau_n^\nmodes(g)^\ad}{\vec n, \vec n} \\
		&=
		\sandwich{\vec n_1}{\tau_n^\nmodes(g)^\ad}{\vec n} \sandwich{\vec n}{\tau_n^\nmodes(g)}{\vec n_2} \\
		&=
		\frac{1}{\vec n! \sqrt{\vec n_1! \vec n_2!}} \per(g^\ad_{\vec n_1, \vec n}) \per(g_{\vec n, \vec n_2}) \, ,
	\end{align}
	Finally, since weights for symmetric irreps of $\SU(\nmodes)$ uniquely identify basis vectors, by selection rules of Clebsch-Gordan coefficients we have that $N_2 = N_1$ (see also \cref{eq:CGCs_symmetric_orthogonality}), and the assertion follows by plugging in \cref{thm:frame_operator_pnr}.
\end{proof}

Clebsch-Gordan coefficients can be calculated in polynomial time \cite{alexNumericalAlgorithmExplicit2011} for $\approx 20$ modes before memory overhead limits the application of such algorithms \cite{alexSUClebschGordanCoefficients2012}.
We remark that --using exact formulae for CG coefficients for $\tau_n^\nmodes \otimes \bar{\tau}_n^\nmodes$ from \cite{vilenkinRepresentationLieGroups1992Vol3}-- one can compute the necessary Clebsch-Gordan coefficients for higher number of modes $m$.
Hence, the computational hardness of $f_{\lambda_k}$ comes from the evaluation of the permanents appearing in \cref{eq:filter_function_PNR}.

Alternatively, $f_{\lambda_k}$ can be expressed as a linear combination of matrix coefficients of $\lambda_k$ (in the corresponding \ac{GT} basis).
To show this fact, we introduce the following technical result that will be used extensively in the rest of this work:
\begin{lemma}
	\label{lem:sandwich_projector}
	Let $N, X$ be \ac{GT} patterns, and let $\bar N, \bar X$ be their dual, respectively.
	Let $\tau_n^\nmodes$ be the $n$-particles maximally symmetric \ac{irrep} of $\SU(\nmodes)$ and consider $\lambda_k \in \omega_n^\nmodes \cong \tau_n^\nmodes \otimes \bar \tau_n^{\nmodes}$.
	Let $P_{\lambda_k}$ be the projector onto $\lambda_k$.
	Then, the following holds:
	\begin{equation}
		\sandwich{N, \bar N}{P_{\lambda_k} (\tau_n^\nmodes \otimes \bar \tau_n^\nmodes)(g)^\ad}{X, \bar X}
		=
		\sum_{M \in \GT{\lambda_k}} C_{N, \bar N}^{M} \sum_{M' \in \GT{\lambda_k}} C_{X, \bar X}^{M'} \sandwich{M}{\lambda_k(g)^\ad}{M'} \, .
	\end{equation}
\end{lemma}
\begin{proof}
	By construction, $P_{\lambda_k}$ selects the $\lambda_k$-th component of $\ket{N, \bar N}$, which can be conveniently isolated by the Clebsch-Gordan decomposition of $\omega_n^\nmodes \cong \tau_n^\nmodes \otimes \bar \tau_n^\nmodes$.
	This implies
	\begin{equation}
		\sandwich{N, \bar N}{P_{\lambda_k} (\tau_n^\nmodes \otimes \bar \tau_n^\nmodes)(g)^\ad}{X, \bar X}
		=
		\sum_{M \in \GT{\lambda_k}} C_{N, \bar N}^{M} \sandwich{M}{\tau_n^\nmodes \otimes \bar \tau_n^\nmodes(g)^\ad}{X, \bar X} \, .
	\end{equation}
	To compute the inner product, recall that $\tau_n^\nmodes \otimes \bar \tau_n^\nmodes = \bigoplus_{l=0}^n \lambda_l$ (Cf. \cref{eq:decomposition_conjugate_action}).
	In particular, observe that $M$ is a basis element in $\lambda_k$, which implies the only non trivial contributions from $\tau_n^\nmodes \otimes \bar \tau_n^\nmodes$ are associated with its $\lambda_k$-th component.
	Likewise, the only relevant contributions to the inner product coming from $\ket{X, \bar X}$ are associated with its restriction to $\lambda_k$ that can be expressed as
	\begin{equation}
		\left. \ket{X, \bar X} \right|_{\lambda_k}
		=
		\sum_{M' \in \GT{\lambda_k}} C_{X, \bar X}^{M'} \ket{M'} \, .
	\end{equation}
	Hence,
	\begin{equation}
		\sandwich{M}{\tau_n^\nmodes \otimes \bar \tau_n^\nmodes(g)^\ad}{X, \bar X}
		=
		\sum_{M' \in \GT{\lambda_k}} C_{X, \bar X}^{M'} \sandwich{M}{\lambda_k(g)^\ad}{M'} \, ,
	\end{equation}
	from which the assertion follows.
\end{proof}
As mentioned before, a first consequence is the following explicit expression for the filter function:
\begin{corollary}
\label{cor:filter-function-matrix-coefficients}
	For a given \ac{irrep} $\lambda_k \in \hat \omega_n^\nmodes$, and assuming $s_{\lambda_k} \neq 0$, the following holds true:
	\begin{equation}
		\label{eq:filter_function_PNR_alternative}
		f_{\lambda_k}(\vec n, g) 
		=
		\frac{1}{s_{\lambda_k}} (-1)^{\varphi(\ninputGT) + \varphi(N)} \sum_{M \in \GT{{\lambda_k}}} C_{\ninputGT, \ninputGTdual}^{M} \sum_{M' \in \GT{{\lambda_k}}} C_{N, \bar N}^{M'} \sandwich{M}{\lambda_k(g)^\ad}{M'} \, ,
	\end{equation}
	where the sums are restricted to all basis states such that \cref{eq:weight_sum_GT+GTdual} is satisfied.
\end{corollary}
\begin{proof}
	By \cref{eq:cg_decomposition_vector}, and denoting by $\ninputGT$ and $N$ the \ac{GT} patterns associated with $\ninput$ and $\vec n$, respectively, the filter function defined in \cref{eq:filter_function} becomes
	\begin{align}
		f_{\lambda_k}(N, g)
		&=
		\frac{1}{s_{\lambda_k}} \sandwich{\ninputGT, \ninputGT}{P_{\lambda_k} (\tau_n^\nmodes \otimes \bar \tau_n^\nmodes)(g)^\ad}{N, N} \\
		&=
		\frac{1}{s_{\lambda_k}} (-1)^{\varphi(\ninputGT) + \varphi(N)}
		\sandwich{\ninputGT, \ninputGTdual}{P_{\lambda_k} (\tau_n^\nmodes \otimes \bar \tau_n^\nmodes)(g)^\ad}{N, \bar N} \\
		&=
		\frac{1}{s_{\lambda_k}} (-1)^{\varphi(\ninputGT) + \varphi(N)} 
		\sum_{M \in \GT{\lambda_k}} C_{\ninputGT, \ninputGTdual}^{M} \sum_{M' \in \GT{\lambda_k}} C_{N, \bar N}^{M'} \sandwich{M}{\lambda_k(g)^\ad}{M'} \, .
	\end{align}
	In the last line, we used \cref{lem:sandwich_projector}.
\end{proof}
Notably, explicit expressions for the matrix elements of \acp{irrep} of $\SU(\nmodes)$ are also available, see for instance \cite[Chapter 3]{vilenkinRepresentationLieGroupsSpecialVol1} for $\SU(2)$ and \cite[Chapter 9]{vilenkinRepresentationLieGroupsSpecialVol2} for $\SU(\nmodes)$.
Moreover, numerical implementations of the bosonic realization of $\mathfrak{su}(\nmodes)$ are available \cite{Dhand_2015} and matrix coefficients can be expressed as suitable permanents of $\lambda_k(g)$.

\subsection{Moments of the filter function for PNR measurement settings}
\label{sec:moments_PNR}

In this section, we provide explicit expressions for the first two moments of the filter function \eqref{eq:filter_function} w.r.t the ideal probability distribution $p(\vec n \vert g) = \sandwich{\vec n}{\omega_n^\nmodes(g)(\rho)}{\vec n}$, $\vec n \in \NN^\nmodes$.
In particular, the ideal second moment will provide an upper bound to the sampling complexity of the protocol, Cf. \cref{sec:guarantees}.
Throughout this section $dg$ denotes the Haar measure on $\SU(\nmodes)$.

The following technical result will be useful:
\begin{lemma}
	\label{lem:sandwich_no_projector}
	Let $N, X$ be \ac{GT} patterns, and let $\bar N, \bar X$ be their dual, respectively.
	Let $\tau_n^\nmodes$ be the $n$-particles maximally symmetric \ac{irrep} of $\SU(\nmodes)$ and let $\lambda_k$ be an \ac{irrep} in  $\omega_n^\nmodes$.
	Then, the following holds:
	\begin{equation}
		\sandwich{X, \bar X}{(\tau_n^\nmodes \otimes \bar \tau_n^\nmodes)(g)}{N, \bar N}
		=
		\sum_{j=0}^n \sum_{M, M' \in \GT{\lambda_j}} C_{X, \bar X}^{M} C_{N, \bar N}^{M'} \sandwich{M}{\lambda_j(g)}{M'} \, .
	\end{equation}
\end{lemma}
\begin{proof}
	The expression follows immediately from the Clebsch-Gordan decomposition.
	More specifically, we have
	\begin{align}
		\ket{N, \bar N} = \sum_{i=0}^n \sum_{M \in \GT{\lambda_i}} C_{N, \bar N}^{M} \ket{M} \, ,
		&
		\quad \ket{M, \bar M} = \sum_{j=0}^n \sum_{M' \in \GT{\lambda_j}} C_{X, \bar X}^{M'} \ket{M'} \, , \\
		\tau_n^\nmodes \otimes \bar \tau_n^\nmodes &= \bigoplus_{k=0}^n \lambda_k \, ,
	\end{align}
	Cf. \cref{eq:cg_decomposition_vector} and \cref{eq:decomposition_conjugate_action}.
	Hence, since $M$ and $M'$ are basis elements of $\lambda_i$ and $\lambda_j$, respectively,
	\begin{equation}
		\sandwich{M}{\lambda_k(g)}{M'} \neq 0
	\end{equation}
	only if $i = j = k$, from which the assertion follows.
\end{proof}
We will also need the following standard result in the representation theory of compact groups, often referred as Schur's orthogonality relations, see~\cite[Thm.~5.8]{follandCourseAbstractHarmonic2015}:
For a given \ac{irrep} $\lambda$ of $\SU(\nmodes)$ (or, in general, of any compact group $G$), if $M_1, M_2, M'_1, M'_2 \in \GT{\lambda}$, then the following relation holds true:
\begin{equation}
	\label{eq:orthogonality_relations}
	\int dg\, \sandwich{M_1}{\lambda(g)}{M'_1} \sandwich{M_2}{\lambda(g)^\ad}{M'_2}
	=
	\frac{1}{d_\lambda} \delta_{M_1, M_2} \delta_{M'_1, M'_2}
	\, ,
\end{equation}
where $dg$ denotes the Haar measure on $G$.

Before we prove the main results of this section, it is worth to quickly consider the first moment $\EE[f_{\lambda_k}]$, as the proof scheme is the same, but in the case of the second moment is hidden behind additional technical details concerning the representations involved.
\begin{lemma}
	\label{lem:first_moment_filter_function_pnr}
	For a \ac{PNR} measurement setting, 
	$\rho = \ketbra{\ninput}{\ninput}$ as input state, and an \ac{irrep} $\lambda_k$ of $\omega^{(n)} = \tau_n^\nmodes (\argdot) \tau_n^{\nmodes \ad}$, the following holds:
	\begin{equation}
		\label{eq:first_moment_filter_function_pnr_CG}
		\EE[f_{\lambda_k}] = 
		\sum_{M \in \GT{\lambda_k}} \abs{ C_{\ninputGT, \ninputGTdual}^{M} }^2 \, ,
	\end{equation}
	where $\ninputGT$ is the \ac{GT} pattern associated with $\ninput$, and $\ninputGTdual$ is its dual.
\end{lemma}

\begin{proof}
	Since $\rho$ is an $n$-particle state and $\omega_n^\nmodes$ is a passive transformation, the outcome of a PNR measurement must also be an $n$-particle Fock state.
	Hence, we have
	\begin{align}
		\EE[f_{\lambda_k}]
		&\coloneqq
		\frac{1}{s_{\lambda_k}} \sum_{\vec n \in \H_n^\nmodes} \int dg\,
		\sandwich{\ninput, \ninput}{P_{\lambda_k} (\tau_n^\nmodes \otimes \bar \tau_n^\nmodes)(g)^\ad}{\vec n, \vec n}
		\sandwich{\vec n, \vec n}{(\tau_n^\nmodes \otimes \bar \tau_n^\nmodes)(g)}{\ninput, \ninput} \\
		&=
		\frac{1}{s_{\lambda_k}} \sum_{N \in \GT{\tau_n^\nmodes}}
		\int dg\, \sandwich{\ninputGT, \ninputGTdual}{P_{\lambda_k}(\tau_n^\nmodes \otimes \bar \tau_n^\nmodes)(g)^\ad}{N, \bar N}
		\sandwich{N, \bar N}{(\tau_n^\nmodes \otimes \bar \tau_n^\nmodes)(g)}{\ninputGT, \ninputGTdual} \\
		&=
		\frac{1}{s_{\lambda_k}} \sum_{M \in \GT{\lambda_k}} C_{\ninputGT, \ninputGTdual}^{M} \sum_{N \in \GT{\tau_n^\nmodes}}
		\int dg\, \sandwich{M}{\lambda_k(g)^\ad}{N, \bar N} \sandwich{N, \bar N}{\bigoplus_{j=0}^n \lambda_j(g)}{\ninputGT, \ninputGTdual} \, .
	\end{align}
	The second line follows since the relative phases between $\ket{M}$ and $\ket{\bar M}$ highlighted in \cref{eq:conjugation_relative_phase} are integers and they appear an even number of times.
	In the third step, we projected $\ket{\ninputGT, \ninputGTdual}$ onto its $\lambda_k$-th component, see \cref{eq:projection_fock_state}.
	Accordingly, the only non-trivial contribution to the integral is determined by the $\lambda_k$-th component of $\tau_n^\nmodes \otimes \bar \tau_n^\nmodes$.
	Similarly, by orthogonality relations, the integral is non-zero only if $\lambda_j = \lambda_k$.
	Hence, it is enough to consider the restricted Clebsch-Gordan decomposition to the $\lambda_k$-th irrep, and the following holds:
	\begin{align}
		\EE[f_{\lambda_k}]
		&=
		\frac{1}{s_{\lambda_k}} \sum_{M \in \GT{\lambda_k}} C_{\ninputGT, \ninputGTdual}^{M} \sum_{N \in \GT{\tau_n^\nmodes}}
		\int dg\, \sandwich{M}{\lambda_k(g)^\ad}{N, \bar N} \sandwich{N, \bar N}{\lambda_k(g)}{\ninputGT, \ninputGTdual} \\
		&=
		\frac{1}{s_{\lambda_k}} \sum_{M \in \GT{\lambda_k}} C_{\ninputGT, \ninputGTdual}^{M} \sum_{N \in \GT{\tau_n^\nmodes}}
		\sum_{M_1, M_2, M_3 \in \GT{\lambda_k}}	C_{N, \bar N}^{M_1} C_{N, \bar N}^{M_2} C_{\ninputGT, \ninputGTdual}^{M_3} \\
		&\times
		\underbrace{\int dg\, \sandwich{M}{\lambda_k(g)^\ad}{M_1} \sandwich{M_2}{\lambda_k(g)}{M_3}}_{=\frac{1}{d_{\lambda_k}} \delta_{M, M_3} \delta_{M_1, M_2}}\\
		&=
		\frac{1}{d_{\lambda_k}} \frac{1}{s_{\lambda_k}}
		\sum_{M \in \GT{\lambda_k}} \abs{C_{\ninputGT, \ninputGTdual}^{M}}^2
		\sum_{N \in \GT{\tau_n^\nmodes}} \sum_{M_1 \in \GT{\lambda_k}} \abs{C_{N, \bar N}^{M_1}}^2 \\
		&=
		\sum_{M \in \GT{\lambda_k}} \abs{C_{\ninputGT, \ninputGTdual}^{M}}^2.
	\end{align}
	In the second step, we expanded $\ket{N, \bar N}, \ket{\ninputGT, \ninputGTdual}$ in the coupled basis, and restricted the decompositions to the $\lambda_k$-th components.
	In the third step, we used Schur's orthogonality relations to compute the integral, and in the final step we used the definition of the frame operator, and in particular the result in \cref{thm:frame_operator_pnr}.
\end{proof}

In general, finding explicit expressions for the second moment $\EE[f_\lambda^2]$ is more involved and the following technical result is necessary:
\begin{lemma}
	\label{lem:lambda_two_copies_decomp}
	Let $\lambda_k$ be a Young diagram as in \cref{eq:lambda_diagram} labeling an irrep of $\SU(\nmodes), \, \nmodes \geq 3$.
	Then,
	\begin{equation}
		\lambda_k \otimes \lambda_k = \bigoplus_{l=0}^k \lambda_l^{(l+1)} \oplus \bigoplus_{l=k+1}^{2k} \lambda_l^{(2k-l+1)} \oplus L \, ,
	\end{equation}
	where $\lambda_0 \equiv \mathbf{1}$, $\lambda_1 \equiv \Ad$, $\lambda_j^{(i)}$ denotes the $i$-th copy of $\lambda_j$ in $\lambda_l^{\otimes 2}$, and $L$ is a suitable direct sum of \acp{irrep} which are not of the form $\lambda_l$ for any $l \in \NN$.
\end{lemma}
Specifically, all irreps $\lambda_l$ in $\lambda_k^{\otimes 2}$ are computed by identifying all admissible ways of combining two copies of $\lambda_k$ to a fixed shape using Littlewood-Richardson's rules.
This result is proved in \cref{app:proof_lambda_two_copies}.

Hence, we derive an explicit expression for $\EE[f_{\lambda_k}^2]$ using \cref{eq:orthogonality_relations}:
\begin{theorem}
	\label{thm:second_moment_filter_function_pnr}
	For a \ac{PNR} measurement setting, 
	$\rho = \ketbra{\ninput}{\ninput}$ as input state, and an \ac{irrep} $\lambda_k$ of $\omega_n^\nmodes$, the following holds:
	\begin{equation}
		\begin{aligned}
			\EE[f_{\lambda_k}^2]
			=
			\frac{1}{s_{\lambda_k}^2} (-1)^{\varphi(\ninputGT)} \sum_{N \in \GT{\tau_n^\nmodes}} (-1)^{\varphi(N)} g_k(N, \ninputGT) \, ,
		\end{aligned}
	\end{equation}
	where $g_{k}(N, \ninputGT)$ is a function of Clebsch-Gordan coefficients of the representations $\tau_n^\nmodes \otimes \bar \tau_n^\nmodes$ and $\lambda_k^{\otimes 2}$ given by
	\begin{equation}
			g_{k}(N, \ninputGT)
			=
			\sum_{l=0}^{\min(n, 2k)} \frac{1}{d_{\lambda_l}} \sum_{r=1}^{m_l}
			\sum_{\substack{M, M' \in \GT{\lambda_k} \\ L, L' \in \GT{\lambda_k} \\ R, R' \in \GT{\lambda_l}}}
			C_{\ninputGT, \ninputGTdual}^{M}
			C_{\ninputGT, \ninputGTdual}^{M'}
			C_{\ninputGT, \ninputGTdual}^{R} 
			C_{N, \bar N}^{L} C_{N, \bar N}^{L'}
			C_{N, \bar N}^{R'}
			C_{M, M'}^{R, r}
			C_{L, L'}^{R', r}
	\end{equation}
	where $m_l$ is the multiplicity of $\lambda_l$ in $\lambda_k^{\otimes 2}$ as in \cref{lem:lambda_two_copies_decomp} and $d_{\lambda_l} \equiv \dim \lambda_l$.
\end{theorem}
\begin{proof}
	For any \ac{irrep} $\lambda_k \in \hat \omega_n^\nmodes$, and by relabeling the second entries as basis elements of the dual \ac{irrep} $\bar \tau_n^\nmodes$, the second moment can be expressed as follows:
	\begin{equation}
		\begin{aligned}
			\EE[f_{\lambda_k}^2]
			&\coloneqq
			\frac{1}{s_{\lambda_k}^2}
			\sum_{\vec n \in \H_n^\nmodes} \int dg\, \sandwich{\ninput, \ninput}{P_{\lambda_k} (\tau_n^\nmodes \otimes \bar \tau_n^\nmodes)(g)^\ad}{\vec n, \vec n}^2 \sandwich{\vec n, \vec n}{\tau_n^\nmodes \otimes \tau_n^\nmodes(g)}{\ninput, \ninput} \\
			&=
			\frac{1}{s_{\lambda_k}^2}
			\sum_{N \in \GT{\tau_n^\nmodes}} \int dg\, 
			\sandwich{\ninputGT, \ninputGT}{P_{\lambda_k} (\tau_n^\nmodes \otimes \bar \tau_n^\nmodes)(g)^\ad}{N, N}^2
			\sandwich{N, N}{\tau_n^\nmodes \otimes \tau_n^\nmodes(g)}{\ninputGT, \ninputGT} \\
			&=
			\frac{1}{s_{\lambda_k}^2}
			\sum_{N \in \GT{\tau_n^\nmodes}} (-1)^{\varphi(\ninputGT) + \varphi(N)} \int dg\, 
			\sandwich{\ninputGT, \ninputGTdual}{P_{\lambda_k} (\tau_n^\nmodes \otimes \bar \tau_n^\nmodes)(g)^\ad}{N, \bar N}^2 \\
			&\times
			\sandwich{N, \bar N}{\tau_n^\nmodes \otimes \tau_n^\nmodes(g)}{\ninputGT, \ninputGTdual} \\
			&=
			\frac{1}{s_{\lambda_k}^2} (-1)^{\varphi(\ninputGT)} \sum_{N \in \GT{\tau_n^\nmodes}} (-1)^{\varphi(N)} g_{k}(N, \ninputGT) \, ,
		\end{aligned}
	\end{equation}
	where
	\begin{equation}
		g_{k}(N, \ninputGT) \equiv \int dg\, 
		\sandwich{\ninputGT, \ninputGTdual}{P_{\lambda_k} (\tau_n^\nmodes \otimes \bar \tau_n^\nmodes)(g)^\ad}{N, \bar N}^2 \sandwich{N, \bar N}{\tau_n^\nmodes \otimes \tau_n^\nmodes(g)}{\ninputGT, \ninputGTdual} \, .
	\end{equation}
	By \cref{lem:sandwich_no_projector,lem:sandwich_projector}, we obtain
	\begin{equation}
		\begin{aligned}
			g_{k,l}(N, \ninputGT)
			&=
			\sum_{M, M' \in \GT{\lambda_k}} C_{\ninputGT, \ninputGTdual}^{M} C_{\ninputGT, \ninputGTdual}^{M'}
			\sum_{L, L' \in \GT{\lambda_k}} C_{N, \bar N}^{L} C_{N, \bar N}^{L'}
			\sum_{j=0}^n \sum_{J, J' \in \GT{\lambda_j}} C_{N, \bar N}^{J} C_{\ninputGT, \ninputGTdual}^{J'} \\
			&\times
			\underbrace{\int dg\, \sandwich{M, M'}{\lambda_k(g)^{\ad \otimes 2}}{L, L'} \sandwich{J}{\lambda_j(g)}{J'}}_{\equiv I}
		\end{aligned}
	\end{equation}
	We compute the integral with Schur's orthogonality relations.
	Specifically, this requires the \ac{irrep} decomposition of $\lambda_k^{\otimes 2}$:
	By \cref{lem:lambda_two_copies_decomp}, we have
	\begin{equation}
		\begin{aligned}
			\sandwich{M, M'}{\lambda_k(g)^{\ad \otimes 2}}{L, L'}
			&=
			\sandwich{M, M'}{\bigoplus_{l=0}^{2k} \lambda_l^{\oplus m_l}(g)^{\ad}}{L, L'} \, ,
		\end{aligned}
	\end{equation}
	where $m_l \in \{0, 1, \dots, k\}$ is the multiplicity of $\lambda_l$ in $\lambda_k^{\otimes 2}$ worked out in \cref{lem:lambda_two_copies_decomp}.
	Then, consider the following Clebsch-Gordan decompositions:
	\begin{equation}
		\ket{M, M'} = \sum_{i=0}^{2k} \sum_{r_i=1}^{m_i} \sum_{R \in \GT{\lambda_i}} C_{M, M'}^{R, r_i} \ket{R, r_i} \, , \quad
		\ket{L, L'} = \sum_{h=0}^{2k} \sum_{r_h=1}^{m_h} \sum_{R' \in \GT{\lambda_h}} C_{L, L'}^{R', r_h} \ket{R', r_h} \, ,
	\end{equation}
	where $r_i, r_h$ denote the $r_i$-th and $r_h$-th copies of $\lambda_i$ and $\lambda_h$ in $\lambda_k^{\otimes 2}$, respectively.
	By orthogonality of irreps, it follows
	\begin{equation}
		\sandwich{M, M'}{\lambda_k(g)^{\ad \otimes 2}}{L, L'}
		=
		\sum_{l=0}^{2k} \sum_{r=1}^{m_l} \sum_{R, R' \in \GT{\lambda_l}} C_{M, M'}^{R, r} C_{L, L'}^{R', r} \,
		\sandwich{R, r}{\lambda_l^{(r)}(g)^\ad}{R', r} \, .
	\end{equation}
	By Schur's orthogonality relations, this implies that the only non trivial contributions in $I$ are associated with irreps $\lambda_l$ which appears in the intersection of the sets of irreps of $\tau_n^\nmodes \otimes \bar \tau_n^\nmodes$ and $\lambda_k \otimes \lambda_k$, i.e. $j=l$ provided that $\lambda_l$ appears in both decomposition.
	More specifically,
	\begin{equation}
		\begin{aligned}
			I
			&=
			\sum_{l=0}^{2k} \sum_{r=1}^{m_l} \sum_{R, R' \in \GT{\lambda_l}} C_{M, M'}^{R, r} C_{L, L'}^{R', r} 
			\int dg\, \sandwich{R, r}{\lambda_l^{(r)}(g)^\ad}{R', r} \sandwich{J}{\lambda_j(g)}{J'} \\
			&=
			\sum_{l=0}^{2k} \delta_{l, j} \sum_{r=1}^{m_l} \sum_{R, R' \in \GT{\lambda_l}} C_{M, M'}^{R, r} C_{L, L'}^{R', r} 
			\int dg\, \sandwich{R, r}{\lambda_l^{(r)}(g)^\ad}{R', r} \sandwich{J}{\lambda_l(g)}{J'} \\
			&=
			\sum_{l=0}^{2k} \frac{1}{d_{\lambda_l}} \delta_{l, j} \sum_{r=1}^{m_l} \sum_{R, R' \in \GT{\lambda_l}} C_{M, M'}^{R, r} C_{L, L'}^{R', r} 
			\delta_{R, J'} \delta_{R', J} \, .
		\end{aligned}
	\end{equation}
	Therefore, we have
	\begin{equation}
		\begin{aligned}
			g_{k}(N, \ninputGT)
			&=
			\sum_{M, M' \in \GT{\lambda_k}} C_{\ninputGT, \ninputGTdual}^{M} C_{\ninputGT, \ninputGTdual}^{M'}
			\sum_{L, L' \in \GT{\lambda_k}} C_{N, \bar N}^{L} C_{N, \bar N}^{L'}
			\sum_{j=0}^n \sum_{J, J' \in \GT{\lambda_j}} C_{N, \bar N}^{J} C_{\ninputGT, \ninputGTdual}^{J'} \\
			&\times
			\sum_{l=0}^{2k} \frac{1}{d_{\lambda_l}} \delta_{l, j} \sum_{r=1}^{m_l} \sum_{R, R' \in \GT{\lambda_l}} C_{M, M'}^{R, r} C_{L, L'}^{R', r} \delta_{R, J'} \delta_{R', J} \\
			&=
			\sum_{l=0}^{\min(n,2k)} \frac{1}{d_{\lambda_l}} \sum_{r=1}^{m_l}
			\sum_{M, M' \in \GT{\lambda_k}} C_{\ninputGT, \ninputGTdual}^{M} C_{\ninputGT, \ninputGTdual}^{M'}
			\sum_{L, L' \in \GT{\lambda_k}} C_{N, \bar N}^{L} C_{N, \bar N}^{L'} \\
			&\times
			\sum_{R, R' \in \GT{\lambda_l}} C_{\ninputGT, \ninputGTdual}^{R} C_{N, \bar N}^{R'} C_{M, M'}^{R, r} C_{L, L'}^{R', r} \, ,
		\end{aligned}
	\end{equation}
	from which the assertion follows.
\end{proof}

\subsection{A worked out example}
\label{sec:su2_example_PNR}

In the case of $2$ modes systems, Clebsch-Gordan coefficients reduce to the usual ones, and the analysis of the filter function and its moments drastically simplifies.
In this section, we show explicit expressions for such a case, which will highlight some technicalities implicit in the general case of $\SU(\nmodes)$.

In the $\SU(2)$ case, it is convenient to switch from the bosonic realization of the $\SU(2)$ algebra to its spin realization, where Clebsch-Gordan coefficients are naturally introduced.
This task is accomplished by the Jordan-Schwinger map \cite{anielloExploringRepresentationTheory2006}:
For given annihilation operators $a_1, a_2$ acting on a $2$ mode system and satisfying the \acp{CCR}, the Jordan-Schwinger map is such that
\begin{align}
	J_1 \coloneqq \frac{1}{2} \left( a_2^\ad a_1 + a_1^\ad a_2 \right) \, , \quad
	J_2 \coloneqq \frac{1}{2} \left( a_2^\ad a_1 - a_1^\ad a_2 \right) \, , \quad
	J_3 \coloneqq \frac{1}{2} \left(a_1^\ad a_1 - a_2^\ad a_2 \right) \, , 
\end{align}
where $[J_i, J_j] = i \epsilon_{ijk} J_k$, $\epsilon$ is the Levi-Civita's pseudo-tensor, and
\begin{equation}
	J^2 = J_1^2 + J_2^2 + J_3^2 = \frac{n}{2} \left( \frac{n}{2} + 1 \right) \, , \quad n = n_1 + n_2 \, , \, n_i = a_i^\ad a_i \, .
\end{equation}
This implies the normalized states $\ket{n_1, n_2}$ correspond to the eigenstates $\ket{j m}$ of $J^2$ and $J_3$, with the identification \cite{mathurSUIrreducibleSchwinger2010, mathurInvariantsProjectionOperators2011}
\begin{equation}
	n_1 = j + m \, , \quad n_2 = j - m \, ,
\end{equation}
hence, in this section, we will consider an input state $\rho = \ketbra{j m}{j m}$ and the Fock state \ac{POVM} becomes $\{\ketbra{j' m'}{j' m'}\}$, where $j' \in \frac{1}{2} \NN$ and $m' = -j', \dots, j'$.
A spin state $\ket{j m}$ and its dual are identified by the \ac{GT} patterns
\begin{equation}
	\begin{pmatrix}
		2j && 0 \\
		&j+m
	\end{pmatrix}
	\, , \quad
	\begin{pmatrix}
		2j && 0 \\
		&j-m
	\end{pmatrix} \, ,
\end{equation}
respectively, which implies the following relation:
\begin{equation}
	\ket{j m} = (-1)^{j-m} \ket{j -m} \, .
\end{equation}
Moreover, given any irrep $\lambda_J$ of $\SU(2)$, the following relations hold:
\begin{equation}
	P_J = \sum_{M = -J}^{J} \ketbra{J M}{J M} \, , \quad s_J = \frac{1}{2J+1} \, .
\end{equation}
In particular, the expression for $s_J$ follows from \cref{eq:frameop_pnr}, the fact that the inner multiplicities of $\SU(2)$ basis vectors are $1$ (or, equivalently, each weight is uniquely associated with a unique weight vector).

In this case, with the identification $\ket{x_1, x_2} \mapsto \ket{j \, l}$, \cref{eq:filter_function_PNR} becomes
\begin{equation}
	f_J(l, g) = \frac{1}{2J+1} (-1)^{2j-m-l} C_{jm, j-m}^{J 0} C_{j l, j -l}^{J 0} \sandwich{J 0}{\lambda_J(g)^\ad}{J 0} \, ,
\end{equation}
or, equivalently, it can be expressed as (Cf. \cref{eq:filter_function_PNR_alternative})
\begin{equation}
	\begin{aligned}
		f_J(l, g) 
		&=
		\frac{(-1)^{2j-m}}{2J+1} C_{jm,j-m}^{J 0} \sum_{m' = -j}^{j} (-1)^{-m'} C_{jm',j-m'}^{J 0} \sandwich{jm', j-m'}{(\tau_n^2 \otimes \bar \tau_n^2)(g)^\ad}{x_1, x_2} \\
		&=
		\frac{(-1)^{2j-m}}{2J+1} C_{jm,j-m}^{J 0} \sum_{m' = -j}^{j} (-1)^{-m'} C_{jm',j-m'}^{J 0}
		\abs{\sandwich{x_1, x_2}{\tau_n^2(g)}{n'_1, n'_2}}^2 \\
		&=
		\frac{1}{2J+1} \frac{(-1)^{2j-m}}{x_1! x_2!} C_{jm,j-m}^{J 0} \sum_{m' = -j}^{j} \frac{(-1)^{-m'}}{n'_1! n'_2!} C_{jm',j-m'}^{J 0}
		\abs{\per(g_{(n'_1,n'_2),(x_1,x_2)})}^2 \, ,
	\end{aligned}
\end{equation}
where we set $\ket{j m'} = \ket{n'_1, n'_2}$ by the inverse Jordan-Schwinger map.

The second moment expression also simplifies significantly.
First, notice that, for a given representation $\lambda_J \otimes \lambda_J$, each $\lambda_K$, with $K \in \{0, \dots, 2J\}$, is multiplicity free as all such irreps are clearly maximally symmetric.
This implies the decomposition of $\lambda_J^{\otimes 2}$ is formally the same as the one of $\tau_n^2 \otimes \bar \tau_n^2$, i.e.
\begin{equation}
	\begin{tikzpicture}
		[scale=.5,baseline={([yshift=1ex]current bounding box.center)}]
		\matrix (m) [matrix of math nodes,
		nodes={draw, minimum size=3mm, anchor=center},
		column sep=-\pgflinewidth,
		row sep=-\pgflinewidth
		]
		{
			\ & |[draw=none]|\dots & \ \\
		};
		\draw[BC] (m-1-3.south) -- node[below=2mm] {$K$} (m-1-1.south);
	\end{tikzpicture}
	\otimes
	\begin{tikzpicture}
		[scale=.5,baseline={([yshift=1ex]current bounding box.center)}]
		\matrix (m) [matrix of math nodes,
		nodes={draw, minimum size=3mm, anchor=center},
		column sep=-\pgflinewidth,
		row sep=-\pgflinewidth
		]
		{
			\ & |[draw=none]|\dots & \ \\
		};
		\draw[BC] (m-1-3.south) -- node[below=2mm] {$K$} (m-1-1.south);
	\end{tikzpicture}
	=
	1 \oplus
	\ydiagram{2} \oplus \ydiagram{4} \oplus \dots \oplus
	\begin{tikzpicture}
		[scale=.5,baseline={([yshift=1ex]current bounding box.center)}]
		\matrix (m) [matrix of math nodes,
		nodes={draw, minimum size=3mm, anchor=center},
		column sep=-\pgflinewidth,
		row sep=-\pgflinewidth
		]
		{
			\ & |[draw=none]|\dots & \ \\
		};
		\draw[BC] (m-1-3.south) -- node[below=2mm] {$2K$} (m-1-1.south);
	\end{tikzpicture}
\end{equation}
and the second moment expression of \cref{thm:second_moment_filter_function_pnr} simplifies to
\begin{equation}
	\begin{aligned}
		\EE[f_J^2]
		&=
		\frac{1}{s_J^2} (-1)^{2j} \abs{C_{jm, j-m}^{J 0}}^2
		\sum_{l = -j}^{j} (-1)^{m+l} \abs{C_{jl, j-l}^{J 0}}^2 \sum_{K=0}^{2 \min(J,j)} \frac{1}{2K+1} C_{jl,j-l}^{K 0} C_{jm,j-m}^{K 0} \abs{ C_{J0,J0}^{K 0} }^2 \, .
	\end{aligned}
\end{equation}

\newpage

\section*{Appendices}
\appendix
\setcounter{equation}{0}
\setcounter{figure}{0}
\setcounter{table}{0}
\setcounter{section}{0}
\setcounter{theorem}{0}
\makeatletter
\renewcommand{\theequation}{\thesection\arabic{equation}}
\renewcommand{\thefigure}{\thesection\arabic{figure}}
\renewcommand{\thetheorem}{\thesection\arabic{theorem}}

\section{Littlewood-Richardson's rules}
\label{app:littlewood_richardson}

In this section, summarize Littlewood-Richardson's rules for the decomposition into irreps of the tensor product of two irreps of $\SU(\nmodes)$.
For more details, see for instance \cite[Sec.~C.3]{sternbergGroupTheoryPhysics1994}.
In particular, let us consider the unitary \acp{irrep} $\pi_{\lambda_1}, \pi_{\lambda_2}$ of $\SU(\nmodes)$ associated with Young diagrams $\lambda_1, \lambda_2$.
Then, the representation
\begin{equation}
	\label{eq:tensor_product_diag_repr}
	\pi_{\lambda_1} \otimes \pi_{\lambda_2} : \SU(\nmodes) \ni g \mapsto \pi_{\lambda_1}(g) \otimes \pi_{\lambda_2}(g) \in \U(\H_{\lambda_1} \otimes \H_{\lambda_2})
\end{equation}
is in general reducible (in particular, it is completely reducible, since $\SU(\nmodes)$ is compact).

For instance, in standard \ac{RB} \cite{heinrichRandomizedBenchmarkingRandom2023}, one is interested in the \ac{irrep} $U \otimes \bar U$, where $U: \SU(\nmodes) \rightarrow \U(\nmodes)$ is the defining representation and $\bar U$ its dual.
Diagrammatically, they correspond to
\begin{equation}
	\lambda_U = 
	\begin{tikzpicture}
		[scale=.5,baseline={([yshift=-.5ex]current bounding box.center)}]
		\matrix (m) [matrix of math nodes,
		nodes={draw, minimum size=4mm, anchor=center},
		column sep=-\pgflinewidth,
		row sep=-\pgflinewidth
		]
		{
			\ \\
		};
	\end{tikzpicture} \, , \quad
	\lambda_{\bar U} = 
	\begin{tikzpicture}
		[scale=.5,baseline={([yshift=-.5ex]current bounding box.center)}]
		\matrix (m) [matrix of math nodes,
		nodes={draw, minimum size=4mm, anchor=center},
		column sep=-\pgflinewidth,
		row sep=-\pgflinewidth
		]
		{
			\ \\
			|[draw=none]|\vdots \\
			\ \\
		};
		\draw[BC] (m-1-1.east) -- node[right=2mm] {$\nmodes-1$} (m-3-1.east);
	\end{tikzpicture} \, ,
\end{equation}
It is well known that
\begin{equation}
	U \otimes \bar{U} = 1 \oplus \Ad ,
\end{equation}
where $1$ denotes the trivial \ac{irrep} and $\Ad: \SU(D) \ni g \mapsto \Ad_g \in \mathrm{Aut}(\mathfrak{su}(D))$ is the adjoint representation.
Roughly speaking, the decomposition is achieved by combining the two Young diagrams in all possible ways, and summing up the results.
In this case, there are only two possibilities that realize legal Young diagrams: 
$\lambda_U$ can be attached on the right of the top row of $\lambda_{\bar U}$, or on the bottom of the column, i.e.
\begin{equation}
	\label{def:conjugate_action_defining_irrep}
	\begin{tikzpicture}
		[scale=.5,baseline={([yshift=-.5ex]current bounding box.center)}]
		\matrix (m) [matrix of math nodes,
		nodes={draw, minimum size=4mm, anchor=center},
		column sep=-\pgflinewidth,
		row sep=-\pgflinewidth
		]
		{
			\ \\
			|[draw=none]|\vdots \\
			\ \\
		};
		\draw[BC] (m-3-1.west) -- node[left=2mm] {$\nmodes-1$} (m-1-1.west);
	\end{tikzpicture}
	\otimes
	\begin{tikzpicture}
		[scale=.5,baseline={([yshift=-.5ex]current bounding box.center)}]
		\matrix (m) [matrix of math nodes,
		nodes={draw, minimum size=4mm, anchor=center},
		column sep=-\pgflinewidth,
		row sep=-\pgflinewidth
		]
		{
			\ \\
		};
	\end{tikzpicture}
	=
	\begin{tikzpicture}
		[scale=.5,baseline={([yshift=-.5ex]current bounding box.center)}]
		\matrix (m) [matrix of math nodes,
		nodes={draw, minimum size=4mm, anchor=center},
		column sep=-\pgflinewidth,
		row sep=-\pgflinewidth
		]
		{
			\ \\
			|[draw=none]|\vdots \\
			\ \\
		};
		\draw[BC] (m-1-1.east) -- node[right=2mm] {$\nmodes$} (m-3-1.east);
	\end{tikzpicture}
	\oplus
	\begin{tikzpicture}
		[scale=.5,baseline={([yshift=-.5ex]current bounding box.center)}]
		\matrix (m) [matrix of math nodes,
		nodes={draw, minimum size=4mm, anchor=center},
		column sep=-\pgflinewidth,
		row sep=-\pgflinewidth
		]
		{
			\ & \ \\
			\ \\
			|[draw=none]|\vdots \\
			\ \\
		};
		\draw[BC] (m-2-1.east) -- node[right=2mm] {$\nmodes-2$} (m-4-1.east);
	\end{tikzpicture} \, .
\end{equation}
The first diagram on the r.h.s. is equivalent to the diagram with no boxes associated with the trivial \ac{irrep}, while the second Young diagram identifies the adjoint representation acting on traceless matrices.

In general, \emph{Littlewood-Richardson's rules} can be used to decompose the tensor product of two arbitrary irreps \cite{fultonYoungTableauxApplications1996, sternbergGroupTheoryPhysics1994, difrancescoConformalFieldTheory1997}.
To spell out such rules, let us consider two Young diagrams $\lambda_1, \lambda_2$ associated with \acp{irrep} of $\SU(\nmodes)$.
The tensor product representation $\lambda_1 \otimes \lambda_2$ can be evaluated algorithmically as follows \cite[Sec.~13.5.3]{difrancescoConformalFieldTheory1997} (or also \cite[Sec. C.3]{sternbergGroupTheoryPhysics1994}):
\begin{enumerate}
	\item Assign distinct labels to boxes in each row of the Young diagram $\lambda_2$.
	For instance, the boxes in the first row will be labeled by `a', the boxes in the second row by `b' and so on.
	\item Attach boxes labeled by $a$ to $\lambda_1$ in all possible ways such that no two $a$'s appear in the same column, and the result is still a proper Young diagram.
	\item Repeat the steps above for all rows of $\lambda_2$.
	\item Elimination rule:
	For each box, assign numbers $n_a = $ number of $a$'s above and right to it, $n_b = $ number of $b$'s above and right to it, and so on.
	If, at any point, the relations $n_a \geq n_b \geq n_c \geq n_d \geq \dots$ are not satisfied, discard this diagram.
	\item Merging rule.
	If two diagrams are the same, then they are counted as the same if the labels are the same.
	Otherwise, they refer to distinct \acp{irrep}.
	\item Cancel columns with $\nmodes$ boxes (since they correspond to constant shifts of the highest weight vector).
	\item Remove all the labels after the cancellation and the merging steps.
\end{enumerate}
\begin{example}
	Let us consider the following diagrams:
	\begin{equation}
		\lambda_1 =
		\ydiagram{2,1} \, , \quad
		\lambda_2 = 
		\ydiagram{2,1} \, .
	\end{equation}
	Assigning labels to $\lambda_2$ as in rule 1, and using the second and third rules, we get
	\begin{equation}
		\begin{aligned}
			\ydiagram{2,1} \otimes
			\begin{ytableau}
				a & a \\
				b
			\end{ytableau}
			& = 
			\left(
			\ytableaushort{\none \none a, \none}*{3,1}
			\oplus
			\ytableaushort{\none \none, \none a}*{2,2} 
			\oplus
			\ytableaushort{\none, \none, a}*{2,1,1}
			\right) \otimes
			\ytableaushort{a,b}*{1,1} \\
			& = 
			\left(
			\ytableaushort{\none \none a a, \none}*{4,1}
			\oplus
			\ytableaushort{\none \none a, \none a}*{3,2}
			\oplus
			\ytableaushort{\none \none a, \none, a}*{3,1,1}
			\right) \otimes \ytableaushort{b}*{1} \\
			& \oplus
			\left(
			\ytableaushort{\none \none a, \none a}*{3,2}
			\oplus
			\ytableaushort{\none \none, \none a, a}*{2,2,1}
			\oplus
			\ytableaushort{\none \none a, \none, a}*{3,1,1}
			\oplus
			\ytableaushort{\none \none, \none a, a}*{2,2,1}
			\right) \otimes \ytableaushort{b}*{1} \\
			& = 
			\ytableaushort{\none \none a a b, \none}*{5,1}
			\oplus
			\ytableaushort{\none \none a a, \none b}*{4,2}
			\oplus
			\ytableaushort{\none \none a a, \none, b}*{4,1,1} \\
			& \oplus
			\ytableaushort{\none \none a b, \none a}*{4,2}
			\oplus
			\ytableaushort{\none \none a, \none a b}*{3,3}
			\oplus
			\ytableaushort{\none \none a, \none a, b}*{3,2,1} \\
			& \oplus
			\ytableaushort{\none \none a b, \none, a}*{4,1,1}
			\oplus
			\ytableaushort{\none \none a, \none b, a}*{3,2,1}
			\oplus
			\ytableaushort{\none \none a, \none, a, b}*{3,1,1,1} \\
			& \oplus
			\ytableaushort{\none \none b, \none a, a}*{3, 2, 1}
			\oplus
			\ytableaushort{\none \none, \none a, a b}*{2,2,2}
			\oplus
			\ytableaushort{\none \none, \none a, a, b}*{2,2,1,1} \, . 
		\end{aligned}
	\end{equation}
	In the second step we got few equivalent diagrams with labels in the same positions, hence they have been merged according to the merging rule.
	Moreover, we ignored the diagrams with two $a$'s in the same column, in agreement with the symmetric constraint.
	
	By the elimination rule, all the diagrams with a $b$ box attached on the top right shall be eliminated, yielding the following decomposition:
	\begin{equation}
		\begin{aligned}
			\ydiagram{2,1} \otimes
			\begin{ytableau}
				a & a \\
				b
			\end{ytableau}
			& = 
			\ytableaushort{\none \none a a, \none b}*{4,2}
			\oplus
			\ytableaushort{\none \none a a, \none, b}*{4,1,1}
			\oplus
			\ytableaushort{\none \none a, \none a b}*{3,3}
			\oplus
			\ytableaushort{\none \none a, \none a, b}*{3,2,1} \\
			& \oplus
			\ytableaushort{\none \none a, \none b, a}*{3,2,1}
			\oplus
			\ytableaushort{\none \none a, \none, a, b}*{3,1,1,1}
			\oplus
			\ytableaushort{\none \none, \none a, a b}*{2,2,2}
			\oplus
			\ytableaushort{\none \none, \none a, a, b}*{2,2,1,1} \, . 
		\end{aligned}
	\end{equation}
	Finally, suppose for instance that these diagrams are associated with $\SU(3)$ \acp{irrep}.
	Then, all columns with three boxes can be omitted, while any diagram with more than $3$ boxes in a column is not allowed.
	This yields the following decomposition:
	\begin{equation}
		\begin{aligned}
			\ydiagram{2,1} \otimes
			\begin{ytableau}
				a & a \\
				b
			\end{ytableau}
			& = 
			\ytableaushort{\none \none a a, \none b}*{4,2}
			\oplus
			\ytableaushort{\none \none a, \none a b}*{3,3}
			\oplus
			\ytableaushort{\none a a}*{3}
			\oplus
			\ytableaushort{\none a, a}*{2,1}
			\oplus
			\ytableaushort{\none a, b}*{2,1}
			\oplus
			1 \, .
		\end{aligned}
	\end{equation}
	In the latter, notice that the diagram $\ydiagram{2,1}$ appears with multiplicity $2$ in the decomposition, as different labels are assigned to the two copies.
	
	In the high energy literature \cite{difrancescoConformalFieldTheory1997}, this decomposition is also written in terms of the dimension of the \ac{irrep} associated with each Young diagram as
	\begin{equation}
		\mathbf{
			8 \otimes \bar 8 = 27 \oplus 10 \oplus 10' \oplus 8 \oplus 8 \oplus 1
		} \, .
	\end{equation}
	Here, $\mathbf{10}, \, \mathbf{10'}$ indicates that the two irreps are inequivalent, while repeated dimensions denote the same irrep appears with a non-trivial multiplicity.
\end{example}

\section{Proof of Proposition \ref{prop:dimension_lambda_k}}
\label{app:proof_dim_irrep}

For convenience, we state again the proposition:
\begin{proposition}[restatement of \cref{prop:dimension_lambda_k}]
	For any $k \in \NN$, set $d_{\lambda_k} \equiv \dim \lambda_k$.
	\begin{equation}
		d_{\lambda_k} = \left( 1 - \frac{k^2}{(k+\nmodes-1)^2} \right) \left( \dim \H_k^\nmodes \right)^2 \, .
	\end{equation}
\end{proposition}
\begin{proof}
	For any irrep $\lambda = (m_1, m_2, \dots, m_{\nmodes})$, the following fact holds: \cite{alexNumericalAlgorithmExplicit2011}
	\begin{equation}
		\dim \lambda = \prod_{1 \leq j < j' \leq \nmodes} \left( 1 + \frac{m_j - m_{j'}}{j' - j} \right) \, .
	\end{equation}
	Let us denote the \ac{irrep} defined in \cref{eq:lambda_diagram} as $\lambda_k = (2k, k, \dots, k, 0)$.
	Hence, notice the following facts:
	\begin{itemize}
		\item For $j=1$ and $j'=2, \dots, \nmodes-1$ we obtain the contribution $\prod_{j'=2}^{\nmodes-1} \left(1 + \frac{k}{j'-1} \right)$.
		\item For $j=1$ and $j'=\nmodes$ we obtain the contribution $1 + \frac{2k}{\nmodes-1}$.
		\item For $2 \leq j < j' \leq \nmodes -1$ all the products are equal to $1$.
		\item For $2 \leq j \leq \nmodes-1$ and $j' = \nmodes$ we obtain the contribution $\prod_{j=2}^{\nmodes -1} \left( 1 + \frac{k}{\nmodes - j} \right)$.
	\end{itemize}
	Using the latter facts, we have
	\begin{align}
		d_{\lambda_k} 
		&= 
		\frac{2k+\nmodes-1}{\nmodes-1} \prod_{j=2}^{\nmodes-1} \frac{k+\nmodes-j}{\nmodes-j} \prod_{l=2}^{\nmodes-1} \frac{k+l-1}{l-1} \\
		&=
		\label{eq:dim_proof_step_2}
		\frac{2k+\nmodes-1}{\nmodes-1} \frac{1}{\left( (\nmodes-2)! \right)^2} \prod_{j=2}^{\nmodes-1} (k+\nmodes-j) \prod_{l=2}^{\nmodes-1} (k+l-1) \\
		&=
		\label{eq:dim_proof_step_3}
		\frac{2k+\nmodes-1}{\nmodes-1} \frac{1}{\left( (\nmodes-2)! \right)^2} \left( \frac{1}{k} (k)_{\nmodes-1} \right)^2  \\
		&=
		\frac{1}{k^2} \frac{2k+\nmodes-1}{\nmodes-1} \left( \frac{(k)_{\nmodes-1}}{(\nmodes-2)!} \right)^2 \\
		&=
		\frac{1}{k^2} \frac{2k+\nmodes-1}{\nmodes-1} (\nmodes-1)^2 \left( \frac{(k)_{\nmodes-1}}{(\nmodes-1)!} \right)^2 \\
		&=
		\label{eq:dim_proof_step_5}
		\frac{(2k+\nmodes-1)(\nmodes-1)}{k^2} \left( \frac{(k+\nmodes-1)!}{k!(\nmodes-1)!} \frac{k}{k+\nmodes-1} \right)^2 \\
		&=
		\frac{(2k+\nmodes-1)(\nmodes-1)}{k^2} \frac{k^2(\nmodes-1)}{(k+\nmodes-1)^2} \binom{k+\nmodes-1}{\nmodes-1}^2 \\
		&=
		\frac{(2k+\nmodes-1)(\nmodes-1)}{(k+\nmodes-1)^2} \binom{k+\nmodes-1}{\nmodes-1}^2 \\
		&=
		\left( 1 - \frac{k^2}{(k+\nmodes-1)^2} \right) \left( \dim \H_k^\nmodes \right)^2
		\, .
	\end{align}
	In \cref{eq:dim_proof_step_2} we factorized the denominators and observed that the factors range between $1$ and $\nmodes-2$.
	In \cref{eq:dim_proof_step_3} we introduced the Pochhammer raising factorial symbol, defined as $(a)_k \coloneqq a(a+1) \dots (a+k-1)$ for $a, k \in \NN$.
	In \cref{eq:dim_proof_step_5} we recognized that, by definition,
	\begin{equation}
		\frac{(k)_{\nmodes-1}}{(\nmodes-2)!} = \binom{k+\nmodes-2}{\nmodes-1}
		=
		\frac{(k+\nmodes-2)!}{(k-1)!(\nmodes-1)!} \cdot \frac{k}{k} \, . 
	\end{equation}
	Finally, rearranging  the terms and by symmetry of the binomial coefficient, the assertion follows.
\end{proof}

\section{Proof of Lemma \ref{lem:weight_space_zero}}
\label{app:weight_space_zero}

In this section, for convenience, we will say that a box in a Young tableau is a $k$-box if it is labeled by $k \in [\nmodes]$.
Recall that for each $N \in \GT{\tau_n^\nmodes}$ we have
\begin{equation}
	w_j^{(N)} + w_j^{(\bar N)} = 0 \, , \quad \forall j = 1, \dots \nmodes-1 \, .
\end{equation}
\begin{lemma}[Restatement of \cref{lem:weight_space_zero}]
	Let $\lambda_k$ be an \ac{irrep} of $\SU(\nmodes)$ as in \cref{eq:lambda_diagram} for any $k \in \NN$.
	Then,
	\begin{equation}
		\gamma_{\lambda_k}(\vec 0) = \binom{k + \nmodes - 2}{k} \, .
	\end{equation}
\end{lemma}
\begin{proof}
	From the point of view of Young tableaux, we remark that a state with weight $\vec 0$ implies that in the corresponding Young tableau $T_M$ --where $M \in \GT{\lambda_k}$ satisfies the latter selection rules-- all the entries appear the same number of times.
	Let $\ssyt{\lambda_k}$ be the set of semi-standard Young tableaux of shape $\lambda_k$ and consider the set
	\begin{equation}
		\ytzero{\lambda_k} \coloneqq \{ T \mid T \in \ssyt{\lambda_k} \text{ s.t. } w_i^{T} = w_{i+1}^T \, \forall i \in [\nmodes-1] \} \, .
	\end{equation}
	It follows that $\gamma_{\lambda_k}(\vec 0) = \abs{ \ytzero{\lambda_k} }$ is the inner multiplicity of $\vec 0$ in $\lambda_k$.
	Clearly, $\gamma(w) = 1$ for each weight $w$ in $\SU(2)$, and \cref{eq:weight_space_zero} holds trivially.
	
	In a similar fashion, counting Young tableaux $T_{\lambda_k}$ for $\SU(3)$ is straightforward:
	any Young tableau $T \in \ytzero{\lambda_k}$ contains the labels $\{1,2,3\}$ exactly $k$ times, with the $1$'s forced to be placed in the first $k$ boxes of the first row, otherwise $T$ would not be a legal tableau.
	Then, if we consider a starting Young tableau of the form
	\begin{equation}
		\begin{tikzpicture}
			[every node/.style={scale=0.8},baseline={([yshift=-.5ex]current bounding box.center)}]
			\matrix (m) [matrix of math nodes,
			nodes={draw, minimum size=8mm, anchor=center},
			column sep=-\pgflinewidth,
			row sep=-\pgflinewidth
			]
			{
				1 & |[draw=none]|\dots & 1 & 2 & |[draw=none]|\dots & 2 \\
				3 & |[draw=none]|\dots & 3 \\
			};
			\draw[BC] (m-1-1.north) -- node[above=4mm] {$k$} (m-1-3.north);
			\draw[BC] (m-1-4.north) -- node[above=4mm] {$k$} (m-1-6.north);
		\end{tikzpicture}
	\end{equation}
	all remaining $T \in \ytzero{\lambda_l}$ can be obtained by permuting the last $2$-box in the first row with the first $3$-box in the second row.
	The total number of allowed swaps is $k$, which implies $\gamma_{\lambda_k}(\vec 0) = k+1$.
	
	Consider now $\nmodes > 3$.
	As in the previous case, the $1$-boxes are fixed to be placed at the beginning of the first row.
	Suppose the $\nmodes$-boxes are all placed in the $\nmodes-1$-th row, i.e. we consider 
	\begin{equation}
		\begin{tikzpicture}
			[every node/.style={scale=0.8},baseline={([yshift=-.5ex]current bounding box.center)}]
			\matrix (m) [matrix of math nodes,
			nodes={draw, minimum size=8mm, anchor=center},
			column sep=-\pgflinewidth,
			row sep=-\pgflinewidth
			]
			{
				1 & |[draw=none]|\dots & 1 & 2 & |[draw=none]|\dots & 2 \\
				3 & |[draw=none]|\dots & 3 & \\
				4 & |[draw=none]|\dots & 4 & \\
				|[draw=none]|\vdots & |[draw=none]|\ddots & |[draw=none]|\vdots \\
				\nmodes & |[draw=none]|\dots & \nmodes \\
			};
			\draw[BC] (m-1-1.north) -- node[above=4mm] {$k$} (m-1-3.north);
			\draw[BC] (m-1-4.north) -- node[above=4mm] {$k$} (m-1-6.north);
			\draw[BC] (m-5-1.west) -- node[left=4mm] {$\nmodes-1$} (m-1-1.west);
		\end{tikzpicture} \, .
	\end{equation}
	As long as the last row is fixed to contain $\nmodes$-boxes only, the total number of such Young tableaux is $\binom{k + \nmodes-3}{k}$.
	Then, we only have to count the remaining allowed configurations of $k$-boxes.
	For this purpose, observe that the remaining allowed positions for $\nmodes$-boxes are only in the first row, and there are $k$ such configurations.
	Hence, it is enough to count all possible configurations for each placement of $\nmodes$-boxes in the first row, which is given by
	\begin{equation}
		\binom{k-l + \nmodes - 2}{k-l} \, ,
	\end{equation}
	where $l$ is the number of free boxes in the first row of the tableau.
	Therefore, the total number of such configurations is
	\begin{equation}
		\begin{aligned}
			\sum_{l=1}^k \binom{(k-l) + \nmodes -2}{k - l}
			&=
			\sum_{j=0}^{k-1} \binom{j + \nmodes - 2}{j}
			=
			\sum_{j=0}^{k-1} \binom{j + \nmodes - 2}{\nmodes - 2} \\
			&=
			\binom{k + \nmodes - 3}{k-1} \, ,
		\end{aligned}
	\end{equation}
	where we used Fermat's identity \cite[Eq. 1.48]{gouldCombinatorialIdentities}
	\begin{equation}
		\sum_{j=0}^n \binom{j+a}{j} = \binom{a + n + 1}{n} \, .
	\end{equation}
	Therefore, by Pascal's identity, we have
	\begin{equation}
		\gamma_{\lambda_k}(\vec 0) = \binom{k + \nmodes - 3}{k} + \binom{k + \nmodes - 3}{k - 1} = \binom{k + \nmodes -2 }{k} \, .
	\end{equation}
	and the proof is complete.	
\end{proof}

\section{Proof of Lemma \ref{lem:lambda_two_copies_decomp}}
\label{app:proof_lambda_two_copies}

\begin{lemma}[Restatement of \cref{lem:lambda_two_copies_decomp}]
	Let $\lambda_k$ be a Young diagram as in \cref{eq:lambda_diagram} labeling an irrep of $\SU(\nmodes), \, \nmodes \geq 3$.
	Then,
	\begin{equation}
		\lambda_k \otimes \lambda_k = \bigoplus_{l=0}^k \lambda_l^{(l+1)} \oplus \bigoplus_{l=k+1}^{2k} \lambda_l^{(2k-l+1)} \oplus L \, ,
	\end{equation}
	where $\lambda_0 \equiv \mathbf{1}$, $\lambda_1 \equiv \Ad$, $\lambda_j^{(i)}$ denotes the $i$-th copy of $\lambda_j$ in $\lambda_l^{\otimes 2}$, and $L$ is a suitable direct sum of \acp{irrep} which are not of the form $\lambda_l$ for any $l \in \NN$.
\end{lemma}
\begin{proof}
	Consider for any $k \in \NN$ the tensor product
	\begin{equation}
		\lambda_k \otimes \lambda_k =
		\ytableausetup{mathmode, boxframe=normal, boxsize=2.3em, centertableaux}
		\begin{ytableau}
			\ & \none[\dots] & \ & \ & \none[\dots] & \ \\
			\ & \none[\dots] & \ \\
			\none[\vdots] & \none[\dots] & \none[\vdots] \\
			\ & \none[\dots] & \
		\end{ytableau}
		\otimes
		\begin{ytableau}
			{a_1} & \none[\dots] & {a_1} & {a_1} & \none[\dots] & {a_1} \\
			{a_2} & \none[\dots] & {a_2} \\
			\none[\vdots] & \none[\dots] & \none[\vdots] \\
			{a_{\nmodes-1}} & \none[\dots] & {a_{\nmodes-1}}
		\end{ytableau}
		\, .
	\end{equation}
	By Littlewood-Richardson's rules, the number of Young diagrams $\lambda_l$ that can be constructed from $\lambda_k^{\otimes 2}$ is determined by all possible allowed configurations we can attach the $a_1$ boxes to the first $\lambda_k$, since the way the remaining $a_i$ boxes, $i = 2, \dots, \nmodes-1$, are attached must follow accordingly.
	First, notice that only the $a_1$ boxes can be attached to the first row of the first copy of $\lambda_k$ due to the elimination rule.
	Hence, we have two different `generating' Young diagrams conditioned by whether $l \leq k$ or $k+1 \leq l \leq 2k$.
	Suppose $l \leq k$ at first.
	The $a_1$ boxes are attached to the first copy of $\lambda_k$ as follows:
	The first $l$ boxes are attached to the first row, the next $k$ boxes are attached to the second row and the remaining $k-l$ boxes are attached to the $\nmodes$-th row.
	Then, all the $a_i$ boxes, for any $i = 2, 3, \dots, \nmodes-2$ are attached to the $i+1$-th row.
	Finally, the $a_{\nmodes-1}$ boxes are attached to the $\nmodes$-th row.
	The resulting Young diagram is given by
	\begin{equation}
		\begin{tikzpicture}
			[every node/.style={scale=0.8},baseline={([yshift=-.5ex]current bounding box.center)}]
			\matrix (m) [matrix of math nodes,
			nodes={draw, minimum size=10mm, anchor=center},
			column sep=-\pgflinewidth,
			row sep=-\pgflinewidth
			]
			{
				\ & |[draw=none]|\dots & \ & \ & |[draw=none]|\dots & \ & \ & |[draw=none]|\dots & \ & |[draw=none]|\dots & \ & a_1 & |[draw=none]|\dots &  a_1 \\
				\ & |[draw=none]|\dots & \ & \ & |[draw=none]|\dots & \ & a_1 & |[draw=none]|\dots & a_1 & |[draw=none]|\dots & a_1 \\
				\ & |[draw=none]|\dots & \ & \ & |[draw=none]|\dots & \ & a_2 & |[draw=none]|\dots & a_2 & |[draw=none]|\dots & a_2 \\
				|[draw=none]|\vdots & & |[draw=none]|\vdots & |[draw=none]|\vdots & & |[draw=none]|\vdots & |[draw=none]|\vdots & & |[draw=none]|\vdots & & |[draw=none]|\vdots \\
				\ & |[draw=none]|\dots & \ & \ & |[draw=none]|\dots & \ & a_{\scriptscriptstyle \nmodes-2} & |[draw=none]|\dots & a_{\scriptscriptstyle \nmodes-2} & |[draw=none]|\dots & a_{\scriptscriptstyle \nmodes-2} \\
				a_1 & |[draw=none]|\dots & a_1 & a_{\scriptscriptstyle \nmodes-1} & |[draw=none]|\dots & a_{\scriptscriptstyle \nmodes-1} & a_{\scriptscriptstyle \nmodes-1} & |[draw=none]|\dots & a_{ \scriptscriptstyle \nmodes-1} \\
			};
			\draw[BC] (m-1-1.north) -- node[above=4mm] {$k$} (m-1-6.north);
			\draw[BC] (m-1-7.north) -- node[above=4mm] {$k$} (m-1-11.north);
			\draw[BC] (m-1-12.north) -- node[above=4mm] {$l$} (m-1-14.north);
			\draw[BC] (m-6-9.south) -- node[below=4mm] {$k$} (m-6-4.south);
			\draw[BC] (m-6-3.south) -- node[below=4mm] {$k-l$} (m-6-1.south);
			
			\draw[BC] (m-6-1.west) -- node[left=4mm] {$\nmodes$} (m-1-1.west);
			\draw[BC] (m-2-11.east) -- node[right=4mm] {$\nmodes-2$} (m-5-11.east);
		\end{tikzpicture}	
	\end{equation}
	Suppose now $l \geq k + 1$.
	The $a_1$ boxes are attached to the first copy of $\lambda_k$ as follows:
	The first $l$ boxes are attached to the first row, while the remaining ones are attached to the second row of $\lambda_k$.
	Then, for the $a_i$ boxes, for any $i = 2, 3, \dots, \nmodes-2$, the first $2k-l$ are attached to the $i$-th row of $\lambda_k$, while the remaining ones are attached to the $i+1$-th row of $\lambda_k$.
	The first $a_{\nmodes-1}$ boxes are attached to the $\nmodes-1$-th row of $T_k$ and the remaining ones will form the $\nmodes$-th row of the diagram.
	The resulting Young diagram is
	\begin{equation}
		\begin{tikzpicture}
			[every node/.style={scale=0.8},baseline={([yshift=-.5ex]current bounding box.center)}]
			\matrix (m) [matrix of math nodes,
			nodes={draw, minimum size=10mm, anchor=center},
			column sep=-\pgflinewidth,
			row sep=-\pgflinewidth
			]
			{
				\ & |[draw=none]|\dots & \ & |[draw=none]|\dots & \ & \ & |[draw=none]|\dots & \ & \ & |[draw=none]|\dots & \ & a_1 & |[draw=none]|\dots &  a_1 \\
				\ & |[draw=none]|\dots & \ & |[draw=none]|\dots & \ & a_1 & |[draw=none]|\dots & a_1 & a_2 & |[draw=none]|\dots & a_2 \\
				\ & |[draw=none]|\dots & \ & |[draw=none]|\dots & \ & a_2 & |[draw=none]|\dots  & a_2 & a_3& |[draw=none]|\dots & a_3 \\
				|[draw=none]|\vdots & & |[draw=none]|\vdots & & |[draw=none]|\vdots & |[draw=none]|\vdots & & |[draw=none]|\vdots & |[draw=none]|\vdots & & |[draw=none]|\vdots \\
				\ & |[draw=none]|\dots & \ & |[draw=none]|\dots & \ & a_{\scriptscriptstyle \nmodes-2} & |[draw=none]|\dots & a_{\scriptscriptstyle \nmodes-2} & a_{\scriptscriptstyle \nmodes-1} & |[draw=none]|\dots & a_{\scriptscriptstyle \nmodes-1} \\
				a_{\scriptscriptstyle \nmodes-1} & |[draw=none]|\dots & a_{\scriptscriptstyle \nmodes-1} \\
			};
			\draw[BC] (m-1-1.north) -- node[above=4mm] {$k$} (m-1-5.north);
			\draw[BC] (m-1-6.north) -- node[above=4mm] {$k$} (m-1-11.north);
			\draw[BC] (m-1-12.north) -- node[above=4mm] {$l$} (m-1-14.north);
			
			\draw[BC] (m-6-3.south) -- node[below=4mm] {$2k-l$} (m-6-1.south);
			\draw[BC] (m-5-8.south) -- node[below=4mm] {$2k-l$} (m-5-6.south);
			\draw[BC] (m-5-11.south) -- node[below=4mm] {$l-k$} (m-5-9.south);
			
			\draw[BC] (m-6-1.west) -- node[left=4mm] {$\nmodes$} (m-1-1.west);
			\draw[BC] (m-2-11.east) -- node[right=4mm] {$\nmodes-2$} (m-5-11.east);
		\end{tikzpicture}	
	\end{equation}	
	For notation purpose, let us refer to the latter two Young diagrams as the generating Young diagrams.
	
	At this point, we can generate all the remaining copies of $\lambda_l$ in the following way:
	\begin{enumerate}
		\item For any $i = 1, \dots, \nmodes-2$, replace the last $a_i$ box in the $i+1$-th row with an $a_{i+1}$ box.
		\item Replace the first $a_{\nmodes-1}$ box in the $\nmodes$-th row of the diagram with $a_1$.
	\end{enumerate}
	It follows that the multiplicity of $\lambda_l$ in the decomposition of $\lambda_k^{\otimes 2}$ is determined by the number of $a_1$ boxes in the second row of the generating Young diagram.
\end{proof}

\section{Passive RB with heterodyne measurement}
\label{sec:passive_rb_Gaussian}

In this section, we prove \cref{thm:general_filter_informal,thm:variance_bound_informal} for passive \ac{RB} with (balanced) heterodyne measurements.
As before, the filter function \eqref{eq:filter_function_general_POVM} plays a central role and here assumes the following form:
\begin{equation}
	f_{\lambda_k} (\vec \alpha, g) = \frac{1}{s_{\lambda_k}} \sandwich{\ninput, \ninput}{P_{\lambda_k} \tau_n^\nmodes \otimes \bar \tau_n^\nmodes(g)^\ad}{\vec \alpha, \vec \alpha} \, ,
\end{equation}
where $\rho = \ketbra{\ninput}{\ninput}$ is the input state, $\lambda_k$ is an \ac{irrep} in $\tau_n^\nmodes \otimes \bar \tau_n^\nmodes$ and $\vec \alpha \in \CC^\nmodes$ denotes an $\nmodes$-modes coherent state.

Throughout this section, we will use the usual multi-index notation \cite[Sec.~9.1]{reed:simon:funcanalysis}:
For elements $\vec n_1, \vec n_2 \in \H_n^\nmodes$, $\vec n_1 + \vec n_2$ denotes the component-wise sum.
The multi-index factorial of $\vec n \in \H_n^\nmodes$ is defined as $\vec n! \coloneqq n_1! \dots n_\nmodes!$.
Also, for a given $\vec \alpha \in \CC^\nmodes$, we consider the power $\vec \alpha^{\vec n} \coloneqq \alpha_1^{n_1} \dots \alpha_\nmodes^{n_\nmodes}$, and we set $\abs{\vec \alpha}^p \coloneqq \alpha_1^p + \dots + \alpha_\nmodes^p$ for $p \geq 1$.
We also use the shorthand notation
\begin{equation}
	\int d^2\vec \alpha \equiv \int d^2\alpha_1 \dots \int d^2\alpha_\nmodes \, ,
\end{equation}
where $d^2\alpha_i$ is the complex measure on $\CC$.
With this notation, the multi-mode coherent state $\ket{\vec \alpha}$ can be expanded as
\begin{equation}
	\label{eq:multimode_cs}
	\ket{\vec \alpha} = e^{-\abs{\vec \alpha}^2/2} \sum_{\vec n \in \fock_\nmodes} \frac{\vec \alpha^{\vec n}}{\sqrt{\vec n!}} \ket{\vec n} \, .
\end{equation}

Consider now the following quantity for any $K \in 2\NN$:
\begin{equation}
	\label{eq:alpha_integral_multi-index}
	I( \{\vec n_i \}_{i=1}^K ) = \frac{1}{\sqrt{ \vec n_1 ! \vec n_2 ! \dots \vec n_K !}} \int d^2 \vec \alpha \, e^{-K/2 \abs{\vec \alpha}^2} \bar{\vec \alpha}^{\vec n_1 + \dots \vec n_{K/2}} \vec \alpha^{\vec n_{K/2+1} + \dots \vec n_K} \, . 
\end{equation}
The latter can be evaluated writing down the integral in polar coordinates and integrating by parts.
Specifically, for the single-mode integral, and for any $\eta > 0$, we have
\begin{equation}
	\label{eq:alpha_integral_solved_single}
	\begin{aligned}
		\int d^2\alpha \, e^{- \eta \abs{\alpha}^2} \alpha^{a + b} \bar{\alpha}^{c + d}
		&=
		\int_0^\infty dr \, e^{-\eta r^2} r^{a + b + c + d + 1} \int_0^{2\pi} d\theta \, e^{i \theta(a + b - c - d)} \\
		&=
		\pi \left(\frac{a+b+c+d}{2}\right)! \, \eta^{-\frac{a+b+c+d}{2}} \delta_{a+b, c+d} \, .
	\end{aligned}
\end{equation}
Notice that the expression in parenthesis is a proper factorial due to the $\delta$.
This implies
\begin{equation}
	\label{eq:alpha_integral_solved_multiindex}
	\begin{aligned}
		I(\{\vec n_i\}_{i=1}^K)
		&=
		\frac{\pi^\nmodes}{(K/2)^n} \frac{\left( \vec n_1 + \dots + \vec n_{K/2} \right)!}{\sqrt{\vec n_1! \dots \vec n_K!}} \, \delta_{\vec n_1 + \dots + \vec n_{K/2}, \vec n_{K/2+1} + \dots + \vec n_K} \, .
	\end{aligned}
\end{equation}
where 
\begin{equation}
	\delta_{\vec n_1 + \dots + \vec n_{K/2}, \vec n_{K/2+1} + \dots + \vec n_K} =
	\begin{cases}
		1, \, &\text{if } \sum_{i=1}^{K/2} \vec n_i = \sum_{i=K/2+1}^K \vec n_i \, , \\
		0, \, &\text{otherwise} \,
	\end{cases}
\end{equation}
and we used the fact that $\abs{\vec n_i} = n$ for each $i = 1, \dots K$.

Lastly, we recall the identification between Fock states and \ac{GT} patterns: we will write $\ket{\vec n_i} = \ket{N_i}$, meaning that $N_i = N_i(\vec n_i)$.

The coherent state \ac{POVM} $\{\ketbra{\vec \alpha}{\vec \alpha}\}_{\alpha \in \CC^\nmodes}$ is informationally complete \cite{darianoInformationallyCompleteMeasurements2004}, which implies $s_{\lambda_k} \neq 0$ for any $\lambda_k \in \hat \omega_n^\nmodes$ \cite{heinrichRandomizedBenchmarkingRandom2023}.
Specifically, we work out explicit formulae for the eigenvalues $s_{\lambda_k}$:
\begin{theorem}
	\label{thm:frame_operator_heterodyne}
	Let $\lambda_k$ an \ac{irrep} of $\tau_n^\nmodes \otimes \bar \tau_n^\nmodes$ as in \cref{eq:lambda_diagram}.
	For a balanced heterodyne measurement setting, the eigenvalues of the frame operator of the filtered \ac{RB} protocol are given by
	\begin{equation}
		\label{eq:frameop_cs}
		s_{\lambda_k}
		=
		\frac{\pi^\nmodes 2^{-n}}{d_{\lambda_k}}
		\sum_{\vec n_1, \vec n_2} (-1)^{\varphi(N_1) + \varphi(N_2)} \binom{\vec n_1 + \vec n_2}{\vec n_1}
		\sum_{M \in \GT{\lambda_k}} C_{N_1, \bar N_1}^{M} C_{N_2, \bar N_2}^{M} \, ,
	\end{equation}
	where $N_1, N_2 \in \GT{\tau_n^\nmodes}$ are the \ac{GT} patterns associated with $\vec n_1, \vec n_2$, respectively, and
	\begin{equation}
		\binom{\vec n_1 + \vec n_2}{\vec n_2}
		\equiv
		\binom{n_{1,1} + n_{2,1}}{n_{2,1}} \dots \binom{n_{1,\nmodes} + n_{2,\nmodes}}{n_{2, \nmodes}} \, ,
	\end{equation}
	where $\vec n_i = (n_{i,1}, \dots, n_{i,\nmodes})$.
\end{theorem}
\begin{proof}
	For the balanced heterodyne measurement setting the corresponding (ideal) \ac{POVM} is $\{ \ketbra{\vec \alpha}{\vec \alpha} \equiv E_{\vec \alpha} \}_{\vec \alpha \in \CC^\nmodes}$, where $\ket{\vec \alpha} = \bigotimes_{i=1}^\nmodes \ket{\alpha_i}$ is an $\nmodes$ modes coherent state.
	The associated measurement channel is given by
	\begin{equation}
		\label{eq:measurement_channel_heterodyne}
		\mathcal M(\cdot) \coloneqq \int_{\CC^\nmodes} d^2\vec \alpha \Tr[\ketbra{\vec \alpha}{\vec \alpha} (\cdot)] \, \ketbra{\vec \alpha}{\vec \alpha} \, .
	\end{equation}
	
	To evaluate \cref{eq:frameop_irrep_subspace}, we use the multi-mode expansion defined in \cref{eq:multimode_cs}.
	Moreover, since $P_{\lambda_k}$ is defined onto a subspace of $\H_n^\nmodes$, such expansions of the copies of $\vec \alpha$ are truncated.
	Hence, by \cref{eq:alpha_integral_multi-index}, we have
	\begin{align}
		s_{\lambda_k}
		&=
		\frac{1}{d_{\lambda_k}} \int d^2\vec \alpha \,
		\sandwich{\vec \alpha}{P_{\lambda_k}(\ketbra{\vec \alpha}{\vec \alpha})}{\vec \alpha} \\
		&=
		\frac{1}{d_{\lambda_k}} \sum_{\substack{\vec n_1, \vec n_2 \in \H_n^\nmodes \\ \vec m_1, \vec m_2 \in \H_n^\nmodes}} I(\vec n_1, \vec n_2, \vec m_1, \vec m_2) 
		\sandwich{\vec n_1}{P_{\lambda_k}(\ketbra{\vec n_2}{\vec m_2})}{\vec m_1} \\
		&=
		\frac{1}{d_{\lambda_k}} \sum_{\substack{\vec n_1, \vec n_2 \in \H_n^\nmodes \\ \vec m_1, \vec m_2 \in \H_n^\nmodes}} I(\vec n_1, \vec n_2, \vec m_1, \vec m_2) \sandwich{\vec n_1,\vec m_1}{P_{\lambda_k}}{\vec n_2,\vec m_2} \\
		&=
		\frac{1}{d_{\lambda_k}} \sum_{\substack{\vec n_1, \vec n_2 \in \H_n^\nmodes \\ \vec m_1, \vec m_2 \in \H_n^\nmodes}} I(\vec n_1, \vec n_2, \vec m_1, \vec m_2) \sandwich{N_1, M_1}{P_{\lambda_k}}{N_2,M_2} \\
		&=
		\frac{1}{d_{\lambda_k}} \sum_{\substack{\vec n_1, \vec n_2 \in \H_n^\nmodes \\ \vec m_1, \vec m_2 \in \H_n^\nmodes}} 
		(-1)^{\varphi(M_1) + \varphi(M_2)} I(\vec n_1, \vec n_2, \vec m_1, \vec m_2) \sandwich{N_1, \bar M_1}{P_{\lambda_k}}{N_2,\bar M_2} \\
		&=
		\frac{1}{d_{\lambda_k}} \sum_{\substack{\vec n_1, \vec n_2 \in \H_n^\nmodes \\ \vec m_1, \vec m_2 \in \H_n^\nmodes}}
		(-1)^{\varphi(M_1) + \varphi(M_2)} I(\vec n_1, \vec n_2, \vec m_1, \vec m_2)
		\sum_{M \in \GT{\lambda_k}} C_{N_1, \bar M_1}^{M} C_{N_2, \bar M_2}^{M}
		\, ,
	\end{align}
	where the phase $\varphi(M)$ is defined in \cref{eq:phase_dual_GTpattern} and in the last step we used the definition of $P_{\lambda_k}$, cf. \eqref{eq:projection_fock_state}:
	\begin{equation}
		\begin{aligned}
			\sandwich{N_1, \bar M_1}{P_{\lambda_k}}{N_2,\bar M_2}
			=
			\sum_{M \in \GT{\lambda_k}} \braket{N_1, \bar M_1}{M} \braket{M}{N_2, M_2}
			=
			\sum_{M \in \GT{\lambda_k}} C_{N_1, \bar M_1}^{M} C_{N_2, \bar M_2}^{M} \, .
		\end{aligned}
	\end{equation}
	Since $N_1, N_2 \in \GT{\tau_n^\nmodes}$ and $M_1, M_2 \in \GT{\bar \tau_n^\nmodes}$, selection rules for Clebsch-Gordan coefficients imply $M_1 = N_1$ and $M_2 = N_2$.
	Hence,
	\begin{equation}
		s_{\lambda_k}
		=
		\frac{1}{d_{\lambda_k}} \sum_{\vec n_1, \vec n_2 \in \H_n^\nmodes} (-1)^{\varphi(N_1) + \varphi(N_2)} I(\vec n_1, \vec n_2)
		\sum_{M \in \GT{\lambda_k}} C_{N_1, \bar N_1}^{M} C_{N_2, \bar N_2}^{M} \, .
	\end{equation}
	By \cref{eq:alpha_integral_solved_multiindex}, we have
	\begin{equation}
		I(\vec n_1, \vec n_2) 
		=
		\frac{\pi^\nmodes}{2^n} \binom{\vec n_1 + \vec n_2}{\vec n_1} \, .
	\end{equation}
	with
	\begin{equation}
		\binom{\vec n_1 + \vec n_2}{\vec n_2}
		\equiv
		\binom{n_{1,1} + n_{2,1}}{n_{2,1}} \dots \binom{n_{1,\nmodes} + n_{2,\nmodes}}{n_{2, \nmodes}} \, ,
	\end{equation}
	from which \cref{eq:frameop_cs} follows.
\end{proof}

A result similar to \cref{eq:filter_function_PNR} is available for heterodyne detectors:
\begin{theorem}[Restatement of \cref{thm:general_filter_informal} - heterodyne version]
	\label{thm:filter_function_heterodyne}
	Let $\rho = \ketbra{\ninput}{\ninput}$ be a $\nmodes$ modes state and let $\{ \ketbra{\vec \alpha}{\vec \alpha} \}_{\vec \alpha \in \CC^\nmodes}$ be the coherent state \ac{POVM}.
	Let $\lambda_k$ an \ac{irrep} of $\tau_n^\nmodes \otimes \bar \tau_n^\nmodes$ as in \cref{eq:lambda_diagram}.
	Then, the filter function \eqref{eq:filter_function} is given by
	\begin{equation}
		\label{eq:filter_function_heterodyne}
		\begin{aligned}
			f_{{\lambda_k}}(\vec \alpha, g) 
			&=
			\frac{(-1)^{\varphi(N_0)}}{s_{\lambda_k}} \sum_{M \in \GT{\lambda_k}} C_{N_0, \bar N_0}^{M}
			\sum_{N' \in \GT{\tau_n^\nmodes}} (-1)^{\varphi(N')} C_{N', \bar N'}^{M}
			\abs{ \sandwich{\vec \alpha}{\tau_n^\nmodes(g)}{\vec n'} }^2 \, ,
		\end{aligned}
	\end{equation}
	where $\ket{\vec n'} = \ket{N'}$.
\end{theorem}
\begin{proof}
	The proof is analogous to the \ac{PNR} case.
	By a slight generalization of \cref{lem:sandwich_projector} to include coherent state measurements, we have
	\begin{align}
		f_{\lambda_k}(\vec \alpha, g)
		&=
		\frac{1}{s_{\lambda_k}} \sandwich{\ninput, \ninput}{P_{\lambda_k} \tau_n^\nmodes \otimes \bar \tau_n^\nmodes(g)^\ad}{\vec \alpha, \vec \alpha} \\
		\label{filter_cs_proof_2}
		&=
		\frac{1}{s_{\lambda_k}} (-1)^{\varphi(N_0)} \sandwich{N_0, \bar N_0}{P_{\lambda_k} \tau_n^\nmodes \otimes \bar \tau_n^\nmodes(g)^\ad}{\vec \alpha, \vec \alpha} \\
		\label{filter_cs_proof_3}
		&=
		\frac{1}{s_{\lambda_k}} (-1)^{\varphi(N_0)} \sum_{M \in \GT{\lambda_k}} C_{N_0, \bar N_0}^{M} \sandwich{M}{\tau_n^\nmodes \otimes \bar \tau_n^\nmodes(g)^\ad}{\vec \alpha, \vec \alpha} \\
		&=
		\label{filter_cs_proof_4}
		\frac{1}{s_{\lambda_k}} (-1)^{\varphi(N_0)} \sum_{M \in \GT{\lambda_k}} C_{N_0, \bar N_0}^{M} \sum_{N_1, N_2 \in \GT{\tau_n^\nmodes}} C_{N_1, \bar N_2}^{M} \sandwich{N_1, \bar N_2}{\tau_n^\nmodes \otimes \bar \tau_n^\nmodes(g)^\ad}{\vec \alpha, \vec \alpha} \\
		\label{filter_cs_proof_5}
		&=
		\frac{1}{s_{\lambda_k}} (-1)^{\varphi(N_0)} \sum_{M \in \GT{\lambda_k}} C_{N_0, \bar N_0}^{M}
		\sum_{N_1, N_2 \in \GT{\tau_n^\nmodes}} (-1)^{\varphi(N_2)} C_{N_1, \bar N_2}^{M}
		\sandwich{N_1, N_2}{\tau_n^\nmodes \otimes \bar \tau_n^\nmodes(g)^\ad}{\vec \alpha, \vec \alpha} \\
		\label{filter_cs_proof_6}
		&=
		\frac{1}{s_{\lambda_k}} (-1)^{\varphi(N_0)} \sum_{M \in \GT{\lambda_k}} C_{N_0, \bar N_0}^{M}
		\sum_{N_1 \in \GT{\tau_n^\nmodes}} (-1)^{\varphi(N_1)} C_{N_1, \bar N_1}^{M}
		\sandwich{N_1, N_1}{\tau_n^\nmodes \otimes \bar \tau_n^\nmodes(g)^\ad}{\vec \alpha, \vec \alpha}
		 \, .
	\end{align}
	In \cref{filter_cs_proof_2}, we used the identification with \ac{GT} patterns and introduced the phases described in \cref{eq:phase_dual_GTpattern}.
	In \cref{filter_cs_proof_3}, we projected $\ket{N_0, \bar N_0}$ onto $\lambda_k$, c.f. \cref{eq:projection_fock_state}.
	In \cref{filter_cs_proof_4}, we applied the Clebsch-Gordan decomposition $\ket{M} = \sum_{N_1, N_2 \in \GT{\tau_n^\nmodes}} C_{N_1, \bar N_2}^{M} \ket{N_1, N_2}$.
	In \cref{filter_cs_proof_5}, we used again \cref{eq:phase_dual_GTpattern}, and in \cref{filter_cs_proof_6} we used selection rules for Clebsch-Gordan coefficients, since $N_1, N_2 \in \GT{\tau_n^\nmodes}$. 
	By the identification $\ket{\vec n_1} = \ket{N_1}$, we have
	\begin{equation}
		\sandwich{N_1, N_1}{\tau_n^\nmodes \otimes \bar \tau_n^\nmodes(g)^\ad}{\vec \alpha, \vec \alpha}
		=
		\abs{ \sandwich{\vec n_1}{\tau_n^\nmodes(g)}{\vec \alpha} }^2 \, ,
	\end{equation}
	and the assertion is proved.
\end{proof}	
We remark that an expression analogous to \cref{eq:filter_function_PNR_alternative} as a weighted sum of matrix coefficients of $\lambda_k$ can be worked out by considering the expansion of $\ket{\vec \alpha}$ in the Fock basis.

\subsection{Moments for the heterodyne measurement setting}

In this section, we provide explicit expressions for first two moments of probability of the filter function \eqref{eq:filter_function} in the case of heterodyne measurements, for which the ideal probability distribution is $p(\vec \alpha \lvert g) = \sandwich{\vec \alpha}{\omega_n^\nmodes(g)(\rho)}{\vec \alpha}$, $\vec \alpha \in \C^\nmodes$.
In particular, the ideal second moment will provide an upper bound to the sampling complexity of the protocol, Cf. \cref{sec:guarantees}.
As in \cref{sec:moments_PNR}, the proofs rely on the application of Schur's orthogonality relations \eqref{eq:orthogonality_relations}.

\begin{lemma}
	\label{lem:first_moment_filter_function_CS}
	For a heterodyne measurement setting, $\rho = \ketbra{\ninput}{\ninput}$ as input state, and an \ac{irrep} $\lambda_k$ of $\tau_n^\nmodes \otimes \bar \tau_n^{\nmodes}$, the following holds:
	\begin{equation}
		\label{eq:first_moment_filter_function_CS_CG}
		\EE[f_{\lambda_k}]
		=
		\frac{1}{d_{\lambda_k} s_{\lambda_k}} \sum_{M \in \GT{\lambda_k}} \abs{C_{N_0, \bar N_0}^{M}}^2
		\sum_{\substack{\vec n_1, \vec n_2 \in \H_n^\nmodes \\ \vec m_1, \vec m_2 \in \H_n^\nmodes}}
		(-1)^{\varphi(M_1) + \varphi(M_2)} I(\vec n_1, \vec n_2, \vec m_1, \vec m_2) \sum_{S \in \GT{\lambda_k}} C_{N_1, \bar M_1}^{S} C_{N_2, \bar M_2}^{S} \, ,
	\end{equation}
	where $\ket{\ninput} = \ket{\ninputGT}, \ket{\vec n_1} \equiv \ket{N_1}, \ket{\vec n_2} \equiv \ket{N_2}, \ket{\vec m_1} \equiv \ket{M_1}, \ket{\vec m_1} \equiv \ket{M_1}$, $I$ is as in \cref{eq:alpha_integral_multi-index} and $\varphi$ is defined in \cref{eq:phase_dual_GTpattern}.
\end{lemma}
\begin{proof}
	As in the proof of \cref{thm:frame_operator_heterodyne}, considering the multi-mode expansion of $\vec \alpha$, only the $n$-particle component provides non-trivial contribution to the first moment, since $\omega_n^\nmodes$ acts non trivially on $\H_n^\nmodes$ only.
	Recalling \cref{eq:alpha_integral_multi-index}, it follows
	\begin{align}
		\EE[f_{\lambda_k}]
		&=
		\frac{1}{s_{\lambda_k}} \int d^2\vec \alpha \int dg\, 
		\sandwich{\ninput}{P_{\lambda_k} \circ \omega_n^\nmodes(g)^\ad(\ketbra{\vec \alpha}{\vec \alpha})}{\ninput} 
		\sandwich{\vec \alpha}{\omega_n^\nmodes(g)(\ketbra{\ninput}{\ninput})}{\vec \alpha} \\
		&=
		\frac{1}{s_{\lambda_k}} \sum_{\substack{\vec n_1, \vec n_2 \in \H_n^\nmodes \\ \vec m_1, \vec m_2 \in \H_n^\nmodes}}
		I(\vec n_1, \vec n_2, \vec m_1, \vec m_2)
		\int d^2\vec \alpha \int dg\, 
		\sandwich{\ninput}{P_{\lambda_k} \circ \omega_n^\nmodes(g)^\ad(\ketbra{\vec n_1}{\vec m_1})}{\ninput} 
		\sandwich{\vec n_2}{\omega_n^\nmodes(g)(\ketbra{\ninput}{\ninput})}{\vec m_2}
		\, .
	\end{align}
	In particular,
	\begin{align}
		H
		&\equiv
		\int d^2\vec \alpha \int dg\, 
		\sandwich{\ninput}{P_{\lambda_k} \circ \omega_n^\nmodes(g)^\ad(\ketbra{\vec n_1}{\vec m_1})}{\ninput} 
		\sandwich{\vec n_2}{\omega_n^\nmodes(g)(\ketbra{\ninput}{\ninput})}{\vec m_2} \\
		&=
		\int dg\, \sandwich{\ninput, \ninput}{P_{\lambda_k} \tau_n^\nmodes \otimes \bar \tau_n^\nmodes(g)^\ad}{\vec n_1, \vec m_1} \sandwich{\vec n_2, \vec m_2}{\tau_n^\nmodes \otimes \bar \tau_n^\nmodes(g)}{\ninput, \ninput} \\
		&=
		\int dg\, \sandwich{\ninputGT, \ninputGT}{P_{\lambda_k} \tau_n^\nmodes \otimes \bar \tau_n^\nmodes(g)^\ad}{N_1, M_1} \sandwich{N_2, M_2}{\tau_n^\nmodes \otimes \bar \tau_n^\nmodes(g)}{\ninputGT, \ninputGT} \\
		&=
		(-1)^{\varphi(M_1) + \varphi(M_2)}
		\int dg\, \sandwich{\ninputGT, \ninputGTdual}{P_{\lambda_k} \tau_n^\nmodes \otimes \bar \tau_n^\nmodes(g)^\ad}{N_1, \bar M_1} \sandwich{N_2, \bar M_2}{\tau_n^\nmodes \otimes \bar \tau_n^\nmodes(g)}{\ninputGT, \ninputGTdual} \\
		&=
		(-1)^{\varphi(M_1) + \varphi(M_2)} \sum_{M \in \GT{\lambda_k}} C_{N_0, \bar N_0}^{M}
		\int dg\,\sandwich{M}{\lambda_k(g)^\ad}{N_1, \bar M_1}
		\sandwich{N_2, \bar M_2}{\bigoplus_{j=0}^n \lambda_j(g)}{N_0, \bar N_0}
		\, .
	\end{align}
	In the last step, we projected $\ket{\ninputGT, \ninputGTdual}$ onto $\lambda_k$ and we used the \ac{irrep} decomposition of $\tau_n^\nmodes \otimes \bar \tau_n^\nmodes$, see \cref{lem:decomposition_conjugate_action}.
	The latter can be computed by slight modifications of \cref{lem:sandwich_projector,lem:sandwich_no_projector}.
	In particular, by orthogonality relations, the integral is non-zero only if $j=k$ and for basis vectors of $\lambda_k$.
	In other words, it is enough to restrict the Clebsch-Gordan decompositions to the $\lambda_k$-th components:
	\begin{align}
		\left. \ket{N_1, \bar M_1} \right|_{\lambda_k} &= \sum_{S_1 \in \GT{\lambda_k}} C_{N_1, \bar M_1}^{S_1} \ket{S_1} \, , \\
		\left. \ket{N_2, \bar M_2} \right|_{\lambda_k} &= \sum_{S_2 \in \GT{\lambda_k}} C_{N_2, \bar M_2}^{S_2} \ket{S_2} \, , \\
		\left. \ket{N_0, \bar N_0} \right|_{\lambda_k} &= \sum_{S_3 \in \GT{\lambda_k}} C_{N_0, \bar N_0}^{S_3} \ket{S_3} \, .		
	\end{align}
	Therefore, Schur's orthogonality relations \eqref{eq:orthogonality_relations}, we have
	\begin{align}
		H
		&=
		\frac{1}{d_{\lambda_k}} (-1)^{\varphi(M_1) + \varphi(M_2)}
		\sum_{M \in \GT{\lambda_k}} \abs{C_{N_0, \bar N_0}^{M}}^2 \sum_{S \in \GT{\lambda_k}} C_{N_1, \bar M_1}^{S} C_{N_2, \bar M_2}^{S}
	\end{align}
	from which the assertion follows .
\end{proof}

Lastly, we have the following explicit expression for the second moment:
\begin{theorem}
	\label{thm:second_moment_filter_function_cs}
	Consider a passive \ac{RB} experiment with balanced heterodyne measurement setting, initial Fock state $\rho = \ketbra{\ninput}{\ninput}$,
	and $\lambda$ is an irrep of $\omega_n^\nmodes = \tau_n^\nmodes (\argdot) \tau_n^{\nmodes \ad}$.
	Then, the following holds:
	\begin{equation}
		\begin{aligned}
			\EE[f_{\lambda_k}^2]
			=
			\frac{1}{s_{\lambda_k}^2} (-1)^{\varphi(N_0)}
			\sum_{\substack{\vec n_1, \vec n_2, \vec n_3,\\ \vec m_1, \vec m_2, \vec m_3}}
			(-1)^{\sum_{i=1}^3 \varphi(M_i)}
			I((\vec n_i)_{i=1}^3 (\vec m_i)_{i=1}^3) g_k(\vec N, \vec M, N_0) \, ,
		\end{aligned}
	\end{equation}
	where $\vec N \equiv (N_1, N_2, N_3), \vec M \equiv (M_1, M_2, M_3)$, $N_i \equiv N(\vec n_i), M_i \equiv M(\vec m_i)$ are \ac{GT} patterns associated with $\vec n_i, \vec m_i$, respectively, $I((\vec n_i)_{i=1}^3, (\vec m_i)_{i=1}^3) \equiv I(\vec n_1, \vec n_2, \vec n_3, \vec m_1, \vec m_2, \vec m_3)$ is as in \cref{eq:alpha_integral_solved_multiindex} and
	\begin{equation}
		\begin{aligned}
			g_k(\vec N, \vec M, N_0)
			&=
			\sum_{l=0}^{\min\{n, 2k\}} \frac{1}{d_{\lambda_k}} \sum_{r=1}^{m_l}
			\sum_{\substack{M, M' \in \GT{\lambda_k} \\ L, L' \in \GT{\lambda_k} \\ R, R' \in \GT{\lambda_l}}}
			C_{N_0, \bar N_0}^{M} C_{N_0, \bar N_0}^{M'} C_{M, M'}^{R, r} C_{N_0, \bar N_0}^{R} C_{N_3, \bar M_3}^{R'}
			C_{N_1, \bar M_1}^{L} C_{N_2, \bar M_2}^{L'} C_{L, L'}^{R', r} \, .
		\end{aligned}
	\end{equation}
\end{theorem}
\begin{proof}
	By \cref{eq:multimode_cs,eq:alpha_integral_multi-index}, we have, for any $\lambda_k$,
	\begin{align}
		\EE[f_{\lambda_k}^2]
		&=
		\frac{1}{s_{\lambda_k}^2} \int d^2\vec \alpha \int dg \,
		\sandwich{\ninput}{P_{\lambda_k} \circ \omega_n^\nmodes(g)^\ad (\ketbra{\vec \alpha}{\vec \alpha})}{\ninput}^2
		\sandwich{\vec \alpha}{\omega_n^\nmodes(g) (\ketbra{\ninput}{\ninput})}{\vec \alpha} \\
		&=
		\frac{1}{s_{\lambda_k}^2} \sum_{\substack{\vec n_1, \vec n_2, \vec n_3 \in \H_n^\nmodes \\ \vec m_1, \vec m_2, \vec m_3 \in \H_n^\nmodes}}
		I((\vec n_i)_{i=1}^3, (\vec m_i)_{i=1}^3)
		\int dg\, 
		\sandwich{\ninput}{P_{\lambda_k} \circ \omega_n^\nmodes(g)^\ad (\ketbra{\vec n_1}{\vec m_1})}{\ninput} \\
		&\times
		\sandwich{\ninput}{P_{\lambda_k} \circ \omega_n^\nmodes(g)^\ad (\ketbra{\vec n_2}{\vec m_2})}{\ninput}
		\sandwich{\vec n_3}{\omega_n^\nmodes(g) (\ketbra{\ninput}{\ninput})}{\vec m_3} \\
		&\equiv
		\frac{1}{s_{\lambda_k}^2} \sum_{\substack{\vec n_1, \vec n_2, \vec n_3,\\ \vec m_1, \vec m_2, \vec m_3}}
		I((\vec n_i)_{i=1}^3, (\vec m_i)_{i=1}^3) \, G_k((\vec n_i)_{i=1}^3, (\vec m_i)_{i=1}^3, \ninput) \, ,
	\end{align}
	where
	\begin{equation}
		\begin{aligned}
			G_k((\vec n_i)_{i=1}^3, (\vec m_i)_{i=1}^3, \ninput)
			&\equiv
			\int dg\,
			\sandwich{\ninput}{P_{\lambda_k} \circ \omega_n^\nmodes(g)^\ad(\ketbra{\vec n_1}{\vec m_1})}{\ninput}
			\sandwich{\ninput}{P_{\lambda_k} \circ \omega_n^\nmodes(g)^\ad(\ketbra{\vec n_2}{\vec m_2})}{\ninput} \\
			&\times
			\sandwich{\vec n_3}{\omega_n^\nmodes(g) (\ketbra{\ninput}{\ninput})}{\vec m_3} \, .
		\end{aligned}
	\end{equation}
	Introducing \ac{GT} patterns, we have 
	\begin{equation}
		G_k((\vec n_i)_{i=1}^3, (\vec m_i)_{i=1}^3, \ninput) \equiv G_k(\vec N, \vec M, N_0) \, , \quad
		\vec N = (N_1, N_2, N_3) \, , \ \vec M = (M_1, M_2, M_3) \, .
	\end{equation}
	Then, by \cref{eq:phase_dual_GTpattern}, the latter assumes the following form:
	\begin{align}
		G_k(\vec N, \vec M, N_0)
		&=
		\int dg\, 
		\sandwich{\ninputGT, \ninputGT}{P_{\lambda_k} \tau_n^\nmodes \otimes \bar \tau_n^\nmodes(g)^\ad}{N_1, M_1}
		\sandwich{\ninputGT, \ninputGT}{P_{\lambda_k} \tau_n^\nmodes \otimes \bar \tau_n^\nmodes(g)^\ad}{N_2, M_2} \\
		&\times
		\sandwich{N_3, M_3}{\tau_n^\nmodes \otimes \bar \tau_n^\nmodes(g)}{\ninputGT, \ninputGT} \\
		&=
		(-1)^{\varphi(N_0) + \sum_{i=1}^3 \varphi(M_i)}
		\sum_{M, M' \in \GT{\lambda_k}} C_{N_0, \bar N_0}^{M} C_{N_0, \bar N_0}^{M'} \\
		&\times
		\int dg\,
		\sandwich{M, M'}{\lambda_k(g)^{\otimes 2 \ad}}{N_1, \bar M_1, N_2, \bar M_2}
		\sandwich{N_3, \bar M_3}{\bigoplus_{j=0}^n \lambda_j(g)}{N_0, \bar N_0} \\
		&\equiv
		(-1)^{\varphi(N_0) + \sum_{i=1}^3 \varphi(M_i)} g_k(\vec N, \vec M, N_0) \, .
	\end{align}
	We compute the integral 
	\begin{equation}
		g_k(\vec N, \vec M, N_0)
		\equiv
		\int dg\,
		\sandwich{M, M'}{\lambda_k(g)^{\otimes 2 \ad}}{N_1, \bar M_1, N_2, \bar M_2}
		\sandwich{N_3, \bar M_3}{\bigoplus_{j=0}^n \lambda_j(g)}{N_0, \bar N_0} 
	\end{equation}
	as in the proof of \cref{thm:second_moment_filter_function_pnr}:
	Consider the \ac{irrep} decomposition of $\lambda_k^{\otimes 2}$ as in \cref{lem:lambda_two_copies_decomp}.
	Then, by orthogonality of the matrix coefficients, the non-trivial contributions to the integral come from \acp{irrep} that appear --with their multiplicities-- in both $\omega_n^\nmodes$ and $\lambda_k^{\otimes 2}$.
	This implies
	\begin{equation}
		g_k(\vec N, \vec M, N_0)
		=
		\sum_{l = 0}^{\min\{n, 2k\}} \sum_{r=1}^{m_l}
		\int dg\, \sandwich{M, M'}{\lambda_l^{(r)}(g)^\ad}{N_1, \bar M_1, N_2, \bar M_2} \sandwich{N_3, \bar M_3}{\lambda_l(g)}{N_0, \bar N_0} \, .
	\end{equation}
	Therefore, by the Clebsch-Gordan decompositions
	\begin{align}
		\ket{M, M'}
		&=
		\sum_{R\in \GT{\lambda_l}} C_{M, M'}^{R, r} \ket{R, r} \, , \\
		\ket{N_3, \bar M_3}
		&=
		\sum_{J \in \GT{\lambda_l}} C_{N_3, \bar M_3}^{J} \ket{J} \, , \\
		\ket{N_0, \bar N_0}
		&=
		\sum_{J' \in \GT{\lambda_l}} C_{N_0, \bar N_0}^{J'} \ket{J'} \, , \\
		\ket{N_1, \bar M_1, N_2, \bar M_2}
		&=
		\sum_{L, L' \in \GT{\lambda_k}} C_{N_1, \bar M_1}^{L} C_{N_2, \bar M_2}^{L'} \ket{L, L'} \, , \\
		\ket{L, L'}
		&=
		\sum_{R' \in \GT{\lambda_l}} C_{L, L'}^{R' r} \ket{R', r} \, ,
	\end{align}
	it follows
	\begin{equation}
		\int dg\, \sandwich{R,r}{\lambda_l^{(r)}(g)^\ad}{R',r} \sandwich{J}{\lambda_l(g)}{J'} = \frac{1}{d_{\lambda_l}} \delta_{R,J'} \delta_{R',J}
	\end{equation}
	by Schur's orthogonality relations \eqref{eq:orthogonality_relations}.
	Hence,
	\begin{equation}
		\begin{aligned}
			g_k(\vec N, \vec M, N_0)
			&=
			\sum_{l=0}^{\min\{n, 2k\}} \frac{1}{d_{\lambda_k}} \sum_{r=1}^{m_l}
			\sum_{\substack{M, M' \in \GT{\lambda_k} \\ L, L' \in \GT{\lambda_k} \\ R, R' \in \GT{\lambda_l}}}
			C_{N_0, \bar N_0}^{M} C_{N_0, \bar N_0}^{M'} C_{M, M'}^{R, r} C_{N_0, \bar N_0}^{R} C_{N_3, \bar M_3}^{R'}
			C_{N_1, \bar M_1}^{L} C_{N_2, \bar M_2}^{L'} C_{L, L'}^{R', r} \, ,
		\end{aligned}
	\end{equation}
	and the assertion follows from a suitable sorting of all the terms. 
\end{proof}

\section{Additional details on numerical experiments}
\label{sec:numerics}

Here we provide additional details and plots for the numerical experiment discussed in \cref{fig:decays_plot} and the numerical implementation of the filter function \eqref{eq:filter_function_PNR}.
Computation of the filter function was performed on a MacBook Pro 2020 with Intel Quad-Core i5 (1,4 GHz), 8 GB LPDDR3 (2133 MHz).

The specifics of the simulation are the following:
We consider the collision-free input state $\ket{\ninput} = \ket{\vec 1_4} \equiv \ket{1111} \in \H_n^4$ and filter onto the $\lambda_k$ \ac{irrep} for $k=0, \dots, 4$.
Random unitaries $g_j^{(i)} \in \SU(4), i \in [\numsamples], j \in [\seqlength]$ are drawn from the Haar probability measure according to the procedure described in \cite{mezzadri_how_2006, ozols_how_2009}, the sequence lengths considered consisting of $\seqlength=1, \dots, 10$ Haar random unitaries.
We assume each gate comes with a probability $1-p$ of losing a particle on each mode.
In particular, we consider experiments with gate independent noise, where $\sqrt{p} = 0.95, 0.975, 0.99$ and simulate a gate dependent noise experiment where the transmittivity $\sqrt{p_j}$ of the $j$-th unitary in the sequence is drawn uniformly at random from the interval $[0.9, 1]$.
Concretely, denoting with $\mathcal L(g)$ the single-mode lossy channel, this means that the noisy gate is modeled as $\omega_n^\nmodes(g) \circ \mathcal L(g)^{\otimes \nmodes}$ for each gate $g$ in the sequence.
Lastly, random Fock states are drawn performing a boson sampling simulation using the Python module Piquasso \cite{kaposiPolynomialSpeedupTorontonian2022, kolarovszki2024piquassophotonicquantumcomputer}.
Additionally, we collected noiseless samples for evaluation of $\EE[f_{\lambda_k}^2]_\mathrm{ideal}$ by simply letting $p = 1$.
We collect $\numsamples=10000$ sampled pairs $( (g_j^{(i)})_{j=1}^{\seqlength}, \vec n^{(i)})$ and store them for post-processing.
We remark that the unitaries $g_1^{(i)}, \dots, g_\seqlength^{(i)}$ are used to collect exactly one state from the boson sampler, as throughout this work we consider the so-called \emph{single-shot estimator} \eqref{eq:filtered_data_estimator} \cite{heinrichRandomizedBenchmarkingRandom2023}.

By post-selecting on the outcome of the boson sampling experiment, we capture the estimation of particle loss rates as described in \cref{sec:particle_loss}.
The results, which employ the indicator filter function \eqref{eq:filter_indicator} are shown in \cref{fig:decays_indicator} (the estimator is again the empirical average of the filter function).
Since we run simulations under particle loss only, the estimated decay rates are similar to the ones shown in \cref{fig:decays_plot}. 

Next, we evaluate the filter function \eqref{eq:filter_function_PNR} for each sampled pair.
Specifically, we compute Clebsch-Gordan coefficients using the SUNRepresentations Julia library \cite{SUNrepresentations.jl}, which implements Alex's et al. algorithm for the computation of $\SU(\nmodes)$ Clebsch-Gordan coefficients \cite{alexNumericalAlgorithmExplicit2011}.
The algorithm can be sketched as follows:
For each \ac{irrep} $\lambda_k$ in $\tau_n^\nmodes \otimes \bar \tau_n^\nmodes$ (or --in the case of the second moment-- for each \ac{irrep} in $\lambda_k^{\otimes 2}$) one shall first find the Clebsch-Gordan coefficients of the highest weight state of $\lambda_k$ (and possibly resolve the ambiguity on the multiplicities by a suitable Gaussian elimination, in the case of $\lambda_k^{\otimes 2}$).
Lower weight states are obtained by repeated application of ladder operators.
This implies the calculation of all remaining Clebsch-Gordan coefficients by solving linear systems of equations.
We remark the computation of Clebsch-Gordan coefficients can be sped up by exploiting the symmetries of the weight spaces under the action of the Weyl group \cite{alexSUClebschGordanCoefficients2012} or using analytic expressions for the coefficients $C_{N, \bar N}^{M}$, with $N \in \GT{\tau_n^\nmodes}, M \in \GT{\lambda_k}$ \cite{vilenkinRepresentationLieGroups1992Vol3}.

Lastly, we analyze the signal form for the \acp{irrep} $\lambda_0, \dots, \lambda_4$ (we do not include $\lambda_1$ as there is no overlap with the chosen input state).
The results are shown in \cref{fig:decays_t0.95}.
\begin{figure}
	\centering
	\begin{subfigure}[ht]{.4\linewidth}
		\includegraphics[width=1\linewidth]{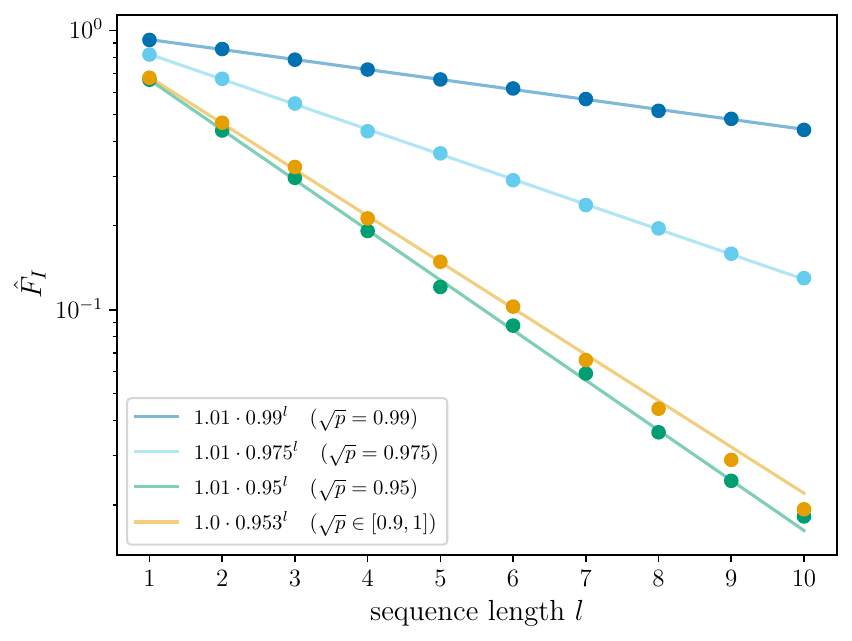}
		\caption{
			Post-selection on particle loss events under different values of transmittivity.
		}
		\label{fig:decays_indicator}	
	\end{subfigure}
	\begin{subfigure}[ht]{.4\linewidth}
		\includegraphics[width=1\linewidth]{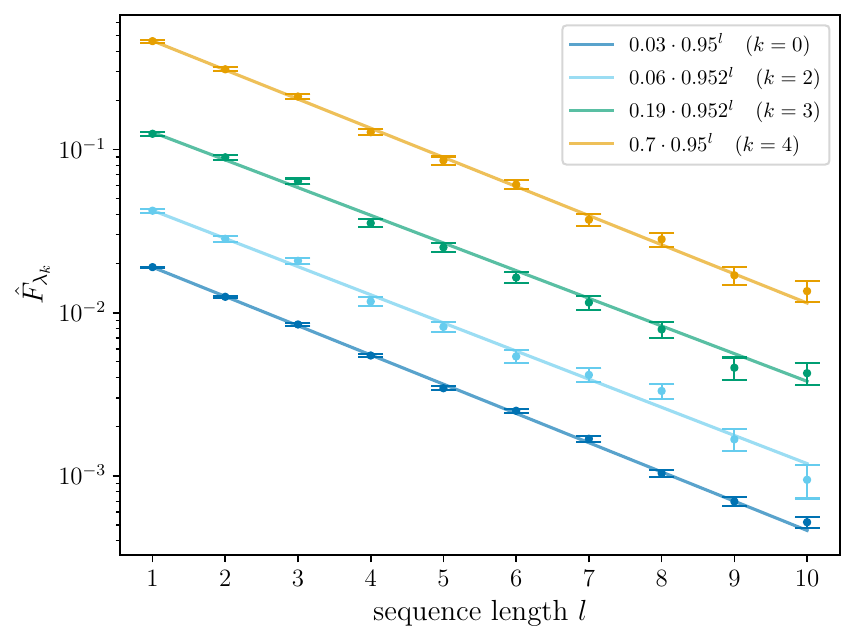}
		\caption{
			Passive \ac{RB} signal in the presence of particles loss with transmittivity $\sqrt{p} = 0.95$.
		}
		\label{fig:decays_t0.95}	
	\end{subfigure}
	\caption{
		Signal forms of passive \ac{RB} with $n = \nmodes = $ and input state $\ket{\vec 1_4}$ using $\numsamples=10000$ samples.
	}
\end{figure}

\begin{acronym}[POVM]\itemsep.5\baselineskip
\acro{ACES}{averaged circuit eigenvalue sampling}
\acro{AGF}{average gate fidelity}

\acro{BS}{Boson Sampling}
\acro{BOG}{binned outcome generation}
\acro{BW}{brickwork}

\acro{CCR}{canonical commutation relation}
\acro{CP}{completely positive}
\acro{CPT}{completely positive and trace preserving}
\acro{CV}{continuous variable}

\acro{DFE}{direct fidelity estimation} 

\acro{FT}{Fourier transform}

\acro{GKP}{Gottesmann-Knill-Preskill}
\acro{GST}{gate set tomography}
\acro{GTM}{gate-independent, time-stationary, Markovian}
\acro{GBS}{Gaussian Boson Sampling}
\acro{GT}{Gelfand--Tsetlin}

\acro{HOG}{heavy outcome generation}

\acro{irrep}{irreducible representation}

\acro{LOP}{linear optical passive}
\acro{LRC}{local random circuit}

\acro{MUBs}{mutually unbiased bases} 
\acro{MW}{micro wave}

\acro{NISQ}{noisy and intermediate scale quantum}

\acro{OVM}{operator-valued measure}

\acro{PNR}{particle number resolving}
\acro{POVM}{positive operator-valued measure}
\acro{PVM}{projector-valued measure}

\acro{QAOA}{quantum approximate optimization algorithm}

\acro{RB}{randomized benchmarking}

\acro{SFE}{shadow fidelity estimation}
\acro{SIC}{symmetric, informationally complete}
\acro{SPAM}{state preparation and measurement}

\acro{QST}{quantum state tomography}
\acro{QPT}{quantum process tomography}

\acro{rf}{radio frequency}

\acro{TT}{tensor train}
\acro{TV}{total variation}

\acro{VQE}{variational quantum eigensolver}

\acro{XEB}{cross-entropy benchmarking}

\end{acronym}

\bibliographystyle{myapsrev4-2}
\bibliography{CVs,Math,RB,mk}

\end{document}